\documentclass[10pt,leqno]{amsart}

\usepackage{hyperref}
\usepackage{tikz}
\usepackage{tikz-3dplot}
\usepackage{mathtools,cleveref}

\newcommand{\Z}{\mathbb{Z}}
\newcommand{\R}{\mathbb{R}}
\newcommand{\conf}{\mathcal{C}}

\usepackage{comment}
\usepackage{amsmath,amssymb}

\newtheorem{theorem}{Theorem}[section]
\newtheorem{lem}[theorem]{Lemma}

\newtheorem{prop}[theorem]{Proposition}

 \newtheorem{remark}[theorem]{Remark}
 \newtheorem{conj}[theorem]{Conjecture}
 \newtheorem{example}[theorem]{Example}
\numberwithin{equation}{section}
\allowdisplaybreaks
\begin{document}
\title[An explicit construction of Kaleidocycles]
 {An explicit construction of Kaleidocycles by elliptic theta functions} 

 \author{Shizuo Kaji}
 \address{Institute of Mathematics for Industry, Kyushu University, Fukuoka, Japan}
 \address{Graduate School of Science, Kyoto University, Kyoto, Japan}
 \email{kaji.shizuo.7r@kyoto-u.ac.jp}

 \author{Kenji Kajiwara}
 \address{Institute of Mathematics for Industry, Kyushu University, Fukuoka, Japan}
 \email{kaji@imi.kyushu-u.ac.jp}

 \author{Shota Shigetomi}
 \address{Institute of Mathematics for Industry, Kyushu University, Fukuoka, Japan}
 \email{s-shigetomi@imi.kyushu-u.ac.jp}

 \thanks{
 The first-named author was 
 partially supported by JSPS KAKENHI Grant Numbers JP25K00921 and JP25H00399.
 The second-named author was partially supported by JSPS KAKENHI Grant Number JP25K21661.
 The third-named author was partially supported by JSPS KAKENHI Grant Numbers JP25K00921 and JP25K17297.
 }


\subjclass[2020]{Primary 53A04, 53A70; Secondary 53A17, 70B15, 37K25, 37K10, 35Q53}
\keywords{Kinematic chain, Configuration space, Discrete differential geometry, Integrable systems, Elliptic theta function}

\begin{abstract}
We study a configuration space consisting of ordered points on the two-dimensional sphere satisfying a system of quadratic constraints.
We construct explicit periodic orbits in the configuration space using elliptic theta functions.
The constructed orbits simultaneously satisfy semi-discrete analogues of the modified KdV and sine-Gordon equations.
This configuration space arises as the state space of a linkage mechanism called a \emph{Kaleidocycle}, and the constructed orbits describe the characteristic motion of the Kaleidocycle. A key consequence of our construction is the proof that Kaleidocycles exist for any number of tetrahedra $k \ge 6$.
Our approach is founded on the relationship between the deformation of spatial curves and integrable systems.
The motion of the mechanism is interpreted as a deformation of a closed discrete spatial curve with constant torsion angle.
This provides an explicit example in which an integrable system is solved to generate periodic orbits in a real solution set of polynomial equations arising from geometric constraints.
\end{abstract}

\maketitle
\section{Introduction}\label{sec:intro}

Given a natural number $k$ and a real number $0 <\lambda<\pi$, 
we consider the following configuration space of $k+1$ ordered points on the 2-sphere $S^2$ satisfying certain quadratic constraints:
\begin{align}\label{eq:configuration_binormal}
\conf^\pm_{k,\lambda} = \{(b_0,\ldots,b_k)\in (S^2)^{k+1} \mid & \langle b_{n-1}, b_n \rangle = \cos\lambda, \, (1\le n\le k), \\
& \sum_{n=0}^{k-1}(b_{n+1}\times b_{n}) = 0, \nonumber\\
& b_k=\pm b_0\}. \nonumber 
\end{align}

This configuration space naturally appears as the state space of a linkage mechanism called a \emph{Kaleidocycle}.
A Kaleidocycle consists of $k$ identical equifacial tetrahedra connected by hinge joints along pairs of opposite edges so that they form a ring (Fig. \ref{fig:K6} (Left)).
The mechanism exhibits a characteristic turning motion.
Although Kaleidocycles have been studied by several authors~\cite{CHEN20052287,FG:mobility,safsten2016analyzing,JF:Kaleidocycle}, the rigorous proof of their existence beyond trivial cases has remained an open problem.

Our main result (\Cref{thm:kaleidocycle}) states that for any $k\ge6$ there exists $\lambda$ for which the configuration space $\conf^\pm_{k,\lambda}$ contains an embedded circle.
We explicitly construct this circle using elliptic theta functions.
Since a point of $\conf^\pm_{k,\lambda}$ corresponds to a particular configuration of the mechanism, a closed path in this space describes a periodic motion of the Kaleidocycle.
Therefore our result provides a constructive proof of the existence of Kaleidocycles consisting of $k$ tetrahedra for all $k\ge6$.


The key observation is that a Kaleidocycle can be identified with a closed discrete spatial curve obtained by connecting the midpoints of the hinge edges.
In this interpretation the hinge edges correspond to the binormal directions of the curve.
The angles between consecutive hinge edges correspond to a torsion angle, while the turning angles of the polygon correspond to signed curvature angles (\Cref{sec:kaleidocycle}).
Hence an element $(b_0,\ldots,b_k)\in \conf^\pm_{k,\lambda}$ can be regarded as a sequence of binormal vectors of a closed polygonal curve with constant torsion angle $\lambda$.
In this geometric formulation the motion of a Kaleidocycle is modelled as a deformation of a closed polygonal curve preserving both the segment length and the torsion angle.
In our construction, the compatibility conditions of the discrete Frenet frame describing such a deformation are compared with the functional identities of elliptic theta functions.

It was known in~\cite{KKP:linkage} that 
the semi-discrete modified KdV (mKdV) and the sine-Gordon equations
are defined on the configuration space (\Cref{subsection:Segment-length and torsion-angle preserving deformation of discrete space curve}).
The main result of this paper can also be viewed as providing a solution to the system through an explicit construction based on the $\tau$ functions of the two-component KP hierarchy~\cite{HIKMO:DLIE} (\Cref{section:Time evolution of curvature and its potential}). 
Moreover, the sweeping surface of the motion 
corresponds to a semi-discrete analogue of the 
fully-discrete constant negative curvature surface
(\emph{K-surface})
studied in \cite{bobenko1996discrete,IKMO:DmKdVsG}.

Combining both views, 
we explicitly integrate the flow generated by an integrable system on a real-algebraic set $\conf^\pm_{k,\lambda}$ using elliptic theta functions.
Our work connects seemingly distant fields of integrable systems and the kinematics of linkages
via an interesting object, the Kaleidocycle.

\medskip
The structure of the paper is as follows.
Section~\ref{sec:def} introduces Kaleidocycles and their description as discrete spatial curves.
The signed curvature angle and torsion angle are defined together with the associated Frenet frame.
Section~\ref{sec:integrable} reviews the relationship between discrete curves and integrable systems through the $\tau$ functions, and recalls previous results describing Kaleidocycle motion by semi-discrete equations.
Section~\ref{section:An explicit formula of closed discrete space curve with constant torsion angle} constructs a set of $\tau$ functions that lead to discrete curves with constant torsion angle using elliptic theta functions.
Section~\ref{section:Time evolution of curvature and its potential} introduces a time parameter and derives the evolution equations satisfied by the curve.
Section~\ref{section:Necessary and sufficient conditions for closedness} proves our main result (\Cref{thm:kaleidocycle}) that for every $k\ge6$ there exist parameters for the $\tau$ functions for which the curve closes, establishing the existence of Kaleidocycles.
Finally, Section~\ref{sec:numerical} presents numerical examples and a conjecture that the constructed solutions correspond to the so-called M\"obius Kaleidocycles.

We briefly explain the strategy of the proof of the main result, \Cref{thm:kaleidocycle}, which establishes the existence of Kaleidocycles for every $k\ge 6$.
The argument consists of three main steps.
First, in \Cref{prop:identification}, we identify Kaleidocycles with closed discrete curves of constant torsion angle.
Next, in \Cref{thm:PhDKal2}, we construct a family of motions of curves having the desired geometric properties.
Finally, in Section~\ref{section:Necessary and sufficient conditions for closedness}, we determine a choice of parameters for which these curves become closed.

\section{Algebraic and geometric models of Kaleidocycle}\label{sec:def}

\subsection{Configuration space of a Kaleidocycle}
\label{sec:kaleidocycle}
A Kaleidocycle consists of (typically six) 
identical copies of equifacial tetrahedra 
connected by their opposing edges as hinge joints to form a ring.
It can be folded from a sheet of paper and 
is an example of rigid origami; it rotates like a bubble ring while keeping all faces rigid.
The origin of Kaleidocycles is obscure, but
one of the first appearances in the literature is in Bricard's paper~\cite{bricard},
where he considers an essentially equivalent closed kinematic chain,
now known as the Bricard 6R.
For nearly a hundred years, the geometry and kinematics of this object have attracted the attention of researchers~(see \cite{safsten2016analyzing} and the references therein).
However, there is no literature that discusses the existence of such objects for arbitrary $k$.

We first define a generalisation of the Kaleidocycle following \cite{seimitsu}.
We consider $k$-copies of an equifacial tetrahedron in $\R^3$ connected by opposing edges to form a ring.
Denote the vertex set by $\{0,1,\ldots,2k-1\}$ such that the vertices $(2n,2n+1)$
form a hinge edge and the vertices $(2n,2n+1,2n+2)$ form one of the triangle faces.
Let $l_{01},l_{12},l_{02}$ be the lengths of the sides of the triangle face, with the edge of length $l_{01}$ serving as a hinge (see \Cref{fig:K6} Right).
Now, a natural way to model a Kaleidocycle is as a \emph{realisation} of an edge-weighted graph.
Let $(V,E)$ be the graph formed by the vertices and edges of the tetrahedra in a Kaleidocycle 
and $\ell:E\to \R$ be the edge length.
A realisation of $(V,E,\ell)$ in $\R^3$
is a map $\phi: V\to \R^3$ satisfying
$|\phi(u)-\phi(v)|^2=\ell(uv)^2$ for all edges $uv\in E$.

In general, the collection of all possible realisations of an edge-weighted finite graph 
$(V,E,\ell)$ in $\R^d$
 forms a real algebraic set in $\mathbb{R}^{d|V|}$, known as the configuration space of $(V,E,\ell)$. 
This space has been the subject of extensive research across various domains, including algebraic geometry, robotics, and mechanism theory~\cite{Hausmann-Knutson,Kapovich-Millson,MR3246296}. 

More concretely, a state of a $k$-Kaleidocycle of type 
$(l_{01},l_{12},l_{02})$ is identified by the real solution
$(x_0,x_1,\ldots,x_{2k+1})\in (\mathbb{R}^3)^{2k+2}$
of the quadratic equations
\begin{align}\label{eq:configuration_coords}
|x_{2n} - x_{2n+1}|^2 &=l^2_{01}, \\
|x_{2n} - x_{2n+2}|^2 &=|x_{2n+1} - x_{2n+3}|^2=l^2_{02}, \nonumber \\
|x_{2n+1} - x_{2n+2}|^2&=|x_{2n} - x_{2n+3}|^2=l^2_{12}\nonumber,
\end{align}
for $0\le n\le k-2$
and $(x_{2k},x_{2k+1})=(x_0,x_1)$ or $(x_1,x_0)$. 
To remove the translational freedom, we demand $x_0+x_1=0$.

Notice that we have two combinatorially distinct cases: 
length $l_{02}$ edges form a $k$-cycle ($x_{2k}=x_0$) or $2k$-cycle, going around the ring twice ($x_{2k}=x_1$).

\begin{remark}\label{rem:dof}
    The solution space of \eqref{eq:configuration_coords} admits a natural action of the orthogonal group $O(3)$ of $\mathbb{R}^3$.
    A naive dimension count then predicts that the quotient space should have dimension
    \[
        3\cdot 2k - 5k - 3 - \dim(O(3)) = k - 6,
    \]
    where $5k$ is the number of edges and the term $3$ comes from the constraint $x_0 + x_1 = 0$.
    This heuristic is commonly referred to as Maxwell's degree-of-freedom counting law~\cite{FG:mobility}.

    It is well known, however, that this estimate can fail in nontrivial ways.
    A classical example is the Bricard 6R linkage ($k=6$), which is equivalent to the classical $6$-Kaleidocycle.
    Its configuration space is one-dimensional, exceeding the expected dimension $0$, and hence it provides a prototypical example of an \emph{overconstrained} mechanism~\cite{FG:mobility}.

    For $k \ge 7$, numerical evidence indicates that, when the angle between adjacent hinge edges attains a critical value,
    the configuration space degenerates to a circle.
    The existence of this single degree of freedom was first identified by the first-named author together with collaborators,
    leading to a subsequent patent application~\cite{KSFG:Kaleidocycle}.
    Shortly thereafter, related results were disseminated in the literature~\cite{JF:Kaleidocycle}.

For $k > 7$, this family of Kaleidocycles--referred to as M\"obius Kaleidocycles~\cite{KSFG:Kaleidocycle}--thus provides a rare example of an \emph{underconstrained} linkage,
    whose configuration space has strictly smaller dimension than that predicted by Maxwell's count.
\end{remark}

We use the following elementary geometric property of 
equifacial tetrahedron:
\begin{lem}\label{lem:equifacial}
An equifacial tetrahedron whose faces are congruent triangles with edge-lengths $(l_{01}, l_{12}, l_{02})$ may be realised within a rectangular cuboid
 of side lengths
\[
\sqrt{\frac{l_{01}^2 + l_{12}^2 - l_{02}^2}{2}}, \quad \sqrt{\frac{l_{12}^2 + l_{02}^2 - l_{01}^2}{2}}, \quad \sqrt{\frac{l_{02}^2 + l_{01}^2 - l_{12}^2}{2}}
\]
so that tetrahedron edges are diagonals of the faces.
Moreover, the twist angle between the pair of opposite edges of length~$l_{01}$ is given by the relation
\[
 \cos(\lambda) = \frac{|l_{12}^2 - l_{02}^2|}{l_{01}^2}.
 \]
\end{lem}

We will see the solution space of \Cref{eq:configuration_coords} is identified with $\conf^\pm_{k,\lambda}$.

\begin{prop}\label{prop:configuration-space}
    The space of real solutions to \Cref{eq:configuration_coords}
    is $O(3)$-equivariantly homeomorphic to $\conf^\pm_{k,\lambda}$ with 
    $\lambda=\arccos\left(\frac{|l_{12}^2 - l_{02}^2|}{l_{01}^2}\right)$.    
\end{prop}
\begin{proof}
Given $(x_0,\ldots,x_{2k+1})$ satisfying \Cref{eq:configuration_coords},
we define $b_n=\dfrac{x_{2n+1}-x_{2n}}{|x_{2n+1}-x_{2n}|}$.
By \Cref{lem:equifacial}, the inner product 
$\langle b_{n-1}, b_{n} \rangle$ is constant
$\cos\lambda$ for $1\le n<k$. 

Let $\gamma_n = (x_{2n}+x_{2n+1})/2 \ (0\le n<k)$ be the mid-point of the hinge edge.
By \Cref{lem:equifacial}, $\gamma_{n+1}-\gamma_{n}$
is orthogonal to both $b_{n}$ and 
$b_{n+1}$.
Therefore, $b_{n+1}\times b_{n}$ is parallel to $\gamma_{n+1}-\gamma_n$, and we have
$l(b_{n+1}\times b_{n})=\gamma_{n+1}-\gamma_{n}$
for some $l\in \R$.
Then, $l\sum_{n=0}^{k-1}( b_{n+1}\times b_{n}) =
\sum_{n=0}^{k-1}(\gamma_{n+1}-\gamma_n)=0$.
This means $(b_0,\ldots,b_k)\in \conf^\pm_{k,\lambda}$.

Conversely, given $(b_0,\ldots,b_k)\in \conf^\pm_{k,\lambda}$,
define $\gamma_n=\frac{l}{\sin\lambda} \sum_{m=0}^{n-1}(b_{m+1}\times b_{m})$,
where $l$ is the distance of the opposing hinge edges of the tetrahedron.
We set $x_{2n}=\gamma_{n} - \frac{l_{01}}{2} b_n$
and  $x_{2n+1}=\gamma_n +\frac{l_{01}}{2} b_n$.
Then, $(x_0,\ldots,x_{2k+1})$ satisfies 
\Cref{eq:configuration_coords}.

It is clear that
the above constructions are mutually inverse and $O(3)$ equivariant.
\end{proof}

While the shape of the tetrahedron is determined by the triangle edge lengths $(l_{01},l_{12},l_{02})$, 
the configuration space depends only on the hinge angle $\lambda$.
We call $\lambda$ the \emph{torsion angle}.

We call the (possible non-flat) polygon, or \emph{closed discrete spatial curve},
defined by $\gamma_n$ the \emph{central curve} of the Kaleidocycle.

Depending on the sign of $b_k=\pm b_0$,
we say the Kaleidocycle is \emph{oriented} (resp. \emph{non-oriented}).

\begin{figure}
    \centering
    \includegraphics[width=0.3\textwidth]{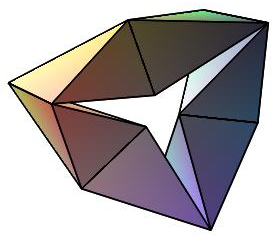}
    \includegraphics[width=0.3\textwidth]{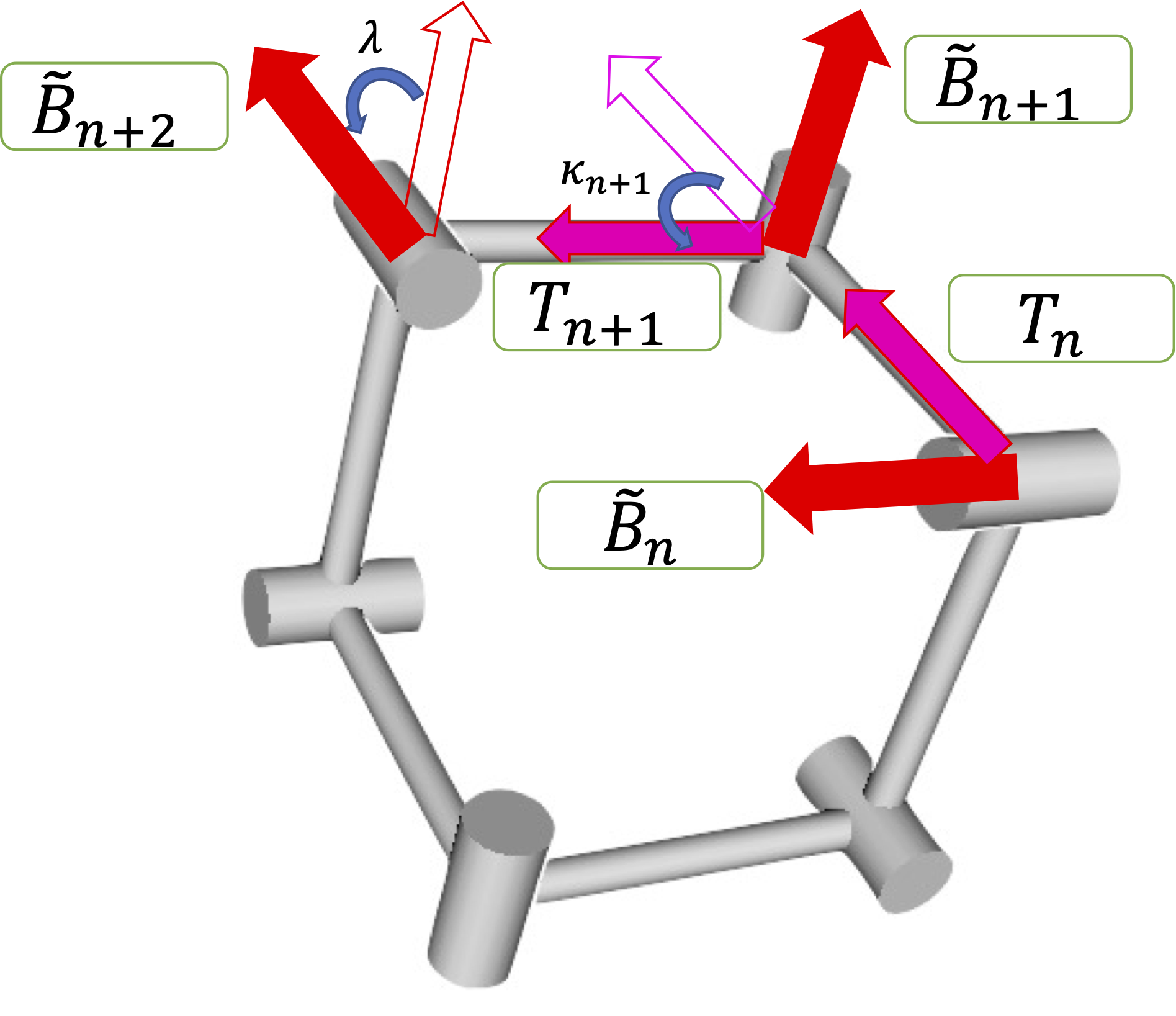}
        \includegraphics[width=0.3\linewidth]{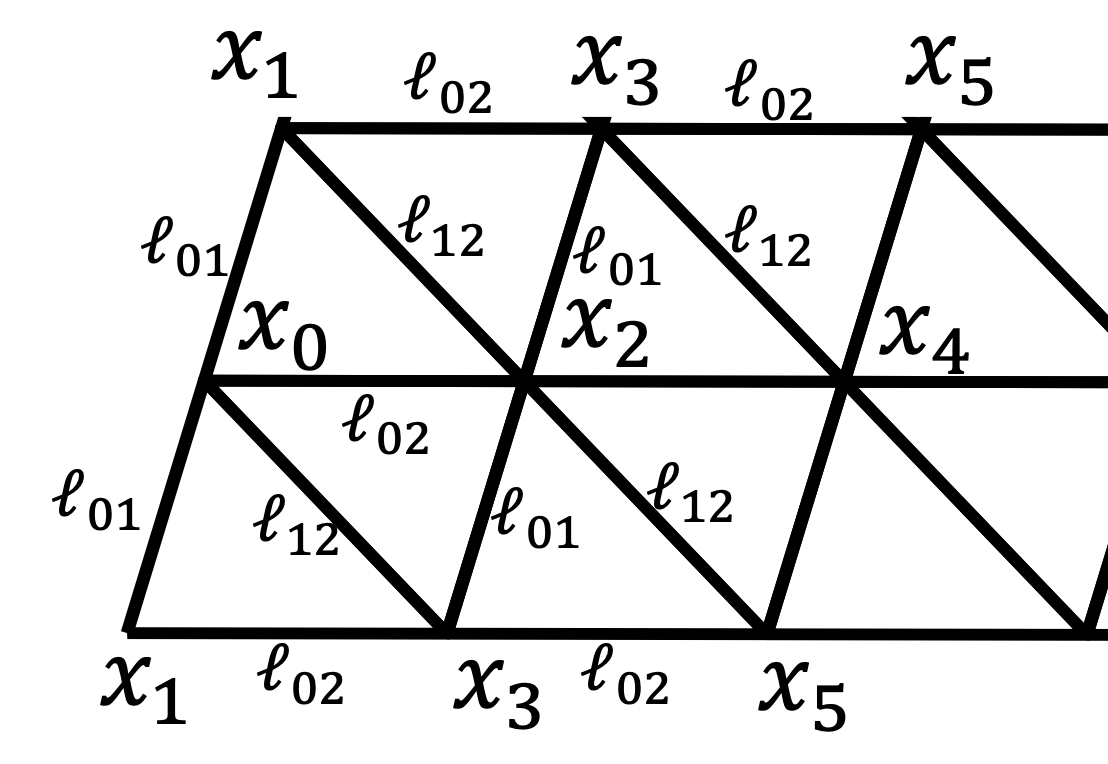}  
        \caption{(Left) The classical 6-Kaleidocycle. (Centre) Kaleidocycle as a discrete curve with its framing (Right) Net of a Kaleidocycle with labelled vertices and edge lengths.}
    \label{fig:K6}
\end{figure}

\subsection{Discrete spatial curve of a constant torsion angle}\label{subsection:Discrete space curves}
We saw that a discrete spatial curve $\gamma_n$ is associated with a Kaleidocycle.
We review some basics on discrete curves.

A discrete spatial curve with non-vanishing torsion angles can be described in terms of binormals.
Let $\widetilde{B}: \mathbb{Z}\rightarrow S^{2}$ 
satisfy
$\left| \widetilde{B}_{n}\times  \widetilde{B}_{n-1}\right|\neq 0$ for all $n\in \mathbb{Z}$.
A discrete spatial curve of unit segment-length 
having $\widetilde{B}$ as 
the binormals can be constructed as follows.
Define the tangent $T: \mathbb{Z}\rightarrow S^{2}$ and the normal $\widetilde{N}: \mathbb{Z}\rightarrow S^{2}$ by
\begin{equation}\label{def:tangent_and_normal}
T_n = \frac{ \widetilde{B}_{n+1}\times  \widetilde{B}_{n}}{\left| \widetilde{B}_{n+1}\times  \widetilde{B}_{n}\right|},\quad
\widetilde{N}_n = \widetilde{B}_n\times T_n.
\end{equation}
Furthermore, given $\gamma_{0}\in \mathbb{R}^3$,
define $\gamma:\mathbb{Z}\to \mathbb{R}^{3}$ by
\begin{equation}\label{def:gamma}
\gamma_{n+1} =\gamma_{n}+T_{n}=\gamma_{n}+\frac{\widetilde{B}_{n+1}\times\widetilde{B}_{n}}{\left|\widetilde{B}_{n+1}\times\widetilde{B}_{n}\right|}.
\end{equation} 
The \emph{signed curvature angle} $\kappa$
and the \emph{torsion angle}
are defined by
\begin{equation}\label{def:kappa}
\begin{split}
\langle T_{n},T_{n-1}\rangle =\cos\kappa_n, \quad\langle \widetilde{N}_{n},T_{n-1}\rangle &=-\sin\kappa_n, \quad-\pi< \kappa_n< \pi, \\
\langle\widetilde{B}_{n},\widetilde{B}_{n-1}\rangle &= \cos\lambda_n, \quad 0< \lambda_n < \pi.
\end{split}
\end{equation}
The \emph{discrete Frenet frame}
$\Phi_n=[T_n,\widetilde{N}_n,\widetilde{B}_n]\in\mathrm{SO}(3)$ satisfies the \emph{discrete Frenet-Serret formula}
\begin{equation}\label{discrete Frenet-Serret formula}
 \Phi_{n+1} = \Phi_n L_n,\quad L_n = M_1(-\lambda_{n+1}) M_3(\kappa_{n+1}),
\end{equation}
where
\begin{equation}
 M_1(\theta) = 
\left[\begin{array}{ccc}
1 & 0          & 0          \\
0 & \cos\theta  & -\sin\theta\\
0 & \sin\theta  &  \cos\theta\\
\end{array}\right],\quad
 M_3(\theta) = 
\left[\begin{array}{ccc}
 \cos\theta  & -\sin\theta & 0\\
 \sin\theta  &  \cos\theta & 0\\
     0      &       0     & 1
\end{array}\right].
\end{equation}

\begin{remark}
Note that our Frenet frame coincides with the usual one only up to sign.
We adopt this convention to avoid discontinuous changes in the framing during the motion of the Kaleidocycle; its motion passes through the state where the curvature vanishes.
\end{remark}

As an immediate corollary of \Cref{prop:configuration-space}, we obtain the following identification.
\begin{prop}\label{prop:identification}
A state of a Kaleidocycle is identified with a closed discrete spatial curve
$(\gamma_k=\gamma_0)$ of constant torsion angle $\lambda$ whose binormal vectors $\widetilde{B}_n$ satisfy
\eqref{eq:configuration_binormal} with $b_n=\widetilde{B}_n$.
\end{prop}

As in the smooth case, a discrete spatial curve of unit segment length is determined by its curvature and torsion angles.
Accordingly, the motion of a discrete curve with constant torsion angle is governed by the family of curvature functions $\kappa_n(t)$.

The fact that a Kaleidocycle is closed imposes interesting topological constraints.
The signed curvature angles $\kappa_n(t)$ are therefore subject to non-trivial restrictions.
In particular, the self-linking number is a half-integer, and hence remains constant during the motion.
By the C\u{a}lug\u{a}reanu--Fuller--White theorem~\cite{banchoff1975behavior}, the writhe of the curve must also remain constant.
Moreover, the self-linking number shows that the configuration space of a Kaleidocycle is not connected in general.

\section{Curves and integrable systems}\label{sec:integrable}
There are deep connections between differential geometry and the theory of integrable systems, where 
the integrable differential or difference equations arise as compatibility conditions of the geometric objects (see, for example, \cite{rogers_schief_2002}). In the context of
the deformation of spatial and plane curves, the Frenet frame and its deformation are described by the system of 
linear partial differential equations. 
The compatibility condition yields the so-called Ablowitz-Kaup-Newell-Segur (AKNS) hierarchy
for the curvature and the torsion, including the nonlinear Schr\"odinger (NLS) equation as a typical example. 
In particular, if the torsion is a constant, the integrable deformation is governed by the mKdV equation and its hierarchy \cite{hasimoto1972soliton,Lamb,rogers_schief_2002}. 
The connections are also investigated in a discrete setting, where curves are discretized as piecewise linear (polygonal) curves. 
For example, continuous deformations of discrete plane and 
spatial curves have been considered in \cite{doliwa1995integrable}, where the deformations are governed by the semi-discrete analogue of 
the mKdV and the NLS equations. Discrete deformations of the discrete plane and spatial curves have been studied in \cite{IKMO:DmKdVsG}.
In principle, for a given solution to the nonlinear differential/difference equation for the curvature and torsion,
one can reconstruct the curve by solving the associated system of linear differential/difference equations for the Frenet frame 
and integrating the tangent vector. 

The theta functions naturally appear in so-called quasi-periodic solutions to the nonlinear integrable equations, and the corresponding solutions to the associated linear problem are nothing but the Baker-Akhiezer functions. Reconstruction of the deformation of the spatial curve with constant torsion described by the theta function solutions to the mKdV equation has been studied in this framework in \cite{calini1998backlund,langer1984knotted}. Also, an explicit formula of the spatial curves in terms of the $\tau$ functions of the two-component KP hierarchy based on Hirota's bilinear formalism has been presented in \cite{HIKMO:DLIE}.

\subsection{Expression of a curve by the \texorpdfstring{$\tau$}{tau} functions}\label{sec:tau-representation}
We introduce a special class of discrete curves $\gamma_n$ and their Frenet frames
$\Phi_{n}=[T_{n}, \widetilde{N}_{n}, \widetilde{B}_{n}]$
indexed by $n\in \Z$ using a modification of the $\tau$ functions presented in \cite{HIKMO:DLIE}.
Let $F_n,\widetilde{G}_n:\mathbb{C}\to \mathbb{R}$ 
and $f_n, g_n, G_n, H_n:\mathbb{C}\to \mathbb{C}$
be functions indexed by $n\in \Z$ 
with a complex variable $z$ that satisfy the following equations
\begin{align}
&f_{n}f_{n}^*+g_{n}g_{n}^*=F_{n+1}F_{n} \label{i1},\\
&\beta_{1} F_{n}f_{n}+G_{n}g_{n}^*=F_{n+1}f_{n-1} \label{i2},\\
&\beta_{1} F_{n}g_{n}-G_{n}f_{n}^*=F_{n+1}g_{n-1} \label{i3},\\
&H_{n+1}F_{n}-H_{n}F_{n+1}=f_{n}^*g_{n} \label{i4},\\
&D_zF_{n+1}\cdot F_{n}:=\frac{dF_{n+1}}{dz}F_{n}-F_{n+1}\frac{dF_{n}}{dz}=g_{n}g_{n}^* \label{i5},\\
&\left|G_{n}\right|^{2}=\widetilde{G}_{n}^{2},\label{i6}
\end{align}
where $\beta_{1}\in\mathbb{R}$ is some constant and $D_z$ is Hirota's bilinear differential operator \cite{hirota_2004}, and the symbol $*$ denotes the complex conjugate.
These functions are called the $\tau$ functions.
Let $\varrho_{n}=\dfrac{1}{2}\dfrac{G_{n}}{\widetilde{G}_{n}}$ and
\begin{align}\label{T}
  T_{n} &=\left. \frac{1}{f_{n}f_{n}^{*}+g_{n}g_{n}^{*}} \left(
    \begin{array}{c}
   \vspace{2mm}
      f_{n}^{*}g_{n}+f_{n}g_{n}^{*} \\
      \vspace{2mm}
      \frac{1}{i}\left(f_{n}^{*}g_{n}-f_{n}g_{n}^{*}\right) \\
      f_{n}f_{n}^{*}-g_{n}g_{n}^{*}
    \end{array}
  \right)\right|_{z=0}, \\
  \widetilde{N}_{n}&=\left.\frac{1}{f_{n}f_{n}^{*}+g_{n}g_{n}^{*}} \left(
    \begin{array}{c}
    \vspace{2mm}
      \left\{(f_{n}^{*})^{2}-(g_{n}^{*})^{2}\right\}\varrho_{n}+ \left\{(f_{n})^{2}-(g_{n})^{2}\right\}\varrho_{n}^{*} \\
      \vspace{2mm}
      \frac{1}{i}\left[\left\{(f_{n}^{*})^{2}+(g_{n}^{*})^{2}\right\}\varrho_{n}-\left\{(f_{n})^{2}+(g_{n})^{2}\right\}\varrho_{n}^{*}\right] \\
      -2\left(f_{n}^{*}g_{n}^{*}\varrho_{n}+f_{n}g_{n}\varrho_{n}^{*}\right)
    \end{array}
  \right)\right|_{z=0}, \nonumber\\
  \widetilde{B}_{n}&=\left. \frac{1}{f_{n}f_{n}^{*}+g_{n}g_{n}^{*}} \left(
    \begin{array}{c}
    \vspace{2mm}
       \frac{1}{i}\left[-\left\{(f_{n}^{*})^{2}-(g_{n}^{*})^{2}\right\}\varrho_{n}+ \left\{(f_{n})^{2}-(g_{n})^{2}\right\}\varrho_{n}^{*}\right] \\
       \vspace{2mm}
      \left\{(f_{n}^{*})^{2}+(g_{n}^{*})^{2}\right\}\varrho_{n}+ \left\{(f_{n})^{2}+(g_{n})^{2}\right\}\varrho_{n}^{*} \\
      \frac{2}{i}\left(f_{n}^{*}g_{n}^{*}\varrho_{n}-f_{n}g_{n}\varrho_{n}^{*}\right)
    \end{array}
  \right)\right|_{z=0}, \nonumber \\
 \gamma_{n}&= \left.\left(
        \begin{array}{c}
        \vspace{2mm}
    \dfrac{H_{n}+H_{n}^*}{F_{n}}\\
    \vspace{2mm}
    \dfrac{1}{i}\dfrac{H_{n}-H_{n}^*}{F_{n}}\\
    n-2\dfrac{\partial\log F_{n}}{\partial z}
        \end{array}
      \right) \right|_{z=0}. \nonumber  
\end{align}

Similarly to \cite[\S 3]{HIKMO:DLIE}, we obtain:
\begin{prop}\label{prop:curve and frame explicit}
$\Phi_{n}=[T_{n}, \widetilde{N}_{n}, \widetilde{B}_{n}]$
defines the discrete Frenet frame for the curve $\gamma$
satisfying the discrete Frenet-Serret formula \eqref{discrete Frenet-Serret formula}.
Its curvature and torsion angles are given by
\begin{align}\label{lam}
\lambda_n&=\left.\arccos\left(\frac{G_{n}^{*}G_{n-1}+G_{n}G_{n-1}^{*}}{2\widetilde{G}_{n}\widetilde{G}_{n-1}}\right)\right|_{z=0}, \\
\kappa_{n}&=\left.2\arctan\left(\frac{\widetilde{G}_{n}}{\beta_{1} F_{n}}\right)\right|_{z=0}, \nonumber
\end{align}
where we take the principal value for $\arccos$ in $[0, \pi]$ and for $\arctan$ in $(-\pi/2, \pi/2)$.
\end{prop}

\subsection{Motion of a Kaleidocycle on the osculating plane}\label{subsection:Segment-length and torsion-angle preserving deformation of discrete space curve}
Let $\gamma$ be a Kaleidocycle, that is, a closed curve with a constant torsion angle and a unit segment length.
A motion of a Kaleidocycle is identified with a time-parametrised family 
$\gamma: \Z \times \R_{\ge 0}\to \R^3$ such that 
$\gamma(t)$ is a Kaleidocycle for all $t\ge 0$.
In \cite{KKP:linkage}, a particular motion of a Kaleidocycle 
satisfying the following conditions is considered:
\begin{equation}\label{deform2}
\left\langle\dot{\gamma}_{n}, \widetilde{B}_{n}\right\rangle=0,\quad\left\langle\dot{\gamma}_{n}, \dot{\gamma}_{n}\right\rangle=\rho^2, 
\quad
|{\gamma}_{n+1}-{\gamma}_{n}|=1 \quad (0\le \forall n \le k),
\end{equation}
where $\rho\neq 0$ is a constant and 
$\dot{\gamma}_{n}=\dfrac{d\gamma_n}{dt}$.
These conditions imply that the velocity of each vertex is restricted to the osculating plane and has a constant magnitude.
It is shown in \cite{KKP:linkage} that the Frenet-Serret formula \eqref{discrete Frenet-Serret formula}, together with 
\eqref{deform2}, requires that there exists a function 
$\Theta_n$, called the \emph{potential function}, with $\kappa_n=\dfrac{\Theta_{n+1}-\Theta_{n-1}}{2}$ that satisfies either 
the \emph{semi-discrete potential mKdV equation}:
\begin{equation}\label{semi-discrete potential mKdV}
\frac{d}{dt}\left(\frac{\Theta_{n+1}+\Theta_{n}}{2}\right)=\left(1+\cos\lambda\right)\rho\sin\left(\frac{\Theta_{n+1}-\Theta_{n}}{2}\right),
\end{equation}
or the \emph{semi-discrete sine-Gordon equation}:
\begin{equation}\label{semi-discrete sG}
\frac{d}{dt}\left(\frac{\Theta_{n+1}-\Theta_{n}}{2}\right)=\left(1-\cos\lambda\right)\rho\sin\left(\frac{\Theta_{n+1}+\Theta_{n}}{2}\right).
\end{equation}
The motions corresponding to these equations differ in general. However, we will construct a motion in \S \ref{section:Time evolution of curvature and its potential} that satisfies both equations simultaneously.

\begin{remark}
As noted in \S \ref{sec:kaleidocycle}, the dimension of the configurations of a Kaleidocycle defined by 
\eqref{eq:configuration_binormal} is generally $k-6$. 
However, for the M\"obius Kaleidocycles, which exhibit critical torsion angles,
the dimension of the configurations is conjectured to be one, regardless of $k$.
An extensive numerical simulation suggests that the motion of 
the M\"obius Kaleidocycle satisfies \eqref{deform2}. 
Therefore, \eqref{deform2} specifies a particular yet essential motion of Kaleidocycles.
\end{remark}

\section{An explicit formula of a discrete curve with a constant torsion angle in terms of elliptic theta functions}\label{section:An explicit formula of closed discrete space curve with constant torsion angle}
Based on the representation of curves 
discussed in \S \ref{sec:tau-representation},
we define certain $\tau$ functions satisfying \eqref{i1}-\eqref{i6} using elliptic theta functions to
construct a family of discrete curves with a constant torsion angle
in \Cref{thm:PhDKal1}.

\subsection{A family of curves with a constant torsion angle}\label{subsection:A family of open curves with constant torsion angle}
Throughout this section, $j$ runs through $\{1,2,3,4\}$.
Fix $w\in\mathbb{C}$ with $\mathrm{Im}(w)>0$.
The elliptic theta functions 
$\vartheta_j: \mathbb{C} \to \mathbb{C}$,
 are defined as follows \cite{Mumford:Tata_I,Khachev_Zabrodin:theta}:
\begin{align}\label{theta1}
\vartheta_1 (u; w)
&= - \sum_{n=-\infty}^\infty e^{\pi i w \left(n + \frac{1}{2}\right)^2 + 2 \pi i \left(n + \frac{1}{2}\right)\left(u + \frac{1}{2}\right)},\\
\vartheta_2 (u; w)
&= \sum_{n=-\infty}^\infty e^{\pi i w \left(n + \frac{1}{2}\right)^2 + 2 \pi i \left(n + \frac{1}{2}\right) u}, \nonumber\\
\vartheta_3 (u; w)
&= \sum_{n=-\infty}^\infty e^{\pi i w n^2 + 2 \pi i n u}, \nonumber\\
\vartheta_4 (u; w)
&= \sum_{n=-\infty}^\infty e^{\pi i w n^2 + 2 \pi i n \left(u + \frac{1}{2}\right)}. \nonumber
\end{align}
We often denote $\vartheta_j(u; w)$ by $\vartheta_j(u)$
when $w$ is fixed.
These functions satisfy various identities, collected in Appendix \ref{section:AppendixA}, that are convenient to construct the $\tau$ functions 
satisfying \eqref{i1}–\eqref{i6}.
We see the torsion angle given by Proposition \ref{prop:curve and frame explicit} is a constant.

The following property is useful for examining whether a $\tau$ function constructed in this paper is real-valued.
\begin{lem}\label{lem:theta_property}
The following holds for any $u\in\mathbb{C}$ when $w=iy$ $(y>0)$.
\begin{equation}\label{eqn:zbar}
\vartheta_j(u)^*=\vartheta_j(u^*).
\end{equation}
\end{lem}
\begin{proof}
Consider the following theta function with characteristics.
\begin{equation}
\vartheta_{a, b}(u):=\sum_{n\in\mathbb{Z}}\exp\left(2\pi i(n+a)(u+b)+\pi i (n+a)^2w\right),\quad a, b\in\left\{0, \frac{1}{2}\right\},
\end{equation}
then,
\begin{equation}
\vartheta_1(u)=-\vartheta_{\frac{1}{2}, \frac{1}{2}}(u),\quad \vartheta_2(u)=\vartheta_{\frac{1}{2}, 0}(u),\quad \vartheta_3(u)=\vartheta_{0, 0}(u),\quad \vartheta_4(u)=\vartheta_{0, \frac{1}{2}}(u),
\end{equation}
hold. Since $w^*=-w$, we have
\begin{equation}
\begin{split}
\vartheta_{a, b}(u)^*=&\sum_{n\in\mathbb{Z}}\exp\left(-2\pi i(n+a)(u^*+b)-\pi i (n+a)^2w^*\right)\\
=&\sum_{n\in\mathbb{Z}}\exp\left(2\pi i(-n-a)(u^*+b)+\pi i (n+a)^2w\right)\\
=&\sum_{n\in\mathbb{Z}}\exp\left(2\pi i(n+a)(u^*+b)+\pi i (n+a)^2w\right)\\
=&\vartheta_{a, b}(u^*).
\end{split}
\end{equation}
\end{proof} 
For $u\in \mathbb{C}$ and $y>0$,
we define 
$d_j(u;y)=\frac{\vartheta^{'}_j(u;iy)}{\vartheta_j(u;iy)}$,
where $^{'}$ denotes the derivative with respect to the first variable of elliptic theta functions.
For $y>0$, $r \in \R \setminus{(1/2\mathbb{Z})}$ and $v \in \R \setminus{(y\mathbb{Z})}$,
we define
\begin{align}
&u_{1}=\vartheta_1\left(-\frac{1}{2}iv+r\right)\vartheta_3\left(-\frac{1}{2}iv+r\right)\vartheta_1\left(\frac{1}{2}iv+r\right)\vartheta_3\left(\frac{1}{2}iv+r\right),\\
&\Delta_{j}=d_{j}\left(-\frac{1}{2}iv+r;y\right)-d_{j}\left(\frac{1}{2}iv+r;y\right), \quad R_{j}=\frac{\vartheta_{j}\left(-\frac{1}{2}iv+r\right)}{\vartheta_{j}\left(\frac{1}{2}iv+r\right)}, \nonumber
\end{align}
where $\vartheta_j(u;iy)$ is denoted by $\vartheta_j(u)$.
Using Lemma\ref{lem:theta_property}, we see
\begin{equation}\label{dr2}
\begin{split} 
u_{1}>0,\quad\Delta_{j}^{*}=-\Delta_{j}, \quad R_{j}^{*}=R_{j}^{-1}.
\end{split}
\end{equation}
Note that from \eqref{dr2}, we see $R_{1}R_{3}^{*}=R_{1}R_{3}^{-1}$ and is a complex number of magnitude $1$. Thus
\begin{equation}\label{dr3}
\begin{split}
\frac{R_{1}R_{3}^{*}+(R_{1}R_{3}^{*})^{*}}{2}=\frac{R_{1}R_{3}^{-1}+R_{1}^{-1}R_{3}}{2},
\end{split}
\end{equation}
is real and lies in $[-1,1]$.
Also, from \eqref{dr2} and
\begin{align}\label{Delta_3andDelta_1}
\Delta_{3}-\Delta_{1}=-\frac{\vartheta_3\left(2r\right)\vartheta_3\left(0\right)\vartheta_1\left(iv\right)\vartheta_1^{'}\left(0\right)}{u_{1}},
\end{align}
we see that $\Delta_{3}-\Delta_{1}\neq 0$.
\begin{prop}\label{thm:PhDKal1}
Put $\mu_{n}=ivn+\frac{z}{\Delta_{3}-\Delta_{1}}$
and define
\begin{align}\label{t1}
&F_{n}(z)=\alpha_1\exp\left(\frac{n\Delta_{3}}{\Delta_{3}-\Delta_{1}}z\right)\vartheta_2\left(\mu_{n}-\frac{1}{2}iv\right),\\
&f_{n}(z)=\alpha_2 R_{3}^{n}\exp\left(\frac{\left(n+\frac{1}{2}\right)\Delta_{3}}{\Delta_{3}-\Delta_{1}}z\right)\vartheta_2\left(\mu_{n}+r\right), \nonumber\\
&g_{n}(z)=\alpha_3 R_{1}^{-n}\exp\left(\frac{\left(n+\frac{1}{2}\right)\Delta_{3}}{\Delta_{3}-\Delta_{1}}z\right)\vartheta_4\left(\mu_{n}-r\right), \nonumber\\
&G_{n}(z)=\alpha_4 R_{1}^{-n}R_{3}^{n}\exp\left(\frac{n\Delta_{3}}{\Delta_{3}-\Delta_{1}}z\right)\vartheta_4\left(\mu_{n}-\frac{1}{2}iv\right), \nonumber\\
&H_{n}(z)=\alpha_5 R_{1}^{-n}R_{3}^{-n}\exp\left(\frac{n\Delta_{3}}{\Delta_{3}-\Delta_{1}}z\right)\vartheta_4\left(\mu_{n}-\frac{1}{2}iv-2r\right), \nonumber \\
&\widetilde{G}_{n}(z)=-i\alpha_4\exp\left(\frac{n\Delta_{3}}{\Delta_{3}-\Delta_{1}}z\right)\vartheta_4\left(\mu_{n}-\frac{1}{2}iv\right), \nonumber
\end{align}
where
\begin{align}
&\alpha_1=\sqrt{\vartheta_3(2r)\vartheta_3(0)},\quad \alpha_2=\vartheta_3\left(-\frac{1}{2}iv+r\right),\quad \alpha_3=\vartheta_1\left(\frac{1}{2}iv+r\right),\nonumber\\
&\alpha_4=\frac{\vartheta_1(iv)\sqrt{\vartheta_3(2r)\vartheta_3(0)}}{\vartheta_3(0)},\quad \alpha_5=\frac{u_{1}}{\vartheta_3(2r)\vartheta_1(iv)\sqrt{\vartheta_3(2r)\vartheta_3(0)}}.\nonumber
\end{align}
Then, they satisfy \eqref{i1}-\eqref{i6} with
\begin{equation}\label{beta_1}
\beta_{1}=\frac{\vartheta_3(iv)}{\vartheta_3(0)},
\end{equation}
and form $\tau$ functions.
The corresponding discrete curve $\gamma$ defined by \eqref{T} has a unit segment length $\left|\gamma_{n+1}-\gamma_{n}\right|=1$.
The torsion angle is given by
\begin{equation}\label{lam1}
\begin{split}
\lambda_n&=\arccos\left(\frac{G_{n}^{*}(0)G_{n-1}(0)+G_{n}(0)G_{n-1}(0)^{*}}{2\widetilde{G}_{n}(0)\widetilde{G}_{n-1}(0)}\right)\\
&=\arccos\left(\frac{R_{1}R_{3}^{-1}+R_{1}^{-1}R_{3}}{2}\right).
\end{split}
\end{equation}
In particular, this does not depend on $n$,
and $\gamma$ has a constant torsion angle
$\lambda_n=\lambda_0$.
Moreover,
the signed curvature angle is given by
\begin{equation}\label{kappa1}
\begin{split}
\kappa_{n}&=2\arctan\left(\frac{\widetilde{G}_{n}(0)}{\beta_{1} F_{n}(0)}\right)\\
&=2\arctan\left(-i\frac{\vartheta_1(iv)\vartheta_4\left(ivn-\frac{1}{2}iv\right)}{\vartheta_3(iv)\vartheta_2\left(ivn-\frac{1}{2}iv\right)}\right).
\end{split}
\end{equation}
\end{prop}
This proposition follows as a corollary of Theorem \ref{thm:PhDKal2} and will be proven in \S\ref{subsection:Explicit formula of the discrete curve with deformation parameter}.
Note that $\widetilde{G}$ and $G$ have the same simple zeros.
Therefore, \eqref{lam1} remains valid even at points where $\widetilde{G}$ vanishes.
Moreover, it follows from the zero distribution of the theta functions that both $\beta_{1}$ and $F_{n}$ are non-vanishing.

Combining this with \eqref{i1}, it follows that \eqref{T} and \eqref{kappa1} hold for any $v$,$r$,$y$, and $n$.
\begin{remark}
Note that $2\tan\frac{\kappa_{n}}{2}$ is often considered as the discretized curvature (see, for example, \cite{hoffmann2009lecture}), which is expressed using Jacobi $\mathrm{cn}$-function as follows:
\begin{equation}\label{kappa222}
\begin{split}
2\tan\frac{\kappa_{n}}{2}&=-2i\frac{\vartheta_1(iv)\vartheta_4\left(ivn-\frac{1}{2}iv\right)}{\vartheta_3(iv)\vartheta_2\left(ivn-\frac{1}{2}iv\right)}\\
&=2\frac{\vartheta_1\left(\frac{v}{y};\frac{i}{y}\right)\vartheta_2\left(\frac{v}{y}n-\frac{1}{2}\frac{v}{y};\frac{i}{y}\right)}{\vartheta_3\left(\frac{v}{y};\frac{i}{y}\right)\vartheta_4\left(\frac{v}{y}n-\frac{1}{2}\frac{v}{y};\frac{i}{y}\right)}\\
&=2\frac{\vartheta_1\left(\frac{v}{y};\frac{i}{y}\right)\vartheta_2\left(0;\frac{i}{y}\right)}{\vartheta_3\left(\frac{v}{y};\frac{i}{y}\right)\vartheta_4\left(0;\frac{i}{y}\right)}\mathrm{cn}\left[2\tilde{K}\left(\frac{v}{y}n-\frac{1}{2}\frac{v}{y}\right),\tilde{\xi}\right],
\end{split}
\end{equation}
where $\tilde{K}=\pi\vartheta_3\left(0;\frac{i}{y}\right)^2/2$ is the complete elliptic integral of the first kind and $\tilde{\xi}=\vartheta_2\left(0;\frac{i}{y}\right)^2/\vartheta_3\left(0;\frac{i}{y}\right)^2$ is the elliptic modulus.
\end{remark}

\section{Deformation of a discrete curve with a constant torsion angle}\label{section:Time evolution of curvature and its potential}
By introducing the deformation parameter $t$ to the curve constructed
in \S\ref{subsection:A family of open curves with constant torsion angle}, we construct an explicit formula for the motion of a Kaleidocycle in Theorem \ref{thm:PhDKal2} (\S\ref{subsection:Explicit formula of the discrete curve with deformation parameter}).
We then derive the evolution equation for the vertex positions of the curve in \Cref{thm:dotgamma} (\S\ref{subsection:Motion of vertex positions}), which reduces to the previously known equations presented in \S\ref{subsection:Segment-length and torsion-angle preserving deformation of discrete space curve} under the specific specialization of the global motion of the curve, as detailed in \Cref{thm:CGamma}.
Furthermore, we derive the evolution equations for the (potential of the) curvature in \Cref{thm:anothermkdv} and \Cref{thm:sdm} (\S\ref{subsection:Evolution of curvature}, \S\ref{subsection:Potential function for curvature}).
Connections to the previously-known isoperimetric deformation of a curve by the semi-discrete mKdV equation \cite{IKMO:DmKdVsG} and by the semi-discrete potential mKdV equation and the semi-discrete sine-Gordon equation \cite{KKP:linkage} are discussed.
This is also summarized in Table \ref{table:deformationequations} in \S\ref{section:Necessary and sufficient conditions for closedness}.

\subsection{Explicit formula for the curve with deformation parameter}\label{subsection:Explicit formula of the discrete curve with deformation parameter}
We introduce the deformation parameter $t$ so that only the signed curvature angles change while the torsion angle is kept constant.
Additionally, the real parameters $C$ and $\Gamma$ are introduced to control the global motion of the curve.
The following theorem describes a Kaleidocycle together with its motion explicitly in terms of elliptic theta functions.
\begin{theorem}\label{thm:PhDKal2}
Denote $\mu_{n}=ivn+\frac{z}{\Delta_{3}-\Delta_{1}}$. Consider the following functions parameterised by $t\in \R$:
\begin{align}\label{t}
&F_{n}(t, z)=\alpha_1\exp\left(\frac{n\Delta_{3}}{\Delta_{3}-\Delta_{1}}z+\frac{C}{2}tz\right)\vartheta_2\left(\mu_{n}-\frac{1}{2}iv+it\right),\\
&f_{n}(t, z)=\alpha_2 R_{3}^{n}\exp\left(\frac{\left(n+\frac{1}{2}\right)\Delta_{3}}{\Delta_{3}-\Delta_{1}}z+\frac{C}{2}tz-\frac{\Gamma}{2} it\right)\vartheta_2\left(\mu_{n}+r+it\right), \nonumber\\
&g_{n}(t, z)=\alpha_3 R_{1}^{-n}\exp\left(\frac{\left(n+\frac{1}{2}\right)\Delta_{3}}{\Delta_{3}-\Delta_{1}}z+\frac{C}{2}tz+\frac{\Gamma}{2} it\right)\vartheta_4\left(\mu_{n}-r+it\right), \nonumber\\
&G_{n}(t, z)=\alpha_4 R_{1}^{-n}R_{3}^{n}\exp\left(\frac{n\Delta_{3}}{\Delta_{3}-\Delta_{1}}z+\frac{C}{2}tz\right)\vartheta_4\left(\mu_{n}-\frac{1}{2}iv+it\right), \nonumber\\
&H_{n}(t, z)=\alpha_5 R_{1}^{-n}R_{3}^{-n}\exp\left(\frac{n\Delta_{3}}{\Delta_{3}-\Delta_{1}}z+\frac{C}{2}tz+\Gamma it\right)\vartheta_4\left(\mu_{n}-\frac{1}{2}iv-2r+it\right), \nonumber\\
&\widetilde{G}_{n}(t, z)=-i\alpha_4 \exp\left(\frac{n\Delta_{3}}{\Delta_{3}-\Delta_{1}}z+\frac{C}{2}tz\right)\vartheta_4\left(\mu_{n}-\frac{1}{2}iv+it\right),\nonumber
\end{align}
where $C, \Gamma\in\mathbb{R}$ are constants.
Then, they satisfy \eqref{i1}-\eqref{i6} with $\beta_{1}$ given by \eqref{beta_1} and form $\tau$ functions.
For each $t$,
the corresponding discrete curve $\gamma$ \eqref{T} has a unit segment length $\left|\gamma_{n+1}-\gamma_{n}\right|=1$.
The torsion angle is given by
\begin{equation}\label{lam2}
\begin{split}
\lambda_n&=\arccos\left(\frac{G_{n}^{*}(t, 0)G_{n-1}(t, 0)+G_{n}(t, 0)G_{n-1}^{*}(t, 0)}{2\widetilde{G}_{n}(t, 0)\widetilde{G}_{n-1}(t, 0)}\right)\\
&=\arccos\left(\frac{R_{1}R_{3}^{-1}+R_{1}^{-1}R_{3}}{2}\right).
\end{split}
\end{equation}
In particular, this depends neither on $n$ nor on $t$,
and $\gamma$ has a constant torsion angle
$\lambda_n(t)=\lambda_0(0)$.
Moreover,
the signed curvature angle is given by
\begin{equation}\label{kappa2}
\begin{split}
\kappa_{n}\left(t\right)&=\kappa_{n}\left(t;v,y\right)=2\arctan\left(\frac{\widetilde{G}_{n}(t, 0)}{\beta_{1} F_{n}(t, 0)}\right)\\
&=2\arctan\left(-i\frac{\vartheta_1(iv)\vartheta_4\left(ivn-\frac{1}{2}iv+it\right)}{\vartheta_3(iv)\vartheta_2\left(ivn-\frac{1}{2}iv+it\right)}\right),
\end{split}
\end{equation}
and has a period in $t$:
\begin{equation}\label{kappa_periodic}
    \kappa_{n}\left(t+y\right)=-\kappa_{n}\left(t\right).
\end{equation}
\end{theorem}
\begin{proof}
Putting
\begin{eqnarray}\label{pr6-1}
\left\{
\begin{array}{l}
\vspace{2mm}
X=ivn+r+it+\frac{z}{\Delta_{3}-\Delta_{1}},\\
\vspace{2mm}
Y=ivn-r+it+\frac{z}{\Delta_{3}-\Delta_{1}},\\
\vspace{2mm}
U=\frac{1}{2}iv+r,\\
V=\frac{1}{2}iv-r,\\
\end{array}
\right.
\end{eqnarray}
in \eqref{pr4-1} yields \eqref{i1}.
Putting
\begin{eqnarray}\label{pr6-2}
\left\{
\begin{array}{l}
\vspace{2mm}
X=iv\left(n-\frac{1}{2}\right)+it+\frac{z}{\Delta_{3}-\Delta_{1}},\\
\vspace{2mm}
Y=ivn+r+it+\frac{z}{\Delta_{3}-\Delta_{1}},\\
\vspace{2mm}
U=iv,\\
V=\frac{1}{2}iv-r,\\
\end{array}
\right.
\end{eqnarray}
in \eqref{pr4-1} yields \eqref{i2}.
Putting
\begin{eqnarray}\label{pr6-3}
\left\{
\begin{array}{l}
\vspace{2mm}
X=iv\left(n-\frac{1}{2}\right)+it+\frac{z}{\Delta_{3}-\Delta_{1}},\\
\vspace{2mm}
Y=\frac{1}{2}iv+r,\\
\vspace{2mm}
U=iv,\\
V=ivn-r+it+\frac{z}{\Delta_{3}-\Delta_{1}},\\
\end{array}
\right.
\end{eqnarray}
in \eqref{pr4-2} yields \eqref{i3}.
Putting
\begin{eqnarray}\label{pr6-4}
\left\{
\begin{array}{l}
\vspace{2mm}
X=ivn-r+it+\frac{z}{\Delta_{3}-\Delta_{1}},\\
\vspace{2mm}
Y=-\frac{1}{2}iv+r,\\
\vspace{2mm}
U=ivn-r+it+\frac{z}{\Delta_{3}-\Delta_{1}},\\
V=-\frac{1}{2}iv+r,\\
\end{array}
\right.
\end{eqnarray}
in \eqref{pr4-3} yields
\begin{equation}\label{pr6-5}
\begin{split}
&\frac{\vartheta_2\left(0\right)\vartheta_4\left(0\right)}{\vartheta_1\left(-\frac{1}{2}iv+r\right)\vartheta_3\left(-\frac{1}{2}iv+r\right)}\vartheta_2\left(\mu_{n}-\frac{1}{2}iv+it\right)\vartheta_4\left(\mu_{n+1}-\frac{1}{2}iv-2r+it\right)\\
&=\frac{\vartheta_2\left(-\frac{1}{2}iv+r\right)\vartheta_4\left(-\frac{1}{2}iv+r\right)}{\vartheta_1\left(-\frac{1}{2}iv+r\right)\vartheta_3\left(-\frac{1}{2}iv+r\right)}\vartheta_2\left(\mu_{n}-r+it\right)\vartheta_4\left(\mu_{n}-r+it\right)\\
&-\vartheta_3\left(\mu_{n}-r+it\right)\vartheta_1\left(\mu_{n}-r+it\right).
\end{split}
\end{equation}
Also, putting
\begin{eqnarray}\label{pr6-6}
\left\{
\begin{array}{l}
\vspace{2mm}
X=ivn-r+it+\frac{z}{\Delta_{3}-\Delta_{1}},\\
\vspace{2mm}
Y=\frac{1}{2}iv+r,\\
\vspace{2mm}
U=ivn-r+it+\frac{z}{\Delta_{3}-\Delta_{1}},\\
V=\frac{1}{2}iv+r,\\
\end{array}
\right.
\end{eqnarray}
in \eqref{pr4-3} yields
\begin{equation}\label{pr6-7}
\begin{split}
&\frac{\vartheta_2\left(0\right)\vartheta_4\left(0\right)}{\vartheta_1\left(\frac{1}{2}iv+r\right)\vartheta_3\left(\frac{1}{2}iv+r\right)}\vartheta_2\left(\mu_{n+1}-\frac{1}{2}iv+it\right)\vartheta_4\left(\mu_{n}-\frac{1}{2}iv-2r+it\right)\\
&=\frac{\vartheta_2\left(\frac{1}{2}iv+r\right)\vartheta_4\left(\frac{1}{2}iv+r\right)}{\vartheta_1\left(\frac{1}{2}iv+r\right)\vartheta_3\left(\frac{1}{2}iv+r\right)}\vartheta_2\left(\mu_{n}-r+it\right)\vartheta_4\left(\mu_{n}-r+it\right)\\
&-\vartheta_3\left(\mu_{n}-r+it\right)\vartheta_1\left(\mu_{n}-r+it\right).
\end{split}
\end{equation}
Subtracting \eqref{pr6-7} from \eqref{pr6-5} and substituting 
\begin{equation}\label{pr6-8}
\begin{split}
&\vartheta_3\left(2r\right)\vartheta_4\left(0\right)\vartheta_2\left(0\right)\vartheta_1\left(iv\right)\\
&=\vartheta_1\left(\frac{1}{2}iv+r\right)\vartheta_3\left(\frac{1}{2}iv+r\right)\vartheta_2\left(-\frac{1}{2}iv+r\right)\vartheta_4\left(-\frac{1}{2}iv+r\right)\\
&-\vartheta_1\left(-\frac{1}{2}iv+r\right)\vartheta_3\left(-\frac{1}{2}iv+r\right)\vartheta_2\left(\frac{1}{2}iv+r\right)\vartheta_4\left(\frac{1}{2}iv+r\right),
\end{split}
\end{equation}
 in the equation, we obtain
\begin{equation}\label{p4-4}
\begin{split}
&\frac{\vartheta_1\left(-\frac{1}{2}iv+r\right)\vartheta_3\left(-\frac{1}{2}iv+r\right)\vartheta_1\left(\frac{1}{2}iv+r\right)\vartheta_3\left(\frac{1}{2}iv+r\right)}{\vartheta_3(2r)\vartheta_1(iv)}\\
&\times\left\{\frac{\vartheta_1\left(\frac{1}{2}iv+r\right)\vartheta_3\left(\frac{1}{2}iv+r\right)}{\vartheta_1\left(-\frac{1}{2}iv+r\right)\vartheta_3\left(-\frac{1}{2}iv+r\right)}\vartheta_2\left(\mu_{n}-\frac{1}{2}iv+it\right)\vartheta_4\left(\mu_{n+1}-\frac{1}{2}iv-2r+it\right)\right.\\
&\left.-\vartheta_2\left(\mu_{n+1}-\frac{1}{2}iv+it\right)\vartheta_4\left(\mu_{n}-\frac{1}{2}iv-2r+it\right)\right\}\\
&=\vartheta_3\left(\frac{1}{2}iv+r\right)\vartheta_2\left(\mu_{n}-r+it\right)\vartheta_1\left(\frac{1}{2}iv+r\right)\vartheta_4\left(\mu_{n}-r+it\right),
\end{split}
\end{equation}
from which we see \eqref{i4} holds.
Putting
\begin{eqnarray}\label{pr6-9}
\left\{
\begin{array}{l}
\vspace{2mm}
X=ivn+r+it+\frac{z}{\Delta_{3}-\Delta_{1}},\\
\vspace{2mm}
Y=ivn-r+it+\frac{z}{\Delta_{3}-\Delta_{1}},\\
\vspace{2mm}
U=\frac{1}{2}iv+r+h,\\
V=\frac{1}{2}iv-r+h,\\
\end{array}
\right.
\end{eqnarray}
in \eqref{pr4-1} yields
\begin{equation}\label{pr6-10}
\begin{split}
&\vartheta_2\left(\mu_{n+1}-\frac{1}{2}iv+it+h\right)\vartheta_2\left(\mu_{n}-\frac{1}{2}iv+it-h\right)\vartheta_3(2r)\vartheta_3(0)\\
&=\vartheta_{2}\left(\mu_{n}+r+it\right)\vartheta_{2}\left(\mu_{n}-r+it\right)\vartheta_{3}\left(\frac{1}{2}iv+r+h\right)\vartheta_{3}\left(\frac{1}{2}iv-r+h\right)\\
&-\vartheta_{4}\left(\mu_{n}+r+it\right)\vartheta_{4}\left(\mu_{n}-r+it\right)\vartheta_{1}\left(\frac{1}{2}iv+r+h\right)\vartheta_{1}\left(\frac{1}{2}iv-r+h\right).
\end{split}
\end{equation}
Differentiating \eqref{pr6-10} with $h$ and then putting $h=0$, we obtain
\begin{equation}\label{pr6-11}
\begin{split}
&\vartheta_3(2r)\vartheta_3(0)\left(\Delta_{3}-\Delta_{1}\right)D_{z}\vartheta_2\left(\mu_{n+1}-\frac{1}{2}iv+it\right)\cdot\vartheta_2\left(\mu_{n}-\frac{1}{2}iv+it\right)\\
&=-\Delta_{3}\vartheta_{3}\left(-\frac{1}{2}iv+r\right)\vartheta_{3}\left(\frac{1}{2}iv+r\right)\vartheta_{2}\left(\mu_{n}+r+it\right)\vartheta_{2}\left(\mu_{n}-r+it\right)\\
&-\Delta_{1}\vartheta_{1}\left(-\frac{1}{2}iv+r\right)\vartheta_{1}\left(\frac{1}{2}iv+r\right)\vartheta_{4}\left(\mu_{n}+r+it\right)\vartheta_{4}\left(\mu_{n}-r+it\right).
\end{split}
\end{equation}
Using the equation \eqref{pr6-10} with $h=0$ and substituting \eqref{pr6-11}, we obtain
\begin{equation}\label{p4-5}
\begin{split}
&\vartheta_3(2r)\vartheta_3(0)\left\{\left(\Delta_{3}-\Delta_{1}\right)D_{z}\vartheta_2\left(\mu_{n+1}-\frac{1}{2}iv+it\right)\cdot\vartheta_2\left(\mu_{n}-\frac{1}{2}iv+it\right)\right.\\
&\left.+\Delta_{3}\vartheta_2\left(\mu_{n+1}-\frac{1}{2}iv+it\right)\vartheta_2\left(\mu_{n}-\frac{1}{2}iv+it\right)\right\}\\
&=\left(\Delta_{3}-\Delta_{1}\right)\vartheta_1\left(-\frac{1}{2}iv+r\right)\vartheta_4\left(\mu_{n}+r+it\right)\vartheta_1\left(\frac{1}{2}iv+r\right)\vartheta_4\left(\mu_{n}-r+it\right).
\end{split}
\end{equation}
From this we see \eqref{i5} holds.
By the definition of $G$ and $\widetilde{G}$, \eqref{i6} follows.
Periodicity of the signed curvature angle \eqref{kappa_periodic} follows from \eqref{cl}.
\Cref{thm:PhDKal1} follows from Theorem \ref{thm:PhDKal2} by setting $t=0$.
\end{proof}
\begin{remark}\label{C_and_Gamma}
The parameter
$C$ contributes only to the parallel translation of the curve along the $z$-axis, whereas $\Gamma$  contributes only to the rotation of the curve around the $z$-axis.
\end{remark}

In Theorem \ref{thm:CGamma}, we will make convenient choices of $C$ and $\Gamma$ to derive the deformation equation considered in the previous studies presented in \S\ref{subsection:Segment-length and torsion-angle preserving deformation of discrete space curve}.

\subsection{Motion of vertex positions}\label{subsection:Motion of vertex positions} 
The construction in Theorem \ref{thm:PhDKal2}
involves parameters $C$ and $\Gamma$ that contribute only to the global rigid transformation of the curve.
In general, the velocity of the vertex position $\dot{\gamma_n}$ has a binormal component, and its magnitude is not constant as described in Theorem \ref{thm:dotgamma}.
However, with the right choice of $C$ and $\Gamma$,
curves that satisfy \eqref{deform2} 
are obtained as described in Theorem \ref{thm:CGamma}.

We first derive equations satisfied by the $\tau$ functions given by Theorem \ref{thm:PhDKal2}
to calculate $\dot{\gamma_n}$.
\begin{lem}[{\cite[\S 3]{HIKMO:DLIE}}]\label{prop:bilinears1} 
Using \eqref{i2} and \eqref{i3}, we see that the following equations hold.
\begin{align}
&f_{n}f_{n-1}^*+g_{n}g_{n-1}^*=f_{n}^*f_{n-1}+g_{n}^*g_{n-1}=\beta_{1} F_{n}F_{n}, \label{2i1}\\
&g_{n}f_{n-1}-f_{n}g_{n-1}=G_{n}F_{n}. \label{2i2}
\end{align}
\end{lem}
\begin{prop}\label{prop:PhDKalbi1}
The $\tau$ functions \eqref{t} satisfy the following equations.
\begin{align}
&f_{n}f_{n-1}^*+g_{n}^{*}g_{n-1}=\beta_{1} F_{n}F_{n}+\beta_{2} G_{n}G_{n}^*,\label{57}\\
&g_{n}^{*}g_{n-1}-g_{n}g_{n-1}^*=f_{n}f_{n-1}^{*}-f_{n}^*f_{n-1}=\beta_{2} G_{n}G_{n}^*,\label{58}\\
&\beta_{4} f_{n}^{*}g_{n-1}-\beta_{4}^{*}f_{n-1}^{*}g_{n}=H_{n}F_{n},\label{59}\\
&\beta_{3} F_{n}F_{n}+\beta_{4}G_{n}G_{n}^{*}=f_{n}f_{n-1}^{*},\label{60}\\
&\left(\beta_{1}-\beta_{3}\right) F_{n}F_{n}+\left(\beta_{2}-\beta_{4}\right)G_{n}G_{n}^{*}=g_{n}^{*}g_{n-1},\label{61}\\
&D_{t}H_{n}\cdot F_{n}=\left(\beta_{5}+i\Gamma\beta_{4}\right) f_{n}^{*}g_{n-1}+\left(\beta_{5}^{*}-i\Gamma\beta_{4}^{*}\right) f_{n-1}^{*}g_{n},\label{62}\\
&D_{t}D_{z}F_{n}\cdot F_{n}=\left(\beta_{6}+C\right)F_{n}F_{n}+\beta_{7}G_{n}G_{n}^{*}.\label{63}
\end{align}
where
\begin{align}\label{50-1}
&\beta_{2}=\frac{\vartheta_{1}\left(2r\right)\vartheta_{3}\left(0\right)}{\vartheta_{3}\left(2r\right)\vartheta_{1}\left(iv\right)},\\
&\beta_{3}=\frac{\vartheta_{3}\left(\frac{1}{2}iv+r\right)^{2}\vartheta_{3}\left(-\frac{1}{2}iv+r\right)^{2}}{\vartheta_{3}\left(2r\right)\vartheta_{3}\left(0\right)^{3}},\quad\beta_{4}=\frac{\vartheta_{1}\left(\frac{1}{2}iv+r\right)^{2}\vartheta_{3}\left(-\frac{1}{2}iv+r\right)^{2}}{\vartheta_{1}\left(iv\right)^{2}\vartheta_{3}\left(2r\right)\vartheta_{3}\left(0\right)},\nonumber\\
&\beta_{5}=\frac{-iu_{1}}{\vartheta_{3}(0)\vartheta_{3}(2r)\vartheta_{1}(iv)^2}\left[\frac{\vartheta_{1}^{'}\left(\frac{1}{2}iv+r\right)\vartheta_{3}\left(-\frac{1}{2}iv+r\right)}{\vartheta_{1}\left(-\frac{1}{2}iv+r\right)\vartheta_{3}\left(\frac{1}{2}iv+r\right)}
+\frac{\vartheta_{1}\left(\frac{1}{2}iv+r\right)\vartheta_{3}^{'}\left(-\frac{1}{2}iv+r\right)}{\vartheta_{1}\left(-\frac{1}{2}iv+r\right)\vartheta_{3}\left(\frac{1}{2}iv+r\right)}\right],\nonumber\\
&\beta_{6}=\frac{2i}{\Delta_{3}-\Delta_{1}}\frac{\vartheta_{3}^{''}(0)}{\vartheta_{3}(0)},\quad\beta_{7}=\frac{2i}{\Delta_{3}-\Delta_{1}}\frac{\vartheta_{1}^{'}(0)^{2}}{\vartheta_{1}(iv)^{2}}.\nonumber
\end{align}
\end{prop}
\begin{proof}
Putting
\begin{eqnarray}\label{prbil57-1}
\left\{
\begin{array}{l}
\vspace{2mm}
X=-\frac{1}{2}iv+r,\\
\vspace{2mm}
Y=\frac{1}{2}iv-r,\\
\vspace{2mm}
U=ivn+r+it+\frac{z}{\Delta_{3}-\Delta_{1}},\\
V=iv\left(n-1\right)-r+it+\frac{z}{\Delta_{3}-\Delta_{1}},\\
\end{array}
\right.
\end{eqnarray}
in \eqref{pr4-3-4} yields \eqref{57}.
Also, \eqref{58} follows from \eqref{2i1} and \eqref{57}.
Putting
\begin{eqnarray}\label{prbil59-1}
\left\{
\begin{array}{l}
\vspace{2mm}
X=ivn-r+it+\frac{z}{\Delta_{3}-\Delta_{1}},\\
\vspace{2mm}
Y=-\frac{1}{2}iv-r+h,\\
\vspace{2mm}
U=\frac{1}{2}iv-r+h,\\
V=iv\left(n-1\right)-r+it+\frac{z}{\Delta_{3}-\Delta_{1}},\\
\end{array}
\right.
\end{eqnarray}
in \eqref{pr4-2} yields
\begin{equation}\label{prbil59-2}
\begin{split}
&\vartheta_4\left(\mu_{n}-\frac{1}{2}iv-2r+it+h\right)\vartheta_3\left(0\right)\vartheta_1\left(iv\right)\vartheta_2\left(\mu_{n}-\frac{1}{2}iv+it-h\right)\\
&=\vartheta_4\left(\mu_{n}-r+it\right)\vartheta_3\left(-\frac{1}{2}iv-r+h\right)\vartheta_1\left(\frac{1}{2}iv-r+h\right)\vartheta_2\left(\mu_{n-1}-r+it\right)\\
&-\vartheta_2\left(\mu_{n}-r+it\right)\vartheta_1\left(-\frac{1}{2}iv-r+h\right)\vartheta_3\left(\frac{1}{2}iv-r+h\right)\vartheta_4\left(\mu_{n-1}-r+it\right),
\end{split}
\end{equation}
After differentiating \eqref{prbil59-2} by $h$ and putting $h=0$, we obtain \eqref{62} by rearrangement.
Also, rearranging the equation with $h=0$ in \eqref{prbil59-2} yields \eqref{59}.
Putting
\begin{eqnarray}\label{prbil60-1}
\left\{
\begin{array}{l}
\vspace{2mm}
X=\frac{1}{2}iv+r,\\
\vspace{2mm}
Y=-\frac{1}{2}iv-r,\\
\vspace{2mm}
U=iv\left(n-\frac{1}{2}\right)+it+\frac{z}{\Delta_{3}-\Delta_{1}},\\
V=-iv\left(n-\frac{1}{2}\right)-it-\frac{z}{\Delta_{3}-\Delta_{1}},\\
\end{array}
\right.
\end{eqnarray}
in \eqref{pr4-3-1} yields \eqref{60}.
Also, \eqref{61} follows from \eqref{57} and \eqref{60}.
Putting
\begin{eqnarray}\label{prbil63-1}
\left\{
\begin{array}{l}
\vspace{2mm}
X=iv\left(n-\frac{1}{2}\right)+it+\frac{z}{\Delta_{3}-\Delta_{1}},\\
\vspace{2mm}
Y=iv\left(n-\frac{1}{2}\right)+it+\frac{z}{\Delta_{3}-\Delta_{1}},\\
\vspace{2mm}
U=h,\\
V=h,\\
\end{array}
\right.
\end{eqnarray}
in \eqref{pr4-1} yields
\begin{equation}\label{prbil63-2}
\begin{split}
&\vartheta_2\left(\mu_{n}-\frac{1}{2}iv+it+h\right)\vartheta_2\left(\mu_{n}-\frac{1}{2}iv+it-h\right)\\
&=-\vartheta_4\left(\mu_{n}-\frac{1}{2}iv+it\right)^{2}\frac{\vartheta_1\left(h\right)^{2}}{\vartheta_3\left(0\right)^{2}}+\vartheta_2\left(\mu_{n}-\frac{1}{2}iv+it\right)^{2}\frac{\vartheta_3\left(h\right)^{2}}{\vartheta_3\left(0\right)^{2}}.
\end{split}
\end{equation}
After differentiating \eqref{prbil63-2} twice with $h$ and putting $h=0$, we obtain \eqref{63}. 
\end{proof}
Coefficients $\beta_{2}$ and $\beta_{4}$ in Proposition \ref{prop:PhDKalbi1} satisfy the following relation.
\begin{lem}\label{prop:56}
The following relation holds.
\begin{equation}\label{56}
\beta_{2}=\beta_{4}-\beta_{4}^{*}.
\end{equation}
\end{lem}

\begin{proof}
Putting
\begin{eqnarray}\label{prbil56-1}
\left\{
\begin{array}{l}
\vspace{2mm}
X=-\frac{1}{2}iv+r,\\
\vspace{2mm}
Y=-\frac{1}{2}iv+r,\\
\vspace{2mm}
U=\frac{1}{2}iv+r,\\
V=\frac{1}{2}iv+r,\\
\end{array}
\right.
\end{eqnarray}
in \eqref{pr4-3-6} yields \eqref{56}. 
\end{proof}

Furthermore, to simplify the calculation of $\dot{\gamma_n}$, we define $P$ and $Q$ as follows
\begin{equation}\label{64}
P_{n}=g_{n}^{*}f_{n-1}-f_{n}g_{n-1}^{*},
\end{equation}
\begin{equation}\label{65}
Q_{n}=g_{n}^{*}f_{n-1}+f_{n}g_{n-1}^{*}.
\end{equation}
Then the following expressions hold
\begin{equation}\label{69}
D_{t}H_{n}\cdot F_{n}=\eta_{1}Q_{n}^{*}+i\eta_{2}P_{n}^{*},
\end{equation}
\begin{equation}\label{N2}
  \tan\frac{\kappa_{n}}{2} \widetilde{N}_{n}=T_{n}-\frac{1}{2\beta_{1} F_{n}^{2}} \left(
    \begin{array}{c}
    \vspace{2mm}
        Q_{n}^{*}+Q_{n}\\
        \vspace{2mm}
       \frac{1}{i}\left(Q_{n}^{*}-Q_{n}\right)\\
       -2\left(\beta_{1}-2\beta_{3}\right)F_{n}^{2}+2\left(\beta_{4}+\beta_{4}^{*}\right) |G_{n}|^{2} 
    \end{array}
  \right),
\end{equation}
\begin{equation}\label{B2}
  \tan\frac{\kappa_{n}}{2}\widetilde{B}_{n}= \frac{1}{2\beta_{1} F_{n}^{2}} \left(
    \begin{array}{c}
    \vspace{2mm}
       \frac{1}{i}\left(P_{n}-P_{n}^{*}\right) \\
       \vspace{2mm}
       P_{n}+P_{n}^{*}\\
       -\frac{2}{i}\beta_{2} |G_{n}|^{2} 
    \end{array}
  \right),
\end{equation}
where
\begin{align}\label{50-2}
\eta_{1}=\frac{1}{2}\left(\beta_{5}+\beta_{5}^{*}\right)+\frac{1}{2}i\Gamma\left(\beta_{4}-\beta_{4}^{*}\right),\quad \eta_{2}=\frac{1}{2}i\left(\beta_{5}-\beta_{5}^{*}\right)-\frac{1}{2}\Gamma\left(\beta_{4}+\beta_{4}^{*}\right).
\end{align}
We have used \eqref{i2} and \eqref{i3} for \eqref{B2}.
Then we have the following lemmas.
\begin{lem}\label{prop:PhDKalbi3}
The following relations hold.
\begin{align}
&2\eta_{1}\left(\beta_{4}+\beta_{4}^{*}\right)+2i\beta_{2}\eta_{2}+\beta_{7}=2\beta_{4}\beta_{5}^{*}+2\beta_{4}^{*}\beta_{5}+\beta_{7}=0,\label{82}\\
&4\eta_{1}^{2}\left(\beta_{4}+\beta_{4}^{*}\right)^{2}+8i\eta_{1}\eta_{2}\beta_{2}\left(\beta_{4}+\beta_{4}^{*}\right)-4\eta_{2}^{2}\beta_{2}^{2}-\beta_{7}^{2}=0.\label{con4}
\end{align}
\end{lem}
\begin{proof}
By the definition of $\eta_{1}$ and $\eta_{2}$ and Lemma \ref{prop:56}, we can see that the leftmost and middle part of \eqref{82} match.
Putting
\begin{eqnarray}\label{prbil82-1}
\left\{
\begin{array}{l}
\vspace{2mm}
X=\frac{1}{2}iv+r+h,\\
\vspace{2mm}
Y=-\frac{1}{2}iv+r-h,\\
\vspace{2mm}
U=\frac{1}{2}iv+r-h,\\
V=-\frac{1}{2}iv+r+h,\\
\end{array}
\right.
\end{eqnarray}
in \eqref{pr4-3-6} yields
\begin{equation}\label{prbil82-2}
\begin{split}
&\vartheta_3\left(\frac{1}{2}iv+r+h\right)\vartheta_3\left(-\frac{1}{2}iv+r-h\right)\vartheta_1\left(\frac{1}{2}iv+r-h\right)\vartheta_1\left(-\frac{1}{2}iv+r+h\right)\\
&-\vartheta_1\left(\frac{1}{2}iv+r+h\right)\vartheta_1\left(-\frac{1}{2}iv+r-h\right)\vartheta_3\left(\frac{1}{2}iv+r-h\right)\vartheta_3\left(-\frac{1}{2}iv+r+h\right)\\
&=\vartheta_3\left(2r\right)\vartheta_3\left(0\right)\vartheta_1\left(iv\right)\vartheta_1\left(2h\right).
\end{split}
\end{equation}
After differentiating \eqref{prbil82-2} by $h$ and setting $h=0$, we obtain \eqref{82}. 
From
\begin{equation}
\begin{split}
&4\eta_{1}^{2}\left(\beta_{4}+\beta_{4}^{*}\right)^{2}+8i\eta_{1}\eta_{2}\beta_{2}\left(\beta_{4}+\beta_{4}^{*}\right)-4\eta_{2}^{2}\beta_{2}^{2}-\beta_{7}^{2}\\
=&\left\{2\eta_{1}\left(\beta_{4}+\beta_{4}^{*}\right)+2i\beta_{2}\eta_{2}+\beta_{7}\right\}\left\{2\eta_{1}\left(\beta_{4}+\beta_{4}^{*}\right)+2i\beta_{2}\eta_{2}-\beta_{7}\right\},
\end{split}
\end{equation}
we obtain \eqref{con4}.
\end{proof}
\begin{lem}\label{prop:PhDKalbi2}
The following relations hold.
\begin{equation}\label{66}
P_{n}P_{n}^{*}=F_{n}^{2}|G_{n}|^{2}+\beta_{2}^{2}|G_{n}|^{4},
\end{equation}
\begin{equation}\label{67}
\begin{split}
P_{n}Q_{n}^{*}-P_{n}^{*}Q_{n}=2\left(2\beta_{3}-\beta_{1}\right)\beta_{2} F_{n}^{2}|G_{n}|^{2}+2\beta_{2}\left(\beta_{4}+\beta_{4}^{*}\right)|G_{n}|^{4},
\end{split}
\end{equation}
\begin{equation}\label{68}
\begin{split}
Q_{n}Q_{n}^{*}=&4\beta_{3}\left(\beta_{1}-\beta_{3}\right)F_{n}^{4}+\left\{1+2\left(\beta_{1}-2\beta_{3}\right)\left(\beta_{4}+\beta_{4}^{*}\right)\right\}F_{n}^{2}|G_{n}|^{2}\\
&-\left(\beta_{4}+\beta_{4}^{*}\right)^{2}|G_{n}|^{4}.
\end{split}
\end{equation}
\end{lem}
\begin{proof}
By definition of $P$, we see that
\begin{equation}\label{prf66}
\begin{split}
P_{n}P_{n}^{*}&=\left(g_{n}^{*}f_{n-1}-f_{n}g_{n-1}^{*}\right)\left(g_{n}f_{n-1}^{*}-f_{n}^{*}g_{n-1}\right)\\
&=\left(f_{n}^{*}g_{n-1}^{*}-g_{n}^{*}f_{n-1}^{*}\right)\left(f_{n}g_{n-1}-g_{n}f_{n-1}\right)\\
&+\left(f_{n}f_{n-1}^{*}-f_{n}^{*}f_{n-1}\right)\left(g_{n}^{*}g_{n-1}-g_{n}g_{n-1}^{*}\right),
\end{split}
\end{equation}
holds.
Substituting \eqref{2i2} and \eqref{58} into \eqref{prf66} yields \eqref{66}.
By definition of $P$ and $Q$, we see that
\begin{equation}\label{prf67-1}
\begin{split}
&P_{n}Q_{n}^{*}-P_{n}^{*}Q_{n}\\
&=\left(g_{n}^{*}f_{n-1}-f_{n}g_{n-1}^{*}\right)\left(g_{n}f_{n-1}^{*}+f_{n}^{*}g_{n-1}\right)\\
&-\left(g_{n}f_{n-1}^{*}-f_{n}^{*}g_{n-1}\right)\left(g_{n}^{*}f_{n-1}+f_{n}g_{n-1}^{*}\right)\\
&=2\left(g_{n}^{*}f_{n-1}f_{n}^{*}g_{n-1}-g_{n}f_{n-1}^{*}f_{n}g_{n-1}^{*}\right),\\
\end{split}
\end{equation}
holds.
On the other hand, from \eqref{2i1} and \eqref{58}, the following holds.
\begin{equation}\label{prf67-2}
\begin{split}
&\beta_{1}\beta_{2} F_{n}^{2}|G_{n}|^{2}\\
&=\left(f_{n}f_{n-1}^*+g_{n}g_{n-1}^*\right)\left(g_{n}^{*}g_{n-1}-g_{n}g_{n-1}^*\right)\\
&=\left(f_{n}^*f_{n-1}+g_{n}^*g_{n-1}\right)\left(g_{n}^{*}g_{n-1}-g_{n}g_{n-1}^*\right)\\
&=\left(f_{n}f_{n-1}^*+g_{n}g_{n-1}^*\right)\left(f_{n}f_{n-1}^{*}-f_{n}^*f_{n-1}\right)\\
&=\left(f_{n}^*f_{n-1}+g_{n}^*g_{n-1}\right)\left(f_{n}f_{n-1}^{*}-f_{n}^*f_{n-1}\right).\\
\end{split}
\end{equation}
Using \eqref{prf67-2}, we obtain
\begin{equation}\label{prf67-3}
\begin{split}
&P_{n}Q_{n}^{*}-P_{n}^{*}Q_{n}\\
&=f_{n}f_{n-1}^*f_{n}f_{n-1}^*-f_{n}^*f_{n-1}f_{n}^*f_{n-1}+g_{n}g_{n-1}^*g_{n}g_{n-1}^*-g_{n}^*g_{n-1}g_{n}^*g_{n-1}\\
&=\left(f_{n}f_{n-1}^*\right)^{2}-\left(f_{n}^{*}f_{n-1}\right)^{2}+\left(g_{n}g_{n-1}^*\right)^{2}-\left(g_{n}^{*}g_{n-1}\right)^{2}\\
&=\left(f_{n}f_{n-1}^{*}-f_{n}^*f_{n-1}\right)\left(f_{n}f_{n-1}^{*}+f_{n}^*f_{n-1}\right)\\
&+\left(g_{n}g_{n-1}^{*}-g_{n}^{*}g_{n-1}\right)\left(g_{n}g_{n-1}^{*}+g_{n}^{*}g_{n-1}\right).
\end{split}
\end{equation}
Substituting \eqref{58}, \eqref{60} and \eqref{61} into \eqref{prf67-3} gives
\begin{equation}\label{prf67-4}
\begin{split}
&P_{n}Q_{n}^{*}-P_{n}^{*}Q_{n}\\
&=\beta_{2} |G_{n}|^{2}\left(f_{n}f_{n-1}^{*}+f_{n}^{*}f_{n-1}-g_{n}g_{n-1}^{*}-g_{n}^{*}g_{n-1}\right)\\
&=2\left(2\beta_{3}-\beta_{1}\right)\beta_{2} F_{n}^{2}|G_{n}|^{2}+2\beta_{2}\left(\beta_{4}+\beta_{4}^{*}\right)|G_{n}|^{4},
\end{split}
\end{equation}
thus \eqref{67}.
Also from \eqref{2i2}, \eqref{60} and \eqref{61}, we see that
\begin{equation}\label{prf68}
\begin{split}
&Q_{n}Q_{n}^{*}-F_{n}^{2}|G_{n}|^{2}\\
&=\left(g_{n}^{*}f_{n-1}+f_{n}g_{n-1}^{*}\right)\left(g_{n}f_{n-1}^{*}+f_{n}^{*}g_{n-1}\right)\\
&-\left(g_{n}f_{n-1}-f_{n}g_{n-1}\right)\left(g_{n}^{*}f_{n-1}^{*}-f_{n}^{*}g_{n-1}^{*}\right)\\
&=\left(f_{n}f_{n-1}^{*}+f_{n}^*f_{n-1}\right)\left(g_{n}g_{n-1}^{*}+g_{n}^{*}g_{n-1}\right)\\
&=\left(2\beta_{3} F_{n}^{2}+\left(\beta_{4}+\beta_{4}^{*}\right)|G_{n}|^{2}\right)\left(2\left(\beta_{1}-\beta_{3}\right)F_{n}^{2}-\left(\beta_{4}+\beta_{4}^{*}\right)|G_{n}|^{2}\right)\\
&=4\beta_{3}\left(\beta_{1}-\beta_{3}\right)F_{n}^{4}+2\left(\beta_{1}-2\beta_{3}\right)\left(\beta_{4}+\beta_{4}^{*}\right)F_{n}^{2}|G_{n}|^{2}-\left(\beta_{4}+\beta_{4}^{*}\right)^{2}|G_{n}|^{4},
\end{split}
\end{equation}
holds. Thus we get \eqref{68}. 
\end{proof}
By using Proposition \ref{prop:PhDKalbi1}, Lemma \ref{prop:56} Lemma \ref{prop:PhDKalbi3} and Lemma \ref{prop:PhDKalbi2}, we obtain the evolution equation of the vertex positions of the curve.
\begin{theorem}\label{thm:dotgamma}
The curve $\gamma$ constructed in Theorem \ref{thm:PhDKal2} satisfies the following deformation equation.
\begin{equation}\label{gdot501}
\begin{split}
\dot{\gamma}_{n}(t)=&2\beta_{1}\left[\eta_{1}T_{n}-\eta_{1}\tan\frac{\kappa_{n}}{2}\widetilde{N}_{n}+\eta_{2}\tan\frac{\kappa_{n}}{2}\widetilde{B}_{n}\right]\\
&+\left(
    \begin{array}{c}
0\\
0\\
2\eta_{1}\left(\beta_{1}-2\beta_{3}\right)-\beta_{6}-C
    \end{array}
  \right).
  \end{split}
\end{equation}
\end{theorem}
\begin{proof}
    By differentiating $\gamma$ with respect to $t$, we get
\begin{equation}\label{gdot201}
\begin{split}
\dot{\gamma}_{n}(t)&=\frac{1}{F_{n}^{2}}\left(
    \begin{array}{c}
    \vspace{2mm}
\displaystyle D_{t}H_{n}\cdot F_{n}+D_{t}H_{n}^{*}\cdot F_{n}\\
\vspace{2mm}
\displaystyle\frac{1}{i}\left(D_{t}H_{n}\cdot F_{n}-D_{t}H_{n}^{*}\cdot F_{n}\right)\\
-D_{t}D_{z}F_{n}\cdot F_{n}
    \end{array}
  \right)\\
  &=\frac{1}{F_{n}^{2}}\left(
    \begin{array}{c}
    \vspace{2mm}
\displaystyle \eta_{1}Q_{n}^{*}+i\eta_{2}P_{n}^{*}+\eta_{1}Q_{n}-i\eta_{2}P_{n}\\
\vspace{2mm}
\displaystyle\frac{1}{i}\left(\eta_{1}Q_{n}^{*}+i\eta_{2}P_{n}^{*}-\eta_{1}Q_{n}+i\eta_{2}P_{n}\right)\\
-\left(\beta_{6}+C\right)F_{n}^{2}-\beta_{7}|G_{n}|^{2}
    \end{array}
  \right).
   \end{split}
\end{equation}
On the other hand, from \eqref{N2} and \eqref{B2}, we see
\begin{equation}\label{gdot301}
\begin{split}
&\eta_{1}\left(T_{n}-\tan\frac{\kappa_{n}}{2}\widetilde{N}_{n}\right)+\eta_{2}\tan\frac{\kappa_{n}}{2}\widetilde{B}_{n}\\
&=\frac{1}{2\beta_{1} F_{n}^{2}}\left(
    \begin{array}{c}
    \vspace{2mm}
\displaystyle \eta_{1}Q_{n}^{*}+i\eta_{2}P_{n}^{*}+\eta_{1}Q_{n}-i\eta_{2}P_{n}\\
\vspace{2mm}
\displaystyle\frac{1}{i}\left(\eta_{1}Q_{n}^{*}+i\eta_{2}P_{n}^{*}-\eta_{1}Q_{n}+i\eta_{2}P_{n}\right)\\
-2\eta_{1}\left(\beta_{1}-2\beta_{3}\right)F_{n}^{2}+2\eta_{1}\left(\beta_{4}+\beta_{4}^{*}\right)|G_{n}|^{2}+2i\eta_{2}\beta_{2}|G_{n}|^{2}
    \end{array}
  \right)
  \end{split}
\end{equation}
holds. From \eqref{82}, \eqref{gdot301} can be rewritten as follows
\begin{equation}\label{gdot401}
\begin{split}
&\eta_{1}\left(T_{n}-\tan\frac{\kappa_{n}}{2}\widetilde{N}_{n}\right)+\eta_{2}\tan\frac{\kappa_{n}}{2}\widetilde{B}_{n}\\
&=\frac{1}{2\beta_{1} F_{n}^{2}}\left(
    \begin{array}{c}
    \vspace{2mm}
\displaystyle \eta_{1}Q_{n}^{*}+i\eta_{2}P_{n}^{*}+\eta_{1}Q_{n}-i\eta_{2}P_{n}\\
\vspace{2mm}
\displaystyle\frac{1}{i}\left(\eta_{1}Q_{n}^{*}+i\eta_{2}P_{n}^{*}-\eta_{1}Q_{n}+i\eta_{2}P_{n}\right)\\
-2\eta_{1}\left(\beta_{1}-2\beta_{3}\right)F_{n}^{2}-\beta_{7}|G_{n}|^{2}
    \end{array}
  \right).
  \end{split}
\end{equation}
Thus \eqref{gdot501} is obtained. 
\end{proof}
This is a deformation equation similar to the one presented in \cite{IKMO:DmKdVsG}. 
Starting from \eqref{gdot501}, the compatibility condition of the Frenet frame yields the \emph{semi-discrete mKdV equation}.
\begin{prop}\label{prop:dotgamma}
When the curve with a unit segment length and a constant torsion angle deforms according to the deformation equation \eqref{gdot501}, the following holds.
\begin{enumerate}
    \item The segment length is preserved.
    \item $\kappa$ satisfies the following semi-discrete mKdV equation: \begin{equation}
\dot{\kappa}_{n}=2\beta_{1}\left(\eta_{1}\cos\lambda+\eta_{2}\sin\lambda\right)\left(\tan\frac{\kappa_{n+1}}{2}-\tan\frac{\kappa_{n-1}}{2}\right).\label{anothermkdv7}
\end{equation}
\end{enumerate}
Furthermore, assume that
\begin{equation}\label{eq:kappa_conserved_quantity}
\tan\frac{\kappa_{n+1}}{2}\tan\frac{\kappa_{n-1}}{2}
-\tan^2\frac{\kappa_n}{2}
+\tan\frac{\kappa_{n+1}}{2}\tan^2\frac{\kappa_n}{2}\tan\frac{\kappa_{n-1}}{2}
\end{equation}
is independent of $n$ (see \Cref{prop:kappa_difference}).
Then,
\begin{enumerate}
\setcounter{enumi}{2}
    \item The torsion angle is preserved.
\end{enumerate}
\end{prop}
\begin{proof}
The proofs of (1) and (2) follow from the same calculation as in Section 5 of \cite{IKMO:DmKdVsG}.
We prove (3). From the deformation equation \eqref{gdot501} and the discrete Frenet-Serret formula \eqref{discrete Frenet-Serret formula}, we see that
\begin{equation}\label{0520-1}
\begin{split}
&-\dot{\lambda}\sin\lambda=\left\langle\dot{\widetilde{B}}_{n}, \widetilde{B}_{n+1}\right\rangle+\left\langle\widetilde{B}_{n}, \dot{\widetilde{B}}_{n+1}\right\rangle\\
&=-\sin\lambda \left(2\beta_{1}\eta_{1}\sin\lambda-2\beta_{1}\eta_{2}\cos\lambda\right)\times\left(\frac{1}{\sin\kappa_{n+1}}\tan\frac{\kappa_{n+2}}{2}\right.\\
&\left.+\frac{1}{\tan\kappa_{n+1}}\tan\frac{\kappa_{n}}{2}-\frac{1}{\tan\kappa_{n}}\tan\frac{\kappa_{n+1}}{2}-\frac{1}{\sin\kappa_{n}}\tan\frac{\kappa_{n-1}}{2}\right),
  \end{split}
\end{equation}
holds.
Put $a_n=\tan\frac{\kappa_n}{2}$.
Using
\[
\sin\kappa_n=\frac{2a_n}{1+a_n^2},
\qquad
\tan\kappa_n=\frac{2a_n}{1-a_n^2},
\]
the bracket in \eqref{0520-1} becomes
\[
\frac12\Bigl(
a_{n+2}a_n-a_{n+1}^2+a_{n+2}a_{n+1}^2a_n
-
a_{n+1}a_{n-1}+a_n^2-a_{n+1}a_n^2a_{n-1}
\Bigr).
\]
Hence it is equal to one half of the difference between
\eqref{eq:kappa_conserved_quantity} at indices $n+1$ and $n$.
Therefore it vanishes if the quantity in \eqref{eq:kappa_conserved_quantity} is independent of $n$, implying $\dot{\lambda}=0$.

\end{proof}
\begin{remark}
Proposition \ref{prop:dotgamma} \textnormal{(3)} does not hold in general without the assumption on \eqref{eq:kappa_conserved_quantity}.
Indeed, there is a motion of a Kaleidocycle 
satisfying \eqref{gdot501}
such that 
\eqref{eq:kappa_conserved_quantity} depends both on $n$
and $t$.
\end{remark}
By adjusting the global motion of $\gamma$ by $C$ and $\Gamma$, the deformation equation of $\gamma$ can be made to satisfy the equation considered in the previous studies presented in \S\ref{subsection:Segment-length and torsion-angle preserving deformation of discrete space curve}.
\begin{theorem}\label{thm:CGamma}
Set $C$ and $\Gamma$ to be
\begin{equation}\label{CGamma}
\begin{split}
\Gamma_{+}\left(r,y\right)=&\frac{\vartheta_{1}^{'}(2r)\vartheta_{3}(0)+ \vartheta_{3}\left(2r\right)\vartheta_{1}^{'}(0)}{\vartheta_{1}\left(2r\right)\vartheta_{3}\left(0\right)},\\
C_{+}\left(v,r,y\right)=&-i\Gamma_{+}\left(\beta_{2}\left(2\beta_{3}-\beta_{1}\right)-\frac{\beta_{4}+\beta_{4}^{*}}{\beta_{2}}\right)\\
&-\beta_{6}-\left(\beta_{5}+\beta_{5}^{*}\right)\left(2\beta_{3}-\beta_{1}\right)+\frac{\beta_{5}-\beta_{5}^{*}}{\beta_{2}},\\
\rm or\\
\Gamma_{-}\left(r,y\right)=&\frac{\vartheta_{1}^{'}(2r)\vartheta_{3}(0)- \vartheta_{3}\left(2r\right)\vartheta_{1}^{'}(0)}{\vartheta_{1}\left(2r\right)\vartheta_{3}\left(0\right)},\\
C_{-}\left(v,r,y\right)=&-i\Gamma_{-}\left(\beta_{2}\left(2\beta_{3}-\beta_{1}\right)-\frac{\beta_{4}+\beta_{4}^{*}}{\beta_{2}}\right)\\
&-\beta_{6}-\left(\beta_{5}+\beta_{5}^{*}\right)\left(2\beta_{3}-\beta_{1}\right)+\frac{\beta_{5}-\beta_{5}^{*}}{\beta_{2}}.
\end{split}
\end{equation}
Then, the curve constructed in Theorem \ref{thm:PhDKal2}
deforms according to the equation \eqref{deform2}.
The magnitude of $\dot{\gamma}_{n}$ is a constant and is given by
\begin{equation}\label{rhopp}
\rho^{2}=\rho_+^{2}:=\left(i\frac{\vartheta_{1}^{'}(0)}{\vartheta_{1}(iv)}\frac{u_{1}}{\vartheta_{3}(r)^{2}\vartheta_{1}(r)^{2}}\right)^{2},
\end{equation}
when $(C,\Gamma)=(C_+,\Gamma_+)$
and
\begin{equation}\label{rhopm}
\rho^{2}=\rho_-^{2}:=\left(-i\frac{\vartheta_{1}^{'}(0)}{\vartheta_{1}(iv)}\frac{u_{1}}{\vartheta_{4}(r)^{2}\vartheta_{2}(r)^{2}}\right)^{2},
\end{equation}
when $(C,\Gamma)=(C_-,\Gamma_-)$.
\end{theorem}
\begin{proof}
    We determine the necessary and sufficient condition for $C$ and $\Gamma$ to satisfy
\begin{equation}\label{mod3}
\left\langle\dot{\gamma}_{n}, \widetilde{B}_{n}\right\rangle=0,
\end{equation}
\begin{equation}\label{mod2}
\left\langle\dot{\gamma}_{n}, \dot{\gamma}_{n}\right\rangle=\rho^{2},
\end{equation}
where $\rho\in\mathbb{R}$ is a constant.
From \eqref{B2} and \eqref{gdot201} we see
\begin{equation}\label{bgdot2}
\begin{split}
2iF_{n}^{3}\widetilde{G}_{n}\left\langle\dot{\gamma}_{n}, \widetilde{B}_{n}\right\rangle=&\left(P_{n}-P_{n}^{*}\right)\left(\eta_{1}Q_{n}^{*}+i\eta_{2}P_{n}^{*}+\eta_{1}Q_{n}-i\eta_{2}P_{n}\right)\\
&+\left(P_{n}+P_{n}^{*}\right)\left(\eta_{1}Q_{n}^{*}+i\eta_{2}P_{n}^{*}-\eta_{1}Q_{n}+i\eta_{2}P_{n}\right)\\
&+2\beta_{2} |G_{n}|^{2}\left\{\left(\beta_{6}+C\right)F_{n}^{2}+\beta_{7}|G_{n}|^{2}\right\}\\
=&2F_{n}^{2}|G_{n}|^{2}\left\{2\eta_{1}\left(2\beta_{3}-\beta_{1}\right)\beta_{2}+2i\eta_{2}+\beta_{2}\left(\beta_{6}+C\right)\right\}\\
&+2\beta_{2}|G_{n}|^{4}\left\{2\eta_{1}\left(\beta_{4}+\beta_{4}^{*}\right)+2i\beta_{2}\eta_{2}+\beta_{7}\right\}\\
=&2F_{n}^{2}|G_{n}|^{2}\left\{2\eta_{1}\left(2\beta_{3}-\beta_{1}\right)\beta_{2}+2i\eta_{2}+\beta_{2}\left(\beta_{6}+C\right)\right\},
\end{split}
\end{equation}
where we have used \eqref{82}.
Therefore, we obtain the following as a necessary and sufficient condition for \eqref{mod3} to hold.
\begin{equation}\label{con1}
2\eta_{1}\left(2\beta_{3}-\beta_{1}\right)\beta_{2}+2i\eta_{2}+\beta_{2}\left(\beta_{6}+C\right)=0,\\
\end{equation}
Using \eqref{50-2}, we rewrite this in terms of $C$ and $\Gamma$ to obtain
\begin{equation}\label{con11}
i\Gamma\left(\beta_{2}\left(2\beta_{3}-\beta_{1}\right)-\frac{\beta_{4}+\beta_{4}^{*}}{\beta_{2}}\right)+\beta_{6}+C+\left(\beta_{5}+\beta_{5}^{*}\right)\left(2\beta_{3}-\beta_{1}\right)-\frac{\beta_{5}-\beta_{5}^{*}}{\beta_{2}}=0.
\end{equation}
Next is the calculation of \eqref{mod2}. The following holds
\begin{equation}\label{eqn:rhopm}
\begin{split}
&F_{n}^{4}\langle\dot{\gamma}_{n}, \dot{\gamma}_{n}\rangle\\
=&\left(\eta_{1}Q_{n}^{*}+i\eta_{2}P_{n}^{*}+\eta_{1}Q_{n}-i\eta_{2}P_{n}\right)^{2}\\
&-\left(\eta_{1}Q_{n}^{*}+i\eta_{2}P_{n}^{*}-\eta_{1}Q_{n}+i\eta_{2}P_{n}\right)^{2}\\
&+\left\{\left(\beta_{6}+C\right)F_{n}^{2}+\beta_{7}|G_{n}|^{2}\right\}^{2}\\
=&F_{n}^{4}\left\{16\eta_{1}^{2}\beta_{3}\left(\beta_{1}-\beta_{3}\right)+\left(\beta_{6}+C\right)^{2}\right\}\\
&+F_{n}^{2}|G_{n}|^{2}\left\{4\eta_{1}^{2}\left(1+2\left(\beta_{1}-2\beta_{3}\right)\left(\beta_{4}+\beta_{4}^{*}\right)\right)-8i\eta_{1}\eta_{2}\beta_{2}\left(2\beta_{3}-\beta_{1}\right)+4\eta_{2}^{2}+2\beta_{7}\left(\beta_{6}+C\right)\right\}\\
&+|G_{n}|^{4}\left(-4\eta_{1}^{2}\left(\beta_{4}+\beta_{4}^{*}\right)^{2}-8i\eta_{1}\eta_{2}\beta_{2}\left(\beta_{4}+\beta_{4}^{*}\right)+4\eta_{2}^{2}\beta_{2}^{2}+\beta_{7}^{2}\right)\\
=&F_{n}^{4}\left\{16\eta_{1}^{2}\beta_{3}\left(\beta_{1}-\beta_{3}\right)+\left(\beta_{6}+C\right)^{2}\right\}\\
&+F_{n}^{2}|G_{n}|^{2}\left\{4\eta_{1}^{2}\left(1+2\left(\beta_{1}-2\beta_{3}\right)\left(\beta_{4}+\beta_{4}^{*}\right)\right)-8i\eta_{1}\eta_{2}\beta_{2}\left(2\beta_{3}-\beta_{1}\right)+4\eta_{2}^{2}+2\beta_{7}\left(\beta_{6}+C\right)\right\},
\end{split}
\end{equation}
where we have used \eqref{con4}.
Therefore, we obtain the following as a necessary and sufficient condition for \eqref{mod2} to hold.
\begin{equation}\label{con3}
4\eta_{1}^{2}\left(1+2\left(\beta_{1}-2\beta_{3}\right)\left(\beta_{4}+\beta_{4}^{*}\right)\right)-8i\eta_{1}\eta_{2}\beta_{2}\left(2\beta_{3}-\beta_{1}\right)+4\eta_{2}^{2}+2\beta_{7}\left(\beta_{6}+C\right)=0.
\end{equation}
By using \eqref{82} and \eqref{con11}, \eqref{con3} can be written in terms of $\Gamma$ as follows.
\begin{equation}\label{con31}
i\beta_{4}\beta_{4}^{*}\left(\beta_{4}-\beta_{4}^{*}\right)\Gamma^{2}+2\beta_{4}\beta_{4}^{*}\left(\beta_{5}+\beta_{5}^{*}\right)\Gamma+i\left(\beta_{4}{\beta_{5}^{*}}^{2}-\beta_{4}^{*}\beta_{5}^{2}\right)=0.
\end{equation}
This is a quadratic equation in $\Gamma$ and has the following solutions.
\begin{equation}\label{solGam}
\Gamma_{\pm}=\frac{-\beta_{4}\beta_{4}^{*}\left(\beta_{5}+\beta_{5}^{*}\right)\pm\sqrt{\beta_{4}\beta_{4}^{*}\left(\beta_{4}\beta_{5}^{*}+\beta_{4}^{*}\beta_{5}\right)^{2}}}{i\beta_{4}\beta_{4}^{*}\left(\beta_{4}-\beta_{4}^{*}\right)}.
\end{equation}
Since
\begin{equation}
\begin{split}
&\beta_{4}\beta_{4}^{*}\left(\beta_{4}\beta_{5}^{*}+\beta_{4}^{*}\beta_{5}\right)^{2}=\left(\beta_{4}\beta_{4}^{*}\right)^{2}\left(i\frac{\vartheta_{1}^{'}(0)}{\vartheta_{1}(iv)}\right)^{2},\\
&\beta_{4}-\beta_{4}^{*}=\frac{\vartheta_{1}(2r)\vartheta_{3}(0)}{\vartheta_{3}(2r)\vartheta_{1}(iv)},\\
&\beta_{5}+\beta_{5}^{*}=-i\frac{\vartheta_{1}^{'}(2r)\vartheta_{3}(0)}{\vartheta_{3}(2r)\vartheta_{1}(iv)},
\end{split}
\end{equation}
holds, $\Gamma_{\pm}$ can be rewritten as
\begin{equation}\label{solGam2}
\Gamma_{\pm}=\frac{\vartheta_{1}^{'}(2r)\vartheta_{3}(0)\pm\vartheta_{3}\left(2r\right)\vartheta_{1}^{'}(0)}{\vartheta_{1}\left(2r\right)\vartheta_{3}\left(0\right)}.
\end{equation}
Substituting this into \eqref{con11} yields the expression for $C_{\pm}$ as well.
From the above calculations, we see that when $C_{\pm}$ and $\Gamma_{\pm}$ are chosen as in \eqref{CGamma}, the following holds
\begin{equation}
    \langle\dot{\gamma}_{n}, \dot{\gamma}_{n}\rangle=16\eta_{1}^{2}\beta_{3}\left(\beta_{1}-\beta_{3}\right)+\left(\beta_{6}+C\right)^{2}.
\end{equation}
We calculate this to compute $\rho^2$.
From \eqref{50-2}, we see that when $(C,\Gamma)=(C_+,\Gamma_+)$, the following holds
\begin{equation}\label{prf:prop12-1}
16\eta_{1}^{2}\beta_{3}\left(\beta_{1}-\beta_{3}\right)=4\left(i\frac{\vartheta_{1}^{'}(0)}{\vartheta_{1}(iv)}\right)^{2}\frac{u_{1}^{2}}{\vartheta_{3}(2r)^{2}\vartheta_{3}(0)^{6}}.
\end{equation}
From \eqref{con1} and \eqref{82}, we see that when $(C,\Gamma)=(C_+,\Gamma_+)$, the following holds
\begin{equation}\label{prf:prop12-2}
\left(\beta_{6}+C\right)^{2}=\left(i\frac{\vartheta_{1}^{'}(0)}{\vartheta_{1}(iv)}\right)^{2}\frac{u_{1}^{2}}{\vartheta_{3}(2r)^{2}\vartheta_{3}(0)^{6}}\frac{\left(\vartheta_{3}(r)^{4}-\vartheta_{1}(r)^{4}\right)^{2}}{\vartheta_{3}(r)^{4}\vartheta_{1}(r)^{4}}.
\end{equation}
Adding the above equations yields
\begin{equation}\label{prf:prop12-3}
\rho^{2}=16\eta_{1}^{2}\beta_{3}\left(\beta_{1}-\beta_{3}\right)+\left(\beta_{6}+C\right)^{2}=\left(i\frac{\vartheta_{1}^{'}(0)}{\vartheta_{1}(iv)}\right)^{2}\frac{u_{1}^{2}}{\vartheta_{3}(r)^{4}\vartheta_{1}(r)^{4}}=\rho_{+}^{2},
\end{equation}
where we have used
\begin{equation}
    \vartheta_{3}(r)^{4}+\vartheta_{1}(r)^{4}=\vartheta_{3}(2r)\vartheta_{3}(0)^{3}.
\end{equation}
In the same way, we see that when $(C,\Gamma)=(C_-,\Gamma_-)$,
\begin{equation}\label{prf:prop12-6}
\rho^{2}=16\eta_{1}^{2}\left(\beta_{1}-\beta_{3}\right)\beta_{3}+\left(\beta_{6}+C\right)^{2}=\left(i\frac{\vartheta_{1}^{'}(0)}{\vartheta_{1}(iv)}\right)^{2}\frac{u_{1}^{2}}{\vartheta_{4}(r)^{4}\vartheta_{2}(r)^{4}}=\rho_{-}^{2},
\end{equation}
holds.
\end{proof}
\begin{remark}\label{rem:k-surface}
A discrete analogue of 
the constant negative curvature surface,
called the discrete K-surface, is introduced in \cite{bobenko1996discrete}.
It is discussed that certain discrete-time deformation of discrete curves with a constant torsion angle generates discrete K-surfaces in \cite{IKMO:DmKdVsG}.
Deformation equation \eqref{deform2} can be seen as a continuous-time analogue of the discrete-time deformation.
Theorem \ref{thm:CGamma} 
shows that specialising $(C,\Gamma)$ to $(C_-,\Gamma_-)$ or $(C_+,\Gamma_+)$, 
the curve constructed in Theorem \ref{thm:PhDKal2} deforms according to \eqref{deform2}.
Therefore, its sweeping surface forms a semi-discrete analogue of the discrete K-surface (see figures in \S \ref{sec:numerical}).
\end{remark}

\subsection{Evolution of curvature}\label{subsection:Evolution of curvature}
The curve constructed in Theorem \ref{thm:PhDKal2} reveals an interesting recurrence equation for $\kappa$. To obtain the recurrence equation, we use the following Lemma.
\begin{lem}\label{prop:kappa_difference1}
    The $\tau$ functions \eqref{t} satisfy the following equations.
\begin{align}
&\widetilde{G}_{n}^{2}+\beta_{1}^2 F_{n}^{2}=F_{n+1}F_{n-1},\label{2i3}\\
&\beta_{1}^2\widetilde{G}_{n}^{2}-\frac{\vartheta_1\left(iv\right)^4}{\vartheta_3\left(iv\right)^4}\beta_{1}^4 F_{n}^{2}=\widetilde{G}_{n+1}\widetilde{G}_{n-1}.\label{eqn:kappa_difference1}
\end{align}
\end{lem}

\begin{proof}
Using \eqref{i1}-\eqref{i3}, we see that \eqref{2i3} hold. Putting
\begin{eqnarray}\label{pranothermkdv4}
\left\{
\begin{array}{l}
\vspace{2mm}
X=0,\\
\vspace{2mm}
Y=0,\\
\vspace{2mm}
U=iv\left(n+\frac{1}{2}\right)+it+\frac{z}{\Delta_{3}-\Delta_{1}},\\
V=-iv\left(n-\frac{3}{2}\right)-it-\frac{z}{\Delta_{3}-\Delta_{1}},\\
\end{array}
\right.
\end{eqnarray}
in \eqref{pr4-3-7}, we obtain \eqref{eqn:kappa_difference1}. 
\end{proof}
\begin{prop}\label{prop:kappa_difference}
The signed curvature angle of the curve obtained in Theorem \ref{thm:PhDKal2}
satisfies the following recurrence equation for any $t$ and $n$:
\begin{equation}
\tan\frac{\kappa_{n+1}}{2}\tan\frac{\kappa_{n-1}}{2}-\tan^2\frac{\kappa_{n}}{2}+\tan\frac{\kappa_{n+1}}{2}\tan^2\frac{\kappa_{n}}{2}\tan\frac{\kappa_{n-1}}{2}=-\frac{\vartheta_1\left(iv\right)^4}{\vartheta_3\left(iv\right)^4}.\label{eqn:kappa_difference}
\end{equation}
The right-hand side of this equation depends neither on $n$ nor on $t$. Therefore, 
the quantity on the left-hand side can be considered a conserved quantity of the signed curvature angle.
\end{prop}
\begin{proof}
    Using \eqref{2i3} and \eqref{eqn:kappa_difference1}, we see that \eqref{eqn:kappa_difference} holds.
\end{proof}
The evolution of the signed curvature angle satisfies the semi-discrete mKdV equation in Proposition \ref{prop:dotgamma}.
To obtain the semi-discrete mKdV equation, we use the following Lemma.
\begin{lem}\label{prop:anothermkdv1}
The $\tau$ functions \eqref{t} satisfy the following equation.
\begin{equation}\label{anothermkdv3}
D_{t}\widetilde{G}_{n}\cdot F_{n}=\beta_{8}\left(\widetilde{G}_{n+1} F_{n-1}-\widetilde{G}_{n-1}F_{n+1}\right),
\end{equation}
where
\begin{equation}\label{anothermkdv5}
\beta_{8}=i\frac{\vartheta_{1}\left(iv\right)\vartheta_{3}\left(iv\right)}{\vartheta_{1}\left(2iv\right)\vartheta_{3}\left(0\right)}\left(d_{1}(iv;y)-d_{3}(iv;y)\right).
\end{equation}
\end{lem}
\begin{proof}
After differentiating the equation \eqref{pr4-2} with
\begin{eqnarray}\label{prtheorem5-1}
\left\{
\begin{array}{l}
\vspace{2mm}
X=iv\left(n+\frac{1}{2}\right)+it+\frac{z}{\Delta_{3}-\Delta_{1}},\\
\vspace{2mm}
Y=iv+h,\\
\vspace{2mm}
U=-iv+h,\\
V=iv\left(n-\frac{3}{2}\right)+it+\frac{z}{\Delta_{3}-\Delta_{1}},\\
\end{array}
\right.
\end{eqnarray}
by $h$ and putting $h=0$, we obtain
\begin{equation}\label{prtheorem5-2}
\begin{split}
&D_{t} \left\{-i\vartheta_1(iv)\vartheta_4\left(\mu_{n}-\frac{1}{2}iv+it\right)\right\} \cdot \left\{ \vartheta_3(iv)\vartheta_2\left(\mu_{n}-\frac{1}{2}iv+it\right)\right\} \\
&=\frac{\vartheta_1(iv)\vartheta_3(iv)}{\vartheta_1(2iv)\vartheta_3(0)} \frac{\vartheta^{'}_{1}\left(iv\right)\vartheta_{3}\left(iv\right)-\vartheta_{1}\left(iv\right)\vartheta^{'}_{3}\left(iv\right)}{i\vartheta_{1}\left(iv\right)\vartheta_{3}\left(iv\right)}\\
&\times \left\{-i\vartheta_1(iv)\vartheta_4\left(\mu_{n-1}-\frac{1}{2}iv+it\right)\right\} \left\{ \vartheta_3(iv)\vartheta_2\left(\mu_{n+1}-\frac{1}{2}iv+it\right)\right\}\\
&-\frac{\vartheta_1(iv)\vartheta_3(iv)}{\vartheta_1(2iv)\vartheta_3(0)}\frac{\vartheta^{'}_{1}\left(iv\right)\vartheta_{3}\left(iv\right)-\vartheta_{1}\left(iv\right)\vartheta^{'}_{3}\left(iv\right)}{i\vartheta_{1}\left(iv\right)\vartheta_{3}\left(iv\right)}\\
&\times \left\{-i\vartheta_1(iv)\vartheta_4\left(\mu_{n+1}-\frac{1}{2}iv+it\right)\right\} \left\{ \vartheta_3(iv)\vartheta_2\left(\mu_{n-1}-\frac{1}{2}iv+it\right)\right\}.
\end{split}
\end{equation}
After rearranging \eqref{prtheorem5-2}, we obtain \eqref{anothermkdv3}.
\end{proof}
\begin{theorem}\label{thm:anothermkdv}
The signed curvature angle \eqref{kappa2} satisfies
\begin{equation}\label{anothermkdv}
\dot\kappa_{n}=2\beta_{1}\left(\eta_{1}\cos\lambda+\eta_{2}\sin\lambda\right)\left(\tan\frac{\kappa_{n+1}}{2}-\tan\frac{\kappa_{n-1}}{2}\right).
\end{equation}
\end{theorem}
\begin{proof}
Differentiating the signed curvature angle yields the following.
\begin{equation}
    \dot{\kappa}_{n}=2\frac{d}{dt}\arctan\left(\frac{\widetilde{G}_{n}}{\beta_{1} F_{n}}\right)=2\beta_{1}\frac{D_{t}\widetilde{G}_{n}\cdot F_{n}}{\widetilde{G}_{n}^{2}+\beta_{1}^2 F_{n}^{2}}.
\end{equation}
    Using \eqref{2i3} and \eqref{anothermkdv3}, we see that the signed curvature angle satisfies
\begin{equation}
\dot\kappa_{n}=2\beta_{1}^2\beta_{8}\left(\tan\frac{\kappa_{n+1}}{2}-\tan\frac{\kappa_{n-1}}{2}\right).
\end{equation}
Since
\begin{equation}
\beta_{1}\beta_{8}=\eta_{1}\cos\lambda+\eta_{2}\sin\lambda,
\end{equation}
holds, this is nothing but the semi-discrete mKdV equation \eqref{anothermkdv}.
\end{proof}

\subsection{Potential function of curvature}\label{subsection:Potential function for curvature}
A family of functions $\Theta_n$ parameterised by $n\in \mathbb{Z}$ is said to be a potential of the curvature $\kappa$ if the equation
$2\kappa_{n}=\Theta_{n+1}-\Theta_{n-1}$ holds.
When $\gamma$ deforms by \eqref{deform2}, the potential function $\Theta$ is known to satisfy the semi-discrete potential mKdV equation or the semi-discrete sine-Gordon equation \cite{KKP:linkage}. 
Here, we derive the evolution equations satisfied by the potential function for the curve constructed by Theorem \ref{thm:PhDKal2}. 

We construct a potential function
which simultaneously satisfy
\eqref{semi-discrete potential mKdV}
and \eqref{semi-discrete sG}
in terms of elliptic theta functions.

\begin{theorem}\label{thm:kappa_and_potential}
For fixed $v$ and $y$, define the function $\Theta$ as follows.
\begin{equation}\label{sdm-2-2}
\Theta_{n}\left(t\right)=\Theta_{n}\left(t;v,y\right)=4\arctan\frac{\vartheta_1\left(\frac{v}{y}\left(n-\frac{1}{2}\right)+\frac{t}{y};\frac{2i}{y}\right)}{\vartheta_4\left(\frac{v}{y}\left(n-\frac{1}{2}\right)+\frac{t}{y};\frac{2i}{y}\right)}.
\end{equation}
Then we see that $\kappa$ defined in \eqref{kappa2} satisfies
\begin{equation}\label{kappa-3}
\kappa_{n}=\frac{\Theta_{n+1}-\Theta_{n-1}}{2}.
\end{equation}
\end{theorem}
\begin{proof}
    From \eqref{kappa-4}, \eqref{kappa-5} and \eqref{kappa-6}, we see that
\begin{equation}\label{kappa-7}
\begin{split}
&\kappa_{n}=2\arctan\left(-i\frac{\vartheta_1(iv; iy)\vartheta_4\left(iv\left(n-\frac{1}{2}\right)+it; iy\right)}{\vartheta_3(iv; iy)\vartheta_2\left(iv\left(n-\frac{1}{2}\right)+it; iy\right)}\right)\\
&=2\arctan\left(-i\frac{\frac{\vartheta_{1}\left(i\frac{v}{2}\left(n+\frac{1}{2}\right)+\frac{it}{2}; \frac{iy}{2}\right)}{\vartheta_{2}\left(i\frac{v}{2}\left(n+\frac{1}{2}\right)+\frac{it}{2}; \frac{iy}{2}\right)}-\frac{\vartheta_{1}\left(i\frac{v}{2}\left(n-\frac{3}{2}\right)+\frac{it}{2}; \frac{iy}{2}\right)}{\vartheta_{2}\left(i\frac{v}{2}\left(n-\frac{3}{2}\right)+\frac{it}{2}; \frac{iy}{2}\right)}}{1-\frac{\vartheta_{1}\left(i\frac{v}{2}\left(n+\frac{1}{2}\right)+\frac{it}{2}; \frac{iy}{2}\right)\vartheta_{1}\left(i\frac{v}{2}\left(n-\frac{3}{2}\right)+\frac{it}{2}; \frac{iy}{2}\right)}{\vartheta_{2}\left(i\frac{v}{2}\left(n+\frac{1}{2}\right)+\frac{it}{2}; \frac{iy}{2}\right)\vartheta_{2}\left(i\frac{v}{2}\left(n-\frac{3}{2}\right)+\frac{it}{2}; \frac{iy}{2}\right)}}\right)\\
&=2\arctan\left(\frac{\frac{\vartheta_{1}\left(\frac{v}{y}\left(n+\frac{1}{2}\right)+\frac{t}{y}; \frac{2i}{y}\right)}{\vartheta_{4}\left(\frac{v}{y}\left(n+\frac{1}{2}\right)+\frac{t}{y}; \frac{2i}{y}\right)}-\frac{\vartheta_{1}\left(\frac{v}{y}\left(n-\frac{3}{2}\right)+\frac{t}{y}; \frac{2i}{y}\right)}{\vartheta_{4}\left(\frac{v}{y}\left(n-\frac{3}{2}\right)+\frac{t}{y}; \frac{2i}{y}\right)}}{1+\frac{\vartheta_{1}\left(\frac{v}{y}\left(n+\frac{1}{2}\right)+\frac{t}{y}; \frac{2i}{y}\right)\vartheta_{1}\left(\frac{v}{y}\left(n-\frac{3}{2}\right)+\frac{t}{y}; \frac{2i}{y}\right)}{\vartheta_{4}\left(\frac{v}{y}\left(n+\frac{1}{2}\right)+\frac{t}{y}; \frac{2i}{y}\right)\vartheta_{4}\left(\frac{v}{y}\left(n-\frac{3}{2}\right)+\frac{t}{y}; \frac{2i}{y}\right)}}\right)\\
&=2\arctan\left(\frac{\vartheta_{1}\left(\frac{v}{y}\left(n+\frac{1}{2}\right)+\frac{t}{y}; \frac{2i}{y}\right)}{\vartheta_{4}\left(\frac{v}{y}\left(n+\frac{1}{2}\right)+\frac{t}{y}; \frac{2i}{y}\right)}\right)\\
&-2\arctan\left(\frac{\vartheta_{1}\left(\frac{v}{y}\left(n-\frac{3}{2}\right)+\frac{t}{y}; \frac{2i}{y}\right)}{\vartheta_{4}\left(\frac{v}{y}\left(n-\frac{3}{2}\right)+\frac{t}{y}; \frac{2i}{y}\right)}\right)\\
&=\frac{\Theta_{n+1}-\Theta_{n-1}}{2},
\end{split}
\end{equation}
holds. 
\end{proof}
The potential function simultaneously solves the semi-discrete potential mKdV equation and the semi-discrete sine-Gordon equation, establishing an explicit relation between our construction and the previous work \cite{KKP:linkage} reviewed in \S \ref{subsection:Segment-length and torsion-angle preserving deformation of discrete space curve}.
Since the semi-discrete sine-Gordon equation can be solved by almost the same calculation, we will solve only the semi-discrete potential mKdV equation here.
We construct a solution to the equation using the theory of $\tau$ functions \cite{inoguchi2012motion}.
For this purpose, we first define the following real-valued functions for fixed $v$ and $y$.
\begin{equation}\label{sdm-6}
\begin{split}
&A_{n}\left(t\right)=\vartheta_4\left(\frac{v}{y}\left(n-\frac{1}{2}\right)+\frac{t}{y};\frac{2i}{y}\right),\\
&B_{n}\left(t\right)=\vartheta_1\left(\frac{v}{y}\left(n-\frac{1}{2}\right)+\frac{t}{y};\frac{2i}{y}\right).
\end{split}
\end{equation}
Then the potential function can be written as follows
\begin{equation}\label{sdm-2}
\Theta_{n}\left(t\right)=4\arctan\frac{B_{n}}{A_{n}}=\frac{2}{i}\log\frac{A_{n}+iB_{n}}{A_{n}-iB_{n}}=\frac{2}{i}\log\frac{\widetilde{\tau}_n}{\widetilde{\tau}^{*}_{n}},
\end{equation}
where $\widetilde{\tau}_{n}=A_{n}+iB_{n}$.
$\widetilde{\tau}$ corresponds to the $\tau$ functions of the one-component KP hierarchy, while $A$ and $B$ correspond to the $\tau$ functions of the two-component KP hierarchy.
Then we have the following lemma.
\begin{lem}\label{prop:sdm2}
The functions $A$ and $B$ defined by \eqref{sdm-6} satisfy the following equations.
\begin{equation}\label{sdm-3}
D_{t}A_{n+1}\cdot A_{n}+D_{t}B_{n+1}\cdot B_{n}=\frac{L_{1}}{2y}\left(\beta_{9}-1\right)\left(A_{n+1}A_{n}+B_{n+1}B_{n}\right),
\end{equation}
\begin{equation}\label{sdm-4}
D_{t}A_{n+1}\cdot B_{n}-D_{t}B_{n+1}\cdot A_{n}=\frac{L_{1}}{2y}\left(\beta_{9}+1\right)\left(A_{n+1}B_{n}-B_{n+1}A_{n}\right),
\end{equation}
where
\begin{align}
&J_{1}(v,y)=\frac{\vartheta^{'}_1\left(0;\frac{2i}{y}\right)\vartheta_1\left(\frac{v}{y};\frac{2i}{y}\right)}{\vartheta_4\left(0;\frac{2i}{y}\right)\vartheta_4\left(\frac{v}{y};\frac{2i}{y}\right)},\quad J_{2}(v,y)=\frac{\vartheta^{'}_1\left(0;\frac{2i}{y}\right)\vartheta_4\left(\frac{v}{y};\frac{2i}{y}\right)}{\vartheta_1\left(\frac{v}{y};\frac{2i}{y}\right)\vartheta_4\left(0;\frac{2i}{y}\right)},\\
&L_{1}(v,y)=d_{1}\left(\frac{v}{y};\frac{2}{y}\right)-d_{4}\left(\frac{v}{y};\frac{2}{y}\right)+J_{1}(v,y)+J_{2}(v,y), \label{sdm-5}\\
&\beta_{9}=\frac{d_{1}\left(\frac{v}{y};\frac{2}{y}\right)+d_{4}\left(\frac{v}{y};\frac{2}{y}\right)-J_{1}(v,y)+J_{2}(v,y)}{L_{1}(v,y)}.
\end{align}
\end{lem}
\begin{proof}
In this proof, all elliptic theta functions have $w=2i/y$ and we denote $\vartheta_j(v; 2i/y)$ by $\vartheta_j(v)$.
After differentiating the equation in \eqref{pr4-3-2} with
\begin{eqnarray}\label{prtheorem2-1}
\left\{
\begin{array}{l}
\vspace{2mm}
X=\frac{v}{y}\left(n+\frac{1}{2}\right)+\frac{t}{y},\\
\vspace{2mm}
Y=\frac{v}{y}\left(n-\frac{1}{2}\right)+\frac{t}{y},\\
\vspace{2mm}
U=\frac{v}{y}+h,\\
V=h,\\
\end{array}
\right.
\end{eqnarray}
by $h$ and putting $h=0$, we obtain
\begin{equation}\label{prtheorem2-2}
\begin{split}
&D_{t} \vartheta_4\left(\frac{v}{y}\left(n+\frac{1}{2}\right)+\frac{t}{y}\right) \cdot \vartheta_4\left(\frac{v}{y}\left(n-\frac{1}{2}\right)+\frac{t}{y}\right) \\
&=\frac{1}{y}\frac{\vartheta^{'}_{4}\left(\frac{v}{y}\right)}{\vartheta_{4}\left(\frac{v}{y}\right)}\vartheta_4\left(\frac{v}{y}\left(n+\frac{1}{2}\right)+\frac{t}{y}\right)\vartheta_4\left(\frac{v}{y}\left(n-\frac{1}{2}\right)+\frac{t}{y}\right)\\
&-\frac{1}{y}\frac{\vartheta_{1}\left(\frac{v}{y}\right)\vartheta^{'}_{1}\left(0\right)}{\vartheta_{4}\left(\frac{v}{y}\right)\vartheta_{4}\left(0\right)}\vartheta_1\left(\frac{v}{y}\left(n+\frac{1}{2}\right)+\frac{t}{y}\right)\vartheta_1\left(\frac{v}{y}\left(n-\frac{1}{2}\right)+\frac{t}{y}\right)\\
&=\frac{1}{y}d_{4}\left(\frac{v}{y};\frac{2}{y}\right)\vartheta_4\left(\frac{v}{y}\left(n+\frac{1}{2}\right)+\frac{t}{y}\right)\vartheta_4\left(\frac{v}{y}\left(n-\frac{1}{2}\right)+\frac{t}{y}\right)\\
&-\frac{1}{y}J_{1}(v,y)\vartheta_1\left(\frac{v}{y}\left(n+\frac{1}{2}\right)+\frac{t}{y}\right)\vartheta_1\left(\frac{v}{y}\left(n-\frac{1}{2}\right)+\frac{t}{y}\right).
\end{split}
\end{equation}
Also, after differentiating the equation in \eqref{pr4-3-3} with
\begin{eqnarray}\label{prtheorem2-3}
\left\{
\begin{array}{l}
\vspace{2mm}
X=\frac{v}{y}\left(n+\frac{1}{2}\right)+\frac{t}{y},\\
\vspace{2mm}
Y=\frac{v}{y}\left(n-\frac{1}{2}\right)+\frac{t}{y},\\
\vspace{2mm}
U=\frac{v}{y}+h,\\
V=h,\\
\end{array}
\right.
\end{eqnarray}
by $h$ and putting $h=0$, we obtain
\begin{equation}\label{prtheorem2-4}
\begin{split}
&D_{t} \vartheta_1\left(\frac{v}{y}\left(n+\frac{1}{2}\right)+\frac{t}{y}\right) \cdot \vartheta_1\left(\frac{v}{y}\left(n-\frac{1}{2}\right)+\frac{t}{y}\right) \\
&=\frac{1}{y}\frac{\vartheta^{'}_{4}\left(\frac{v}{y}\right)}{\vartheta_{4}\left(\frac{v}{y}\right)}\vartheta_1\left(\frac{v}{y}\left(n+\frac{1}{2}\right)+\frac{t}{y}\right)\vartheta_1\left(\frac{v}{y}\left(n-\frac{1}{2}\right)+\frac{t}{y}\right)\\
&-\frac{1}{y}\frac{\vartheta_{1}\left(\frac{v}{y}\right)\vartheta^{'}_{1}\left(0\right)}{\vartheta_{4}\left(\frac{v}{y}\right)\vartheta_{4}\left(0\right)}\vartheta_4\left(\frac{v}{y}\left(n+\frac{1}{2}\right)+\frac{t}{y}\right)\vartheta_4\left(\frac{v}{y}\left(n-\frac{1}{2}\right)+\frac{t}{y}\right)\\
&=\frac{1}{y}d_{4}\left(\frac{v}{y};\frac{2}{y}\right)\vartheta_1\left(\frac{v}{y}\left(n+\frac{1}{2}\right)+\frac{t}{y}\right)\vartheta_1\left(\frac{v}{y}\left(n-\frac{1}{2}\right)+\frac{t}{y}\right)\\
&-\frac{1}{y}J_{1}(v,y)\vartheta_4\left(\frac{v}{y}\left(n+\frac{1}{2}\right)+\frac{t}{y}\right)\vartheta_4\left(\frac{v}{y}\left(n-\frac{1}{2}\right)+\frac{t}{y}\right).
\end{split}
\end{equation}
Also, after differentiating the equation in \eqref{pr4-3-3} with
\begin{eqnarray}\label{prtheorem2-5}
\left\{
\begin{array}{l}
\vspace{2mm}
X=\frac{v}{y}\left(n+\frac{1}{2}\right)+\frac{t}{y},\\
\vspace{2mm}
Y=\frac{v}{y}+h,\\
\vspace{2mm}
U=\frac{v}{y}\left(n-\frac{1}{2}\right)+\frac{t}{y},\\
V=h,\\
\end{array}
\right.
\end{eqnarray}
by $h$ and putting $h=0$, we obtain
\begin{equation}\label{prtheorem2-6}
\begin{split}
&D_{t} \vartheta_1\left(\frac{v}{y}\left(n+\frac{1}{2}\right)+\frac{t}{y}\right) \cdot \vartheta_4\left(\frac{v}{y}\left(n-\frac{1}{2}\right)+\frac{t}{y}\right) \\
&=\frac{1}{y}\frac{\vartheta^{'}_{1}\left(\frac{v}{y}\right)}{\vartheta_{1}\left(\frac{v}{y}\right)}\vartheta_1\left(\frac{v}{y}\left(n+\frac{1}{2}\right)+\frac{t}{y}\right)\vartheta_4\left(\frac{v}{y}\left(n-\frac{1}{2}\right)+\frac{t}{y}\right)\\
&-\frac{1}{y}\frac{\vartheta_{4}\left(\frac{v}{y}\right)\vartheta^{'}_{1}\left(0\right)}{\vartheta_{1}\left(\frac{v}{y}\right)\vartheta_{4}\left(0\right)}\vartheta_4\left(\frac{v}{y}\left(n+\frac{1}{2}\right)+\frac{t}{y}\right)\vartheta_1\left(\frac{v}{y}\left(n-\frac{1}{2}\right)+\frac{t}{y}\right)\\
&=\frac{1}{y}d_{1}\left(\frac{v}{y};\frac{2}{y}\right)\vartheta_1\left(\frac{v}{y}\left(n+\frac{1}{2}\right)+\frac{t}{y}\right)\vartheta_4\left(\frac{v}{y}\left(n-\frac{1}{2}\right)+\frac{t}{y}\right)\\
&-\frac{1}{y}J_{2}(v,y)\vartheta_4\left(\frac{v}{y}\left(n+\frac{1}{2}\right)+\frac{t}{y}\right)\vartheta_1\left(\frac{v}{y}\left(n-\frac{1}{2}\right)+\frac{t}{y}\right).
\end{split}
\end{equation}
Also, after differentiating the equation in \eqref{pr4-3-3} with
\begin{eqnarray}\label{prtheorem2-7}
\left\{
\begin{array}{l}
\vspace{2mm}
X=\frac{v}{y}\left(n-\frac{1}{2}\right)+\frac{t}{y},\\
\vspace{2mm}
Y=\frac{v}{y}+h,\\
\vspace{2mm}
U=\frac{v}{y}\left(n+\frac{1}{2}\right)+\frac{t}{y},\\
V=h,\\
\end{array}
\right.
\end{eqnarray}
by $h$ and putting $h=0$, we obtain
\begin{equation}\label{prtheorem2-8}
\begin{split}
&D_{t}\vartheta_4\left(\frac{v}{y}\left(n+\frac{1}{2}\right)+\frac{t}{y}\right) \cdot \vartheta_1\left(\frac{v}{y}\left(n-\frac{1}{2}\right)+\frac{t}{y}\right)\\
&=\frac{1}{y}\frac{\vartheta^{'}_{1}\left(\frac{v}{y}\right)}{\vartheta_{1}\left(\frac{v}{y}\right)}\vartheta_4\left(\frac{v}{y}\left(n+\frac{1}{2}\right)+\frac{t}{y}\right)\vartheta_1\left(\frac{v}{y}\left(n-\frac{1}{2}\right)+\frac{t}{y}\right)\\
&-\frac{1}{y}\frac{\vartheta_{4}\left(\frac{v}{y}\right)\vartheta^{'}_{1}\left(0\right)}{\vartheta_{1}\left(\frac{v}{y}\right)\vartheta_{4}\left(0\right)}\vartheta_1\left(\frac{v}{y}\left(n+\frac{1}{2}\right)+\frac{t}{y}\right)\vartheta_4\left(\frac{v}{y}\left(n-\frac{1}{2}\right)+\frac{t}{y}\right)\\
&=\frac{1}{y}d_{1}\left(\frac{v}{y};\frac{2}{y}\right)\vartheta_4\left(\frac{v}{y}\left(n+\frac{1}{2}\right)+\frac{t}{y}\right)\vartheta_1\left(\frac{v}{y}\left(n-\frac{1}{2}\right)+\frac{t}{y}\right)\\
&-\frac{1}{y}J_{2}(v,y)\vartheta_1\left(\frac{v}{y}\left(n+\frac{1}{2}\right)+\frac{t}{y}\right)\vartheta_4\left(\frac{v}{y}\left(n-\frac{1}{2}\right)+\frac{t}{y}\right).
\end{split}
\end{equation}
Adding \eqref{prtheorem2-2} and \eqref{prtheorem2-4} yields \eqref{sdm-3}, and subtracting \eqref{prtheorem2-6} from \eqref{prtheorem2-8} yields \eqref{sdm-4}. 
\end{proof}
Then we have the following theorem.
\begin{theorem}\label{thm:sdm}
Let $y>0$ and $v \in \R \setminus{(y\mathbb{Z})}$.
The potential function $\Theta$ defined by \eqref{sdm-2-2} simultaneously satisfies the semi-discrete potential mKdV 
and sine-Gordon equations
\begin{align}
\frac{d}{dt}\left(\frac{\Theta_{n+1}+\Theta_{n}}{2}\right) & =\left(1+\cos\lambda\right)\rho_{+}\sin\left(\frac{\Theta_{n+1}-\Theta_{n}}{2}\right),
\label{sdm-1} \\
\frac{d}{dt}\left(\frac{\Theta_{n+1}-\Theta_{n}}{2}\right)&=\left(1-\cos\lambda\right)\rho_{-}\sin\left(\frac{\Theta_{n+1}+\Theta_{n}}{2}\right).
\label{sdsg-1}
\end{align}
\end{theorem}
\begin{proof}
We first show that the potential function satisfies
\begin{align}\label{sdm-7}
\frac{d}{dt}\left(\frac{\Theta_{n+1}+\Theta_{n}}{2}\right) & =\frac{L_{1}(v,y)}{y}\sin\left(\frac{\Theta_{n+1}-\Theta_{n}}{2}\right),
\end{align}
using Lemma \ref{prop:sdm2}.
Then we show that
\begin{equation}\label{sdm-r-2}
     \frac{L_{1}(v,y)}{y}\frac{1}{1+\cos\lambda}=i\frac{\vartheta_{1}^{'}(0)}{\vartheta_{1}(iv)}\frac{u_{1}}{\vartheta_{3}(r)^{2}\vartheta_{1}(r)^{2}}=\rho_{+}.
\end{equation}
Suppose that a solution to the equation \eqref{sdm-7} is expressed as \eqref{sdm-2}. 
Then both sides of the equation can be rewritten as follows.
\begin{equation}\label{sdm-5-1}
\begin{split}
&\frac{d}{dt}\left(\frac{\Theta_{n+1}+\Theta_{n}}{2}\right)=\frac{1}{i}\frac{d}{dt}\left(\log\frac{\widetilde{\tau}_{n+1}}{\widetilde{\tau}^{*}_{n+1}}+\log\frac{\widetilde{\tau}_{n}}{\widetilde{\tau}^{*}_{n}}\right)\\
&=\frac{1}{i}\frac{d}{dt}\left(\log\frac{\widetilde{\tau}_{n+1}}{\widetilde{\tau}^{*}_{n}}-\log\frac{\widetilde{\tau}^{*}_{n+1}}{\widetilde{\tau}_{n}}\right)\\
&=\frac{1}{i}\left(\frac{D_{t}\widetilde{\tau}_{n+1}\cdot \widetilde{\tau}^{*}_{n}}{\widetilde{\tau}_{n+1}\widetilde{\tau}^{*}_{n}}-\frac{D_{t}\widetilde{\tau}^{*}_{n+1}\cdot \widetilde{\tau}_{n}}{\widetilde{\tau}^{*}_{n+1}\widetilde{\tau}_{n}}\right),
\end{split}
\end{equation}
and
\begin{equation}\label{sdm-5-2}
\begin{split}
&\sin\left(\frac{\Theta_{n+1}-\Theta_{n}}{2}\right)=\frac{1}{2i}\left(\frac{\widetilde{\tau}_{n+1}\widetilde{\tau}^{*}_{n}}{\widetilde{\tau}^{*}_{n+1}\widetilde{\tau}_{n}}-\frac{\widetilde{\tau}^{*}_{n+1}\widetilde{\tau}_{n}}{\widetilde{\tau}_{n+1}\widetilde{\tau}^{*}_{n}}\right).
\end{split}
\end{equation}
Then it can be seen that a sufficient condition for \eqref{sdm-2} to be a solution of \eqref{sdm-7} is that
\begin{equation}\label{sdm-5-3}
\begin{split}
&D_{t}\widetilde{\tau}_{n+1}\cdot \widetilde{\tau}^{*}_{n}=-\frac{L_{1}}{2y}\widetilde{\tau}^{*}_{n+1}\widetilde{\tau}_{n}+\frac{L_{1}}{2y}\beta_{9}\widetilde{\tau}_{n+1}\widetilde{\tau}_{n}^{*},
\end{split}
\end{equation}
 is satisfied. Substituting $\widetilde{\tau}_{n}=A_{n}+iB_{n}$ for this, we get
\begin{equation}\label{sdm-5-4}
\begin{split}
&D_{t}\left(A_{n+1}\cdot A_{n}+B_{n+1}\cdot B_{n}\right)-iD_{t}\left(A_{n+1}\cdot B_{n}-B_{n+1}\cdot A_{n}\right)\\
&=\frac{L_{1}}{2y}\left(\beta_{9}-1\right)\left(A_{n+1}A_{n}+B_{n+1}B_{n}\right)-i\frac{L_{1}}{2y}\left(\beta_{9}+1\right)\left(A_{n+1}B_{n}-B_{n+1}A_{n}\right).
\end{split}
\end{equation}
The real part of this equation yields \eqref{sdm-3}, and the imaginary part yields \eqref{sdm-4}.
Thus, Lemma \ref{prop:sdm2} shows that the potential function satisfies the equation \eqref{sdm-7}.
We then calculate $L_{1}(v,y)/y$.
From \eqref{id:const}, \eqref{id:modular} and \eqref{id:modular2}, we see
\begin{equation}\label{eqn:L1-1}
\begin{split}
\frac{L_{1}(v,y)}{y}&=\frac{1}{y}\frac{\vartheta_{1}^{'}\left(\frac{v}{y}; \frac{2i}{y}\right)\vartheta_{4}\left(\frac{v}{y}; \frac{2i}{y}\right)-\vartheta_{1}\left(\frac{v}{y}; \frac{2i}{y}\right)\vartheta_{4}^{'}\left(\frac{v}{y}; \frac{2i}{y}\right)}{\vartheta_{1}\left(\frac{v}{y}; \frac{2i}{y}\right)\vartheta_{4}\left(\frac{v}{y}; \frac{2i}{y}\right)}\\
&+\frac{1}{y}\frac{\vartheta_{1}^{'}\left(0; \frac{2i}{y}\right)\vartheta_{1}\left(\frac{v}{y}; \frac{2i}{y}\right)}{\vartheta_{4}\left(0; \frac{2i}{y}\right)\vartheta_{4}\left(\frac{v}{y}; \frac{2i}{y}\right)}+\frac{1}{y}\frac{\vartheta_{1}^{'}\left(0; \frac{2i}{y}\right)\vartheta_{4}\left(\frac{v}{y}; \frac{2i}{y}\right)}{\vartheta_{4}\left(0; \frac{2i}{y}\right)\vartheta_{1}\left(\frac{v}{y}; \frac{2i}{y}\right)}\\
&=i\frac{\vartheta_{1}^{'}\left(0; iy\right)}{\vartheta_{4}\left(0; iy\right)}\frac{\vartheta_{4}\left(iv; iy\right)}{\vartheta_{1}\left(iv; iy\right)}+i\frac{\vartheta_{1}^{'}\left(0; iy\right)}{\vartheta_{2}\left(0; iy\right)}\frac{\vartheta_{2}\left(iv; iy\right)}{\vartheta_{1}\left(iv; iy\right)}\\
&=i\frac{\vartheta_{1}^{'}\left(0; iy\right)}{\vartheta_{1}\left(iv; iy\right)}\frac{\vartheta_{4}\left(iv; iy\right)\vartheta_{2}\left(0; iy\right)+\vartheta_{2}\left(iv; iy\right)\vartheta_{4}\left(0; iy\right)}{\vartheta_{4}\left(0; iy\right)\vartheta_{2}\left(0; iy\right)}.
\end{split}
\end{equation}
On the other hand, we see
\begin{equation}\label{eqn:L1-2}
\begin{split}
&1+\cos\lambda\\
&=\frac{\left\{\vartheta_{1}\left(\frac{iv}{2}+r; iy\right)\vartheta_{3}\left(-\frac{iv}{2}+r; iy\right)+\vartheta_{1}\left(-\frac{iv}{2}+r; iy\right)\vartheta_{3}\left(\frac{iv}{2}+r; iy\right)\right\}^{2}}{2u_{1}}.
\end{split}
\end{equation}
We can see that \eqref{sdm-r-2} holds by rewriting \eqref{eqn:L1-1} and \eqref{eqn:L1-2} using the following.
\begin{equation}
\begin{split}
&\vartheta_{2}\left(iv/2\right)\vartheta_{4}\left(iv/2\right)\left\lbrace\vartheta_{1}\left(iv/2+r\right)\vartheta_{3}\left(-iv/2+r\right)+\vartheta_{1}\left(-iv/2+r\right)\vartheta_{3}\left(iv/2+r\right)\right\rbrace\\
&=\vartheta_{3}\left(r\right)\vartheta_{1}\left(r\right)\left\lbrace\vartheta_{4}\left(iv\right)\vartheta_{2}\left(0\right)+\vartheta_{2}\left(iv\right)\vartheta_{4}\left(0\right)\right\rbrace,\\
&\vartheta_{3}\left(r\right)\vartheta_{1}\left(r\right)\left\lbrace\vartheta_{1}\left(iv/2+r\right)\vartheta_{3}\left(-iv/2+r\right)+\vartheta_{1}\left(-iv/2+r\right)\vartheta_{3}\left(iv/2+r\right)\right\rbrace\\
&=\vartheta_{2}\left(iv/2\right)\vartheta_{4}\left(iv/2\right)\left\lbrace\vartheta_{4}\left(2r\right)\vartheta_{2}\left(0\right)-\vartheta_{2}\left(2r\right)\vartheta_{4}\left(0\right)\right\rbrace,\\
&\vartheta_{2}\left(0\right)\vartheta_{4}\left(0\right)\left\lbrace\vartheta_{4}\left(2r\right)\vartheta_{2}\left(0\right)-\vartheta_{2}\left(2r\right)\vartheta_{4}\left(0\right)\right\rbrace\\
&=2\vartheta_{3}\left(r\right)^{2}\vartheta_{1}\left(r\right)^{2}.
\end{split}
\end{equation}
\end{proof}


\begin{remark}
Rewriting \eqref{anothermkdv} in terms of the potential function $\Theta$, we obtain
\begin{equation}\label{anothermkdv9}
\frac{d}{dt}\Theta_{n}=4\beta_{1}\left(\eta_{1}\cos\lambda+\eta_{2}\sin\lambda\right)\tan\left(\frac{\Theta_{n+1}-\Theta_{n-1}}{4}\right),
\end{equation}
which is also called the \emph{semi-discrete potential mKdV equation} \cite{inoguchi2012explicit,hoffmann2004discrete,feng2011discrete}. In other words, the potential function $\Theta$ defined by \eqref{sdm-2-2} is a simultaneous solution of three differential-difference equations, \eqref{sdm-1}, \eqref{sdsg-1} and \eqref{anothermkdv9}. It does not usually happen that these equations are simultaneously solved. For example, it can be checked  that the 2- and 3-soliton solutions of the semi-discrete sine-Gordon equation constructed in \cite{feng2011discrete,feng2010integrable} are not solutions of the semi-discrete potential mKdV equation.
\end{remark}

Note the difference in approaches of 
\cite{KKP:linkage} and this work;
\cite{KKP:linkage} starts with a geometric property
observed in the motion of Kaleidocycles, and derive evolution equations,
whereas we aim to construct a particular motion of Kaleidocycles and investigate the equations satisfied by them.
They turn out to arrive at the same place; 
the latter gives an explicit solution to the equations derived in the former.


\section{Existence of Kaleidocycles}\label{section:Necessary and sufficient conditions for closedness}

The curves given in Theorem \ref{thm:PhDKal2} are not necessarily closed.
In this section, we derive necessary and sufficient conditions on the parameters $v$, $r$, and $y$ for the curve $\gamma$ to be closed, and hence to define a Kaleidocycle.
We also prove that such parameters exist when $k \geq 6$.

The overall strategy is as follows.
First, we derive a necessary condition for closedness.
For the curve $\gamma$ to be closed, the absolute value of the signed curvature angle must be periodic with period $k$.
This implies that $m := kv/y \in \mathbb{Z} \setminus (k\mathbb{Z})$ (\S\ref{subsection:Period of the curvature angle}).

Assuming $m \in \mathbb{Z} \setminus (k\mathbb{Z})$, we then rewrite the closedness condition $\gamma_{n+k}(t) = \gamma_n(t)$ as two equations in the parameters $v$, $r$, and $y$ (\S\ref{subsection:Closure conditions}).
One of these equations determines $r = r(y)$ uniquely when $y$ lies in a certain range (\S\ref{subsection:An explicit solution to the closure condition}).
Substituting this expression into the remaining equation, we show that a solution $y$ exists whenever $3 \leq m \leq k/2$ (\S\ref{subsection:Existence of a solution to closure condition}).

Finally, by considering two transformations of $m$ and $n$ that preserve the shape of the curve, we may assume without loss of generality that $3 \leq m \leq k/2$.
Combining these results, we conclude that for $k \geq 6$, there always exists a choice of parameters for which the curve in \Cref{thm:PhDKal2} is closed, thereby establishing the existence of a Kaleidocycle (\S\ref{subsection:Transformations of a parameter}).

\subsection{Period of the curvature angle}\label{subsection:Period of the curvature angle}
We say a curve $\gamma$ is closed with a period $k\geq 2$
when $\gamma_{n+k}(t)=\gamma_n(t)$ for all $n$ and $t$.
In this case, the absolute value of the curvature angles also has a period $k$.
Hence, we obtain the necessary condition for the curve to be closed that 
\begin{equation}
\begin{split}
\tan^2\frac{\kappa_{n+k}(t)}{2}-\tan^2\frac{\kappa_n(t)}{2}=0,
\end{split}
\end{equation}
holds for all $n$ and $t$.
Rewriting this condition using Jacobi $\mathrm{cn}$-function, we obtain
\begin{equation}\label{tan^2kappa}
\begin{split}
&\tan^2\frac{\kappa_{n+k}(t)}{2}-\tan^2\frac{\kappa_n(t)}{2}\\
&=\frac{\vartheta_1\left(\frac{v}{y};\frac{i}{y}\right)^2\vartheta_2\left(0;\frac{i}{y}\right)^2}{\vartheta_3\left(\frac{v}{y};\frac{i}{y}\right)^2\vartheta_4\left(0;\frac{i}{y}\right)^2}\left\{\mathrm{cn}^2\left[2\tilde{K}\left(\tilde{\mu}_n(t)+\frac{v}{y}k\right)\right]-\mathrm{cn}^2\left[2\tilde{K}\tilde{\mu}_n(t)\right]\right\}=0,
\end{split}
\end{equation}
where we note $\tilde{\mu}_n(t)=\frac{v}{y}\left(n-\frac{1}{2}\right)+\frac{t}{y}$.
Since $\tilde{\mu}_n(t)$ can be regarded as a linear function of $t$ when $v, y$ and $n$ are fixed, $\tilde{\mu}_n(t)$ attains all reals as $t$ varies.
Thus, we only need to check that \eqref{tan^2kappa} holds for $n=0$.
Since $\mathrm{cn}^2$ has the real period $2\tilde{K}$, it follows that $m=kv/y$ must belong to $\mathbb{Z}\setminus(k\mathbb{Z})$.
Hereafter, let $v$ satisfy 
\begin{align}\label{kvy}
    v=\frac{m}{k}y.
\end{align}
\subsection{Closure conditions}\label{subsection:Closure conditions}
Here we obtain the necessary and sufficient conditions satisfied by $v$,$r$, and $y$ for the curve to be closed.
We first rewrite $\gamma_{n+k}(t)=\gamma_n(t)$ in terms of $\tau$ functions. We have
\begin{align}\label{cl2}
&\frac{H_{n+k}}{F_{n+k}}=\frac{H_{n}}{F_{n}},\\
& n+k-2\displaystyle\frac{\partial\log F_{n+k}}{\partial z}=n-2\displaystyle\frac{\partial\log F_{n}}{\partial z}. \nonumber
\end{align} 
We further rewrite the left-hand side using $m$ and \eqref{cl}. Then we obtain
\begin{align}\label{cl4}
&\frac{H_{n+k}}{F_{n+k}}=\frac{H_{n}}{F_{n}}\left(R_{1}R_{3}\right)^{-k}\exp\left(4mr\pi i\right)(-1)^{m},\\
&F_{n+k}=F_{n}\exp\left(k\frac{\Delta_{3}}{\Delta_{3}-\Delta_{1}}z\right)\exp\left[-\pi i\left\{2m\left(\mu_{n}-\frac{1}{2}iv\right)+m^{2}w\right\}\right],\nonumber
\end{align}
where $\mu_{n}=ivn+\frac{z}{\Delta_{3}-\Delta_{1}}$. 
Using \eqref{cl4}, \eqref{cl2} can be rewritten as
\begin{align}\label{cl5}
&\Lambda(v, r, y):=\exp\left\{\pi im\left(4r+1\right)\right\}\left(R_{1}R_{3}\right)^{-k}-1=0, \\
&\Xi(v, r, y):=\frac{\left(\Delta_{3}+\Delta_{1}\right)k-4\pi im}{\Delta_{3}-\Delta_{1}}=0.\label{cl5Xi}
\end{align}
Note that from the definitions of $R_j$ and $\Delta_{j}$, we see that
\begin{equation} 
 \frac{\partial}{\partial r}\log R_j=\Delta_{j} \quad(j=1,2,3,4),
\end{equation}
holds.
From this, we have 
\begin{align}\label{perio6} 
\frac{\partial}{\partial r}\log (\Lambda + 1) =4\pi im-k\left(\Delta_{3}+\Delta_{1}\right)= -\left(\Delta_{3 }-\Delta_{1}\right)\Xi.
\end{align} 
Since $\Delta_{3}-\Delta_{1}\neq 0$ holds, \eqref{cl5Xi} can be written as
\begin{align}
\frac{\partial}{\partial r}\log (\Lambda + 1)=0.\label{cl6}
\end{align}

On the other hand, the closure condition \eqref{cl2} holds for any $n\in\mathbb{Z}$ and $t\in\R$ when the parameters $v$, $r$, and $y$ satisfy $kv=my$, \eqref{cl5} and \eqref{cl6}. Thus, they give necessary and sufficient condition for the curve given by Theorem \ref{thm:PhDKal2} to be a closed curve for every $t\in\R$.
\subsection{An explicit solution to the condition \texorpdfstring{$\frac{\partial}{\partial r}\log (\Lambda + 1)=0$}{for the derivative of Lambda}}\label{subsection:An explicit solution to the closure condition}
Here we solve \eqref{cl6} with respect to $r$.
From the period of elliptic theta function, we see the period and symmetry of $\Lambda(v, r, y)$ with respect to $r$ as follows: 
$\Lambda(v, r+1, y)=\Lambda(v, r, y)$, $\Lambda(v, -r, y)=\Lambda(v, r, y)^{*}$.
Therefore, when searching for $(v, r, y)$ such that $\Lambda(v, r, y)=0$, we do not lose generality by assuming $0\leq r\leq 1/2$. We exclude the case $r=0, 1/2$ because $\gamma$ forms a plane curve. Hereafter, we assume $0<r<1/2$.
Since
\begin{equation}
\begin{split}
&2u_{1}=\vartheta_4\left(0\right)\vartheta_2\left(0\right)\left(\vartheta_4\left(2r\right)\vartheta_2\left(iv\right)-\vartheta_2\left(2r\right)\vartheta_4\left(iv\right)\right),\\
&u_{1}\left(\Delta_{3}+\Delta_{1}\right)=\vartheta_4\left(0\right)\vartheta_2\left(0\right)\left(\vartheta_2\left(2r\right)\vartheta_4^{'}\left(iv\right)-\vartheta_4\left(2r\right)\vartheta_2^{'}\left(iv\right)\right),\\
\end{split}
\end{equation}
holds, we have
\begin{equation}\label{close1}
\begin{split}
-&\frac{u_{1}}{k}\frac{\partial}{\partial r}\log (\Lambda + 1)=-u_{1}\left(\Delta_{3}+\Delta_{1}\right)+4\pi i \frac{v}{y}u_{1}\\
&=\vartheta_4\left(0\right)\vartheta_2\left(0\right)\left(\vartheta_4\left(2r\right)\vartheta_2^{'}\left(iv\right)-\vartheta_2\left(2r\right)\vartheta_4^{'}\left(iv\right)\right)\\
&+2\pi i \frac{v}{y}\vartheta_4\left(0\right)\vartheta_2\left(0\right)\left(\vartheta_4\left(2r\right)\vartheta_2\left(iv\right)-\vartheta_2\left(2r\right)\vartheta_4\left(iv\right)\right)=0.
\end{split}
\end{equation}
Arranging this we obtain
\begin{equation}\label{cnphi}
\frac{\vartheta_4\left(0\right)\vartheta_2\left(2r\right)}{\vartheta_2\left(0\right)\vartheta_4\left(2r\right)}=\frac{\vartheta_4\left(0\right)\left(2\pi i\frac{v}{y}\vartheta_2\left(iv\right)+\vartheta_2^{'}\left(iv\right)\right)}{\vartheta_2\left(0\right)\left(2\pi i\frac{v}{y}\vartheta_4\left(iv\right)+\vartheta_4^{'}\left(iv\right)\right)}.
\end{equation}
The right-hand side of this equation does not contain $r$. The left-hand side is nothing but Jacobi $\mathrm{cn}$-function.
\begin{equation}
\mathrm{cn}\left[4Kr,\xi\right]=\frac{\vartheta_4\left(0\right)\vartheta_2\left(2r\right)}{\vartheta_2\left(0\right)\vartheta_4\left(2r\right)},
\end{equation}
where $\xi=\vartheta_2\left(0\right)^2/\vartheta_3\left(0\right)^2$ is an elliptic modulus and $K=\pi\vartheta_3\left(0\right)^2/2$ is a complete elliptic integral of the first kind.
We define the following continuous real-valued function $\varphi$ for the ratio of $m$ and $k$ and for $y>0$.
\begin{equation}
\begin{split}
\varphi\left(\frac{m}{k},y\right):=\frac{\vartheta_4\left(0\right)\left(2\pi i\frac{m}{k}\vartheta_2\left(i\frac{m}{k}y\right)+\vartheta_2^{'}\left(i\frac{m}{k}y\right)\right)}{\vartheta_2\left(0\right)\left(2\pi i\frac{m}{k}\vartheta_4\left(i\frac{m}{k}y\right)+\vartheta_4^{'}\left(i\frac{m}{k}y\right)\right)}
&=\frac{\vartheta_2\left(0;\frac{i}{y}\right)\vartheta_4^{'}\left(\frac{m}{k};\frac{i}{y}\right)}{\vartheta_4\left(0;\frac{i}{y}\right)\vartheta_2^{'}\left(\frac{m}{k};\frac{i}{y}\right)}.
\end{split}
\end{equation}
When $-1<\varphi\left(m/k,y\right)<1$ is satisfied, 
\eqref{cnphi} determines a unique $r=r(y)$ in $(0,1/2)$
 because on the interval $0<r<1/2$, the function
$r\mapsto \mathrm{cn}(4Kr,\xi)$ is strictly decreasing from $1$ to $-1$.

\subsubsection{Existence of an inverse function}\label{subsubsection:Proof of existence of inverse function}
Here we show that for a certain range of $y$, $-1<\varphi\left(m/k,y\right)\leq 0$ holds when $0<m\leq k/2$. 
Then we find $r=r(y)$ explicitly using elliptic integrals.
First, from the zeros and evenness of elliptic theta function we see the following: For any $y>0$, $\varphi(1/2,y)=0$ and $\varphi(m/k,y)\neq 0 \quad(0<m/k<1/2)$.
Therefore, when $m/k=1/2$, $r=1/4$.
Next, consider the limit of $\varphi(m/k,y)$. From the definition of elliptic theta function, when $u$ is fixed and $y\rightarrow +0$, we see that
\begin{align}
\vartheta_2 \left(u; \frac{i}{y}\right)
&= 2q^{\frac{1}{4}}\cos \pi u +  \mathcal{O} (q^{\frac{9}{4}}), \\
\vartheta_4 \left(u; \frac{i}{y}\right)
&= 1 +  \mathcal{O} (q), \nonumber\\
\vartheta_2^{'} \left(u; \frac{i}{y}\right)
&= -2\pi q^{\frac{1}{4}}\sin \pi u +  \mathcal{O} (q^{\frac{9}{4}}), \nonumber\\
\vartheta_4^{'} \left(u; \frac{i}{y}\right)
&= 4\pi q\sin 2\pi u +  \mathcal{O} (q^{4}),\nonumber
\end{align}
where we denote $q=e^{-\frac{\pi}{y}}$.
From this we have
\begin{align}
    \lim_{y\to +0}\varphi\left(\frac{m}{k},y\right)=\lim_{y\to +0}\frac{\vartheta_2\left(0;\frac{i}{y}\right)\vartheta_4^{'}\left(\frac{m}{k};\frac{i}{y}\right)}{\vartheta_4\left(0;\frac{i}{y}\right)\vartheta_2^{'}\left(\frac{m}{k};\frac{i}{y}\right)}=0\quad(0<\frac{m}{k}<\frac{1}{2}).
\end{align}
On the other hand, the following hold when $y\rightarrow \infty$
\begin{align}
p^{\left(\frac{m}{k}\right)^2}\vartheta_2 \left(i\frac{m}{k}y; iy\right)
&= p^{\left(\frac{1}{2}-\frac{m}{k}\right)^2} +  \mathcal{O} (p^{\left(\frac{1}{2}+\frac{m}{k}\right)^2}), \\
p^{\left(\frac{m}{k}\right)^2}\vartheta_4 \left(i\frac{m}{k}y; iy\right)
&= p^{\left(\frac{m}{k}\right)^2} +  \mathcal{O} (p^{\left(1-\frac{m}{k}\right)^2}), \nonumber\\
ip^{\left(\frac{m}{k}\right)^2}\vartheta_2^{'} \left(i\frac{m}{k}y; iy\right)
&= \pi p^{\left(\frac{1}{2}-\frac{m}{k}\right)^2} +  \mathcal{O} (p^{\left(\frac{1}{2}+\frac{m}{k}\right)^2}), \nonumber\\
ip^{\left(\frac{m}{k}\right)^2}\vartheta_4^{'} \left(i\frac{m}{k}y; iy\right)
&=-2\pi p^{\left(1-\frac{m}{k}\right)^2} +  \mathcal{O} (p^{\left(1+\frac{m}{k}\right)^2})\quad(0<\frac{m}{k}<\frac{1}{2}),\nonumber
\end{align}
where we denote $p=e^{-\pi y}$.
From this we have
\begin{align}
    \lim_{y\to \infty}\varphi\left(\frac{m}{k},y\right)=\lim_{y\to \infty}\frac{\vartheta_4\left(0\right)\left(2\pi i\frac{m}{k}\vartheta_2\left(i\frac{m}{k}y\right)+\vartheta_2^{'}\left(i\frac{m}{k}y\right)\right)}{\vartheta_2\left(0\right)\left(2\pi i\frac{m}{k}\vartheta_4\left(i\frac{m}{k}y\right)+\vartheta_4^{'}\left(i\frac{m}{k}y\right)\right)}=-\infty\quad(0<\frac{m}{k}<\frac{1}{2}).
\end{align}
Hence, when $m$ is fixed within $0<m<k/2$, there exists $y$ such that $\varphi\left(m/k,y\right)=-1$. We denote the smallest one among them by $y_0(m/k)$.
Summarizing the above, when $m$ is fixed within $0<m<k/2$ and $y$ moves in the range $0<y< y_0(m/k)$, $r=r(y)$ satisfying \eqref{cl6} is uniquely determined within the range of $0<r<1/2$.
By taking the inverse function of $\mathrm{cn}$ in \eqref{cnphi}, we obtain
\begin{equation}\label{cnphi2}
 \begin{split}
 &r\left(y;\frac{m}{k}\right)=\frac{1}{4K}\int^{1}_{\varphi(m/k,y)}\frac{dx}{\sqrt{(1-x^2)(1-\xi^2+\xi^2x^2)}}.
  \end{split}
\end{equation}
This also holds when $m/k=1/2$, hence we see that $1/4\leq r<1/2$.
If we assume $-1/2<r<0$, the same calculation yields 
\begin{equation}\label{cnphi3} 
 \begin{split} 
 r\left(y;\frac{m}{k}\right)=-\frac{1}{4K}\int^{1}_{\varphi(m/k,y)}\frac{dx}{\sqrt{(1-x^2)(1-\xi^2+\xi^2x^2)}}.
 \end{split} 
\end{equation}
From this, we see $-1/2< r \leq -1/4$ holds when $0<m\leq k/2$.
The sign of $r$ is chosen so that $\sin\lambda = \langle\widetilde{B}_{n},\widetilde{N}_{n-1}\rangle>0$ holds.

\subsection{Existence of a solution to the condition
\texorpdfstring{$\Lambda=0$}{of Lambda}
}\label{subsection:Existence of a solution to closure condition}
By substituting $r=r(y)$ obtained in the previous section into $\Lambda(v, r, y)$, we prove that when $3\leq m\leq k/2$, $\Lambda(v, r(y), y)=0$ has a solution in $0<y<y_0$.
Note that $\Lambda + 1$ can also be written as 
\begin{equation}\label{Lambdarewrite}
\begin{split}
\Lambda+1&=\exp\left\{\pi im\left(4r+1\right)\right\}\left[\frac{\vartheta_1\left(\frac{1}{2}iv+r;iy\right)\vartheta_3\left(\frac{1}{2}iv+r;iy\right)}{\vartheta_1\left(-\frac{1}{2}iv+r;iy\right)\vartheta_3\left(-\frac{1}{2}iv+r;iy\right)}\right]^k\\
&=\exp{\left(\pi i m\right)}\left[\frac{\vartheta_1\left(\frac{1}{2}\frac{m}{k}+\frac{r}{iy};\frac{i}{y}\right)\vartheta_3\left(\frac{1}{2}\frac{m}{k}+\frac{r}{iy};\frac{i}{y}\right)}{\vartheta_1\left(-\frac{1}{2}\frac{m}{k}+\frac{r}{iy};\frac{i}{y}\right)\vartheta_3\left(-\frac{1}{2}\frac{m}{k}+\frac{r}{iy};\frac{i}{y}\right)}\right]^k,
 \end{split}
\end{equation}
by transforming the parameter $y\rightarrow1/y$.
Since $\Lambda+1$ is a complex number of magnitude $1$, one may denote
\begin{align}
    \Lambda+1=\exp{\left(i\Psi\right)},
\end{align}
using a continuous real-valued function $\Psi=\Psi(v, r, y)$.
When
\begin{align}
    \left |\Psi(v, r(y_0), y_0)-\lim_{y\to +0}\Psi(v, r(y), y)\right |,
\end{align}
is greater than $2\pi$, then there exists a solution to $\Lambda(v, r(y), y)=0$ in the range $0<y< y_0(m/k)$.
Here, we compute the above quantities and show that such a $y$ exists if $3\leq m\leq k/2$.
From a representation by an infinite product of elliptic theta functions,
\begin{align}
   \vartheta_1 \left(u; \frac{i}{y}\right)
&= 2q^{\frac{1}{4}}\sin (\pi u)\prod_{n=1}^\infty \left(1-q^{2n}\right)\left(1-q^{2n}e^{2\pi i u}\right)\left(1-q^{2n}e^{-2\pi i u}\right),\\
\vartheta_3 \left(u; \frac{i}{y}\right)
&= \prod_{n=1}^\infty \left(1-q^{2n}\right)\left(1+q^{2n-1}e^{2\pi i u}\right)\left(1+q^{2n-1}e^{-2\pi i u}\right), \nonumber
\end{align}
we obtain
\begin{equation}
\begin{split}
 &\frac{\vartheta_1 \left(\frac{r}{iy}+\frac{m}{2k}; \frac{i}{y}\right)}{\vartheta_1 \left(\frac{r}{iy}-\frac{m}{2k}; \frac{i}{y}\right)}\\
 &=\frac{\sin\left(\pi\frac{r}{iy}+\pi\frac{m}{2k}\right)}{\sin\left(\pi\frac{r}{iy}-\pi\frac{m}{2k}\right)}\exp{\left(-2i\sum_{n=1}^\infty \frac{1}{n}\frac{\sin (\pi n\frac{m}{k})}{1-q^{2n}}\left(q^{2n(1-r)}-q^{2n(1+r)}\right)\right)},\\
    &\frac{\vartheta_3 \left(\frac{r}{iy}+\frac{m}{2k}; \frac{i}{y}\right)}{\vartheta_3 \left(\frac{r}{iy}-\frac{m}{2k}; \frac{i}{y}\right)}\\
    &=\exp{\left(-2i\sum_{n=1}^\infty \frac{(-1)^n}{n}\frac{\sin (\pi n\frac{m}{k})}{1-q^{2n}}\left(q^{n(1-2r)}-q^{n(1+2r)}\right)\right)}.
 \end{split}
\end{equation}
From the explicit expression of $r$ in \eqref{cnphi2}, it follows that $1/4\leq r<1/2$ when $0<m\leq k/2$, which yields the following limit.
\begin{equation}
\begin{split}
    &\sum_{n=1}^\infty \left| \frac{(-1)^n}{n}\frac{\sin (\pi n\frac{m}{k})}{1-q^{2n}}\left(q^{n(1-2r)}-q^{n(1+2r)}\right)\right|\\
    &\leq\sum_{n=1}^\infty \frac{1}{1-q^{2}}q^{n(1-2r)}+\sum_{n=1}^\infty \frac{1}{1-q^{2}}q^{n(1+2r)}\\
    &=\frac{1}{1-q^{2}}\left(\frac{q^{(1-2r)}}{1-q^{(1-2r)}}+\frac{q^{(1+2r)}}{1-q^{(1+2r)}}\right)\rightarrow 0\quad(y\rightarrow +0).\\
    \end{split}
\end{equation}
This calculation remains valid even when $r$ is determined by \eqref{cnphi3}, that is, in the case where $-1/2 < r \leq -1/4$.
Therefore, we obtain
\begin{align}
    \lim_{y\to +0}\frac{\vartheta_3 \left(\frac{r}{iy}+\frac{m}{2k}; \frac{i}{y}\right)}{\vartheta_3 \left(\frac{r}{iy}-\frac{m}{2k}; \frac{i}{y}\right)}=1.
\end{align}
Similarly, we obtain
\begin{align}
    \lim_{y\to +0}\frac{\vartheta_1 \left(\frac{r}{iy}+\frac{m}{2k}; \frac{i}{y}\right)}{\vartheta_1 \left(\frac{r}{iy}-\frac{m}{2k}; \frac{i}{y}\right)}=\exp{\left(\pi i\frac{m}{k}\right)}.
\end{align}
From the above, we have
\begin{align}
    \lim_{y\to +0}\Psi(v,r(y), y)=2\pi m.
\end{align}
From the pseudoperiod of elliptic theta function, we have
\begin{align}
    \Psi(v,r(y_0), y_0)=3\pi m.
\end{align}
Summarizing the above discussion, we have
\begin{align}
    \left |\Psi(v,r(y_0), y_0)-\lim_{y\to +0}\Psi(v,r(y), y)\right |=\pi m.
\end{align}
Therefore, when $3\leq m\leq k/2$, there is always a solution in $\Lambda(v,r(y), y)=0$.
\subsection{Transformations of a parameter and a variable}\label{subsection:Transformations of a parameter}
We show that the set of our explicit solutions admits two
types of self-bijections that transform
\begin{itemize}
    \item $m\rightarrow m+k$, $n\rightarrow k-n$,
    \item $m\rightarrow -m$, $n\rightarrow n$.
\end{itemize}
A finite number of combinations of these transformations can make any $m\in\mathbb{Z}\setminus(k\mathbb{Z})$ fall in the interval $0<m\leq k/2$. Combining with the previous subsection, this proves the existence of Kaleidocycles for $k\ge 6$.

In the following, we will construct the self-bijections and prove that the Kaleidocycles obtained by the self-bijections are essentially identical to the Kaleidocycles before the transformation.
A combination of the two bijection exchanges a solution of the semi-discrete potential mKdV equation with one of the semi-discrete sine-Gordon equation, and the results are summarized in Table \ref{table:deformationequations}.
\subsubsection{Transformation $m\rightarrow m+k$, $n\rightarrow k-n$}\label{subsubsection:m+k}

Consider transformation $v\rightarrow v+y$.
Then, $m$ changes to $k + m$.
The period of elliptic theta function $\varphi(m/k+1,y)=-\varphi(m/k,y)$
and \eqref{cnphi3} implies the change $r\rightarrow -r-1/2$.
From the period of elliptic theta function, we see that $\Lambda(v+y, -r-1/2, y)=\Lambda(v, r, y)^{*}$ holds. Thus, this transformation preserves the closure conditions.
From the original binormal vectors $\widetilde{B}_{n}=\widetilde{B}_{n}(t;v,r,y)$,
the new binormal vectors are computed as
\begin{equation}
    \begin{split}
        (-1)^{k-n}\widetilde{B}_{k-n}\left(t-y/2; v, r, y\right).
    \end{split}
\end{equation}
That is, the new Kaleidocycle is obtained by shifting the time parameter, reversing the orientation of the central curve, and flipping every other hinge,
which is equivalently encoded by the torsion angle update  $\lambda\mapsto \pi - \lambda$.
The argument remains valid when $r$ is calculated by \eqref{cnphi2}.

\subsubsection{Transformation $m\rightarrow -m$, $n\rightarrow n$}\label{subsubsection:-m}
Consider the transformation $v\rightarrow -v$ and $r\rightarrow -r$.
Then, $m$ changes to $-m$.
From the period of elliptic theta function, we see that $\Lambda(-v, -r, y)=\Lambda(v, r, y)$ holds. Thus, this transformation preserves the closure conditions.
From the original binormal vectors $\widetilde{B}_{n}=\widetilde{B}_{n}(t;v,r,y)$,
the new binormal vectors are computed as
\begin{equation}
    \begin{split}
 -\widetilde{B}_{n}\left(-t; v, r, y\right).
    \end{split}
\end{equation}
That is, the new Kaleidocycle is obtained by reversing the time parameter and flipping all hinges.


Summarizing the results of this section, we have the following theorem.
\begin{theorem}\label{thm:kaleidocycle}
For $k\geq 6$ and an integer satisfying
$3\le m\le \frac{k}{2}$,
there exist $y>0$, $v \in \R \setminus{(y\mathbb{Z})}$, $r \in \R \setminus{(1/2\mathbb{Z})}$
such that
$kv=my$ and \eqref{cl5} and \eqref{cl6} hold.
Hence the curve given in \Cref{thm:PhDKal2} is closed for every $t\in\mathbb R$
and defines a $k$-Kaleidocycle whose
torsion angle is given by \eqref{lam2}.
\end{theorem}
\begin{remark}\label{remark:mKdVandsG}
When a curve is deformed according to \eqref{deform2}, a combination of the two transformations described above exchanges the deformation by the semi-discrete potential mKdV equation and the deformation by the semi-discrete sine-Gordon equation.
Consider the transformation $v\rightarrow -v-y$, $r\rightarrow r+1/2$ and $n\rightarrow k-n$.
Then the new $\Gamma$ and $C$ are computed as
\begin{equation}\label{transGammaC}
    \begin{split}
        &\Gamma_{\pm}\left(r+1/2,y\right)=\Gamma_{\mp}\left(r,y\right),\\
        &C_{\pm}\left(-v-y,r+1/2,y\right)=C_{\mp}\left(v,r,y\right).
    \end{split}
\end{equation}
The new potential function is computed as
\begin{equation}
 (-1)^{k-n+1}\Theta_{k-n}\left(-t+y/2;v,y\right).
\end{equation}
The coefficients of the semi-discrete potential mKdV equation and of the semi-discrete sine-Gordon equation are transformed as
\begin{equation}
    \left(1+\cos\lambda\right)\rho_{+}\rightarrow -\left(1-\cos\lambda\right)\rho_{-},\quad\left(1-\cos\lambda\right)\rho_{-}\rightarrow -\left(1+\cos\lambda\right)\rho_{+}.
\end{equation}
The transformation $\Theta_n\rightarrow\widetilde{\Theta}_{n}:=(-1)^{k-n+1}\Theta_{k-n}$ exchanges a solution of the semi-discrete potential mKdV equation \eqref{sdm-1} with one of the semi-discrete sine-Gordon equation \eqref{sdsg-1}. 
We see that this equation-level transformation is also achieved at the solution level by transforming parameters and a variable.
This is also consistent with the result of the transformation of $\Gamma$ and $C$.
\end{remark}
Table \ref{table:deformationequations} summarizes the relationships among the three deformation equations.
\begin{table}[htb]
\centering
  \caption{Relation between the three deformation equations\label{table:deformationequations}}
  \begin{tabular}{|l||l|l|}  \hline
   \begin{tabular}{c} 
   \end{tabular}
   & Deformation equation of the curve & Equation for the potential function \\ \hline\hline
    $C, \Gamma\in\mathbb{R}$ & 
    \begin{tabular}{c}
    $\dot{\gamma}_{n}=\left(
    \begin{array}{c}
0\\
0\\
2\eta_{1}\left(\beta_{1}-2\beta_{3}\right)-\beta_{6}-C
    \end{array}
  \right)$ \\
    $+2\beta_{1}\left(\eta_{1}T_{n}-\eta_{1}\tan\frac{\kappa_{n}}{2}\widetilde{N}_{n}+\eta_{2}\tan\frac{\kappa_{n}}{2}\widetilde{B}_{n}\right)$  \\
     \end{tabular}
 & 
 \begin{tabular}{c}
  $\displaystyle\dot{\Theta}_{n}=4\beta_{1}\left(\eta_{1}\cos\lambda+\eta_{2}\sin\lambda\right)$\\
  $\displaystyle\times\tan\left(\frac{\Theta_{n+1}-\Theta_{n-1}}{4}\right)$
  \end{tabular}
   \\ \hline
   \multicolumn{3}{|l|}{\begin{tabular}{c}$\downarrow$ Specialize the deformation equation by choosing \\
   rigid transformation parameters $C$ and $\Gamma$ as in Theorem \ref{thm:CGamma}.\end{tabular}} \\
    \hline
    $C_{+}, \Gamma_{+}$ & 
   $\displaystyle\left\langle\dot{\gamma}_{n}, \widetilde{B}_{n}\right\rangle=0,\quad\langle\dot{\gamma}_{n}, \dot{\gamma}_{n}\rangle=\rho^{2}_{+}$
     & 
     \begin{tabular}{c}
     $\displaystyle\frac{\dot{\Theta}_{n+1}+\dot{\Theta}_{n}}{2}=\left(1+\cos\lambda\right)\rho_{+}$\\
     $\displaystyle\times\sin\left(\frac{\Theta_{n+1}-\Theta_{n}}{2}\right)$
     \end{tabular}
      \\ \hline 
     \multicolumn{3}{|l|}{\begin{tabular}{c}
     $\updownarrow$ As in Remark \ref{remark:mKdVandsG}, these equations exchange with \\
     the transformation of parameters $v, r$ and a variable $n$ of the solution.
      \end{tabular}
     }  \\ \hline
    $C_{-}, \Gamma_{-}$ & 
    $\displaystyle\left\langle\dot{\gamma}_{n}, \widetilde{B}_{n}\right\rangle=0,\quad\langle\dot{\gamma}_{n}, \dot{\gamma}_{n}\rangle=\rho^{2}_{-}$ &
     \begin{tabular}{c}
   $\displaystyle\frac{\dot{\Theta}_{n+1}-\dot{\Theta}_{n}}{2}=\left(1-\cos\lambda\right)\rho_{-}$\\
   $\displaystyle\times\sin\left(\frac{\Theta_{n+1}+\Theta_{n}}{2}\right)$
   \end{tabular}
    \\ \hline
  \end{tabular}
\end{table}

\begin{example}
The classical 6-Kaleidocycle (Figure \ref{fig:K6} Left) with $\lambda=\pi/2$ 
(e.g., see \cite{FG:mobility}) is obtained by $m=3$ and $(v,r,y) \approx (0.6400, -0.2500, 1.279)$.
\end{example}

\begin{remark}
By using \eqref{cl5} and \eqref{cl}, we see that $\widetilde{B}_{0}=(-1)^{m}\widetilde{B}_{k}$.
That is, the parity of $m$ corresponds to the orientability of the Kaleidocycle.
\end{remark}

\section{Examples}\label{sec:numerical}

In this section, we present examples of Kaleidocycles 
constructed by Theorem \ref{thm:kaleidocycle}
and the semi-discrete K-surfaces obtained as their motion trajectory as explained in Remark \ref{rem:k-surface}.

Numerical experiments using our computer codes~\cite{code} suggest that
the $8$-, $9$- and $15$-Kaleidocycles in Figs. \ref{fig:K8},
\ref{fig:K9},
and \ref{fig:K15} have the minimum torsion angles
among non-orientable $8$-, $9$-, and $15$-Kaleidocycles,
respectively~\cite{seimitsu}. In other words, they are 
the M\"obius Kaleidocycles~\cite{KSFG:Kaleidocycle} (see \Cref{rem:dof}).

For Fig. \ref{fig:K15_m5}, 
the torsion angle is larger than the above minimum one.
However, it has a self-linking number $7.5$.
As the self-linking number of a Kaleidocycle is a half-integer,
these $15$-Kaleidocycles depicted in 
Figs. \ref{fig:K15} and \ref{fig:K15_m5}
belong to different connected components of the configuration space of $15$-Kaleidocycles.
We numerically find that the latter attains the local minima of the torsion angle.
We present the following conjecture based on our extensive numerical experiments:
\begin{conj}\label{conj:mobius}
 Any Kaleidocycle obtained by Theorem \ref{thm:kaleidocycle}
 has a minimal or maximal torsion angle.
 In other words, the projection from the space defined by 
\eqref{eq:configuration_binormal}
onto the $\lambda$-coordinate 
  attains a critical value.
\end{conj}

As briefly discussed in \S \ref{sec:intro} and \Cref{rem:dof},
these Kaleidocycles with critical torsion angles
numerically exhibit special properties, including the single degree of freedom~\cite{seimitsu,KSFG:Kaleidocycle,JF:Kaleidocycle}. Our construction produces this interesting family of Kaleidocycles and may shed light on the rigorous proof for some of their unique properties.

\begin{figure}[ht]
    \centering
    \includegraphics[width=0.25\textwidth]{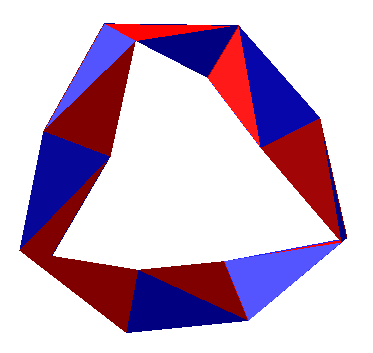}
    \includegraphics[width=0.25\textwidth]{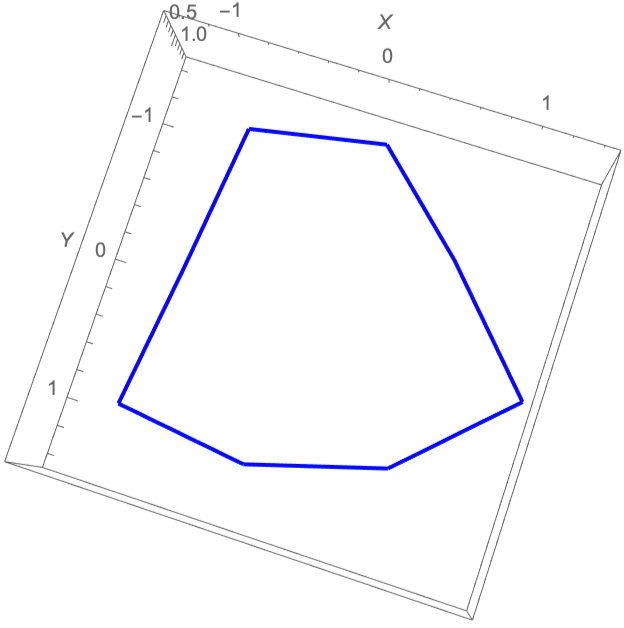}
    \includegraphics[width=0.20\textwidth]{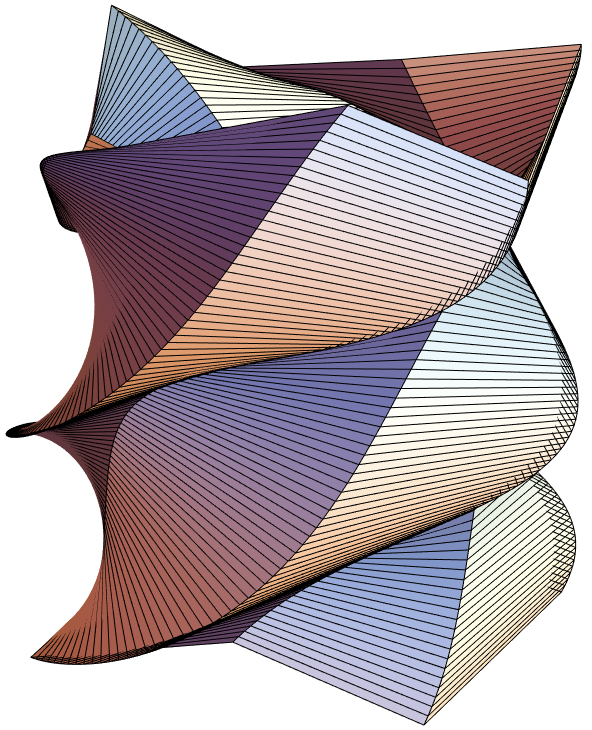}
    \caption{(Left) 
    An example of 8-Kaleidocycle with $(v,r,y) \approx (0.4000, -0.2834, 1.067)$ having $m=3$ and $\cos\lambda\approx 0.4700$.
    (Middle) 8-Kaleidocycle as a discrete curve.
    (Right) the semi-discrete K-surface obtained as its trajectory.}
    \label{fig:K8}
\end{figure}

\begin{figure}[ht]
    \centering
    \includegraphics[width=0.25\textwidth]{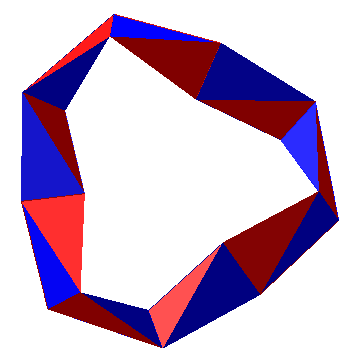}
    \includegraphics[width=0.25\textwidth]{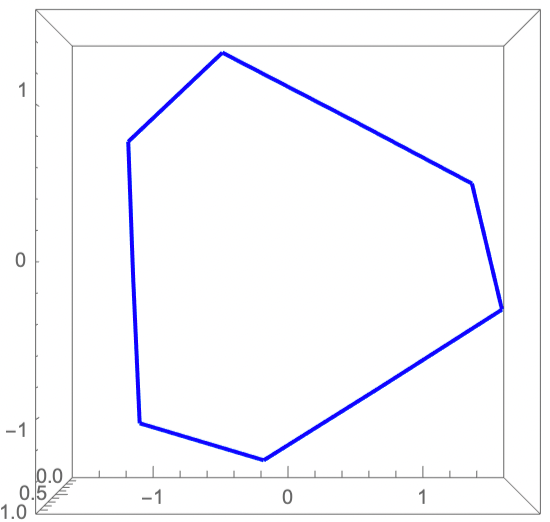}
    \includegraphics[width=0.20\textwidth]{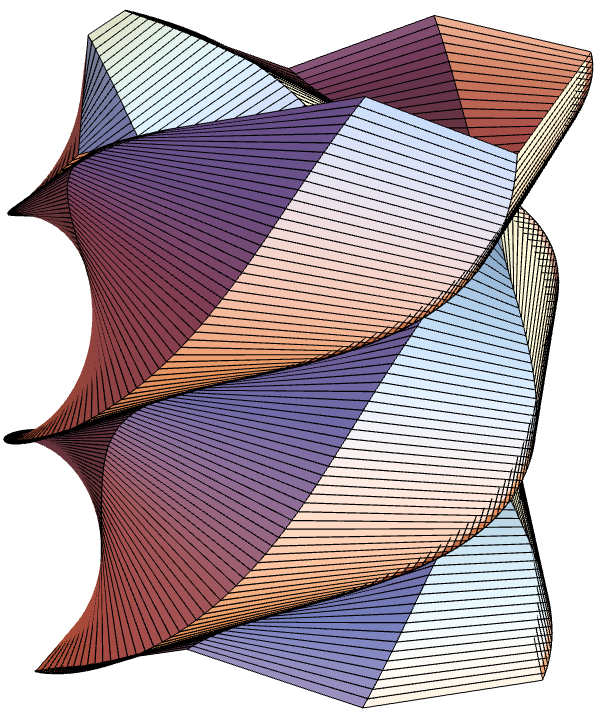}
    \caption{(Left) 
    An example of 9-Kaleidocycle with $(v,r,y) \approx (0.3422, -0.2890, 1.027)$ having $m=3$ and
    $\cos\lambda\approx 0.5852$.
    (Middle) 9-Kaleidocycle as a discrete curve.
    (Right) the semi-discrete K-surface obtained as its trajectory.}
    \label{fig:K9}
\end{figure}

\begin{figure}[ht]
    \centering
    \includegraphics[width=0.25\textwidth]{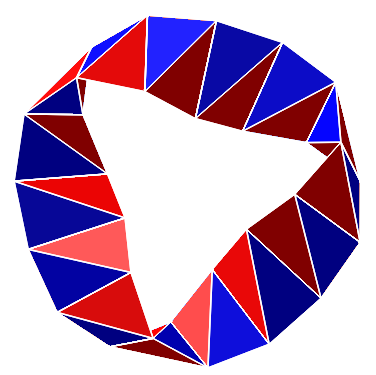}
    \includegraphics[width=0.25\textwidth]{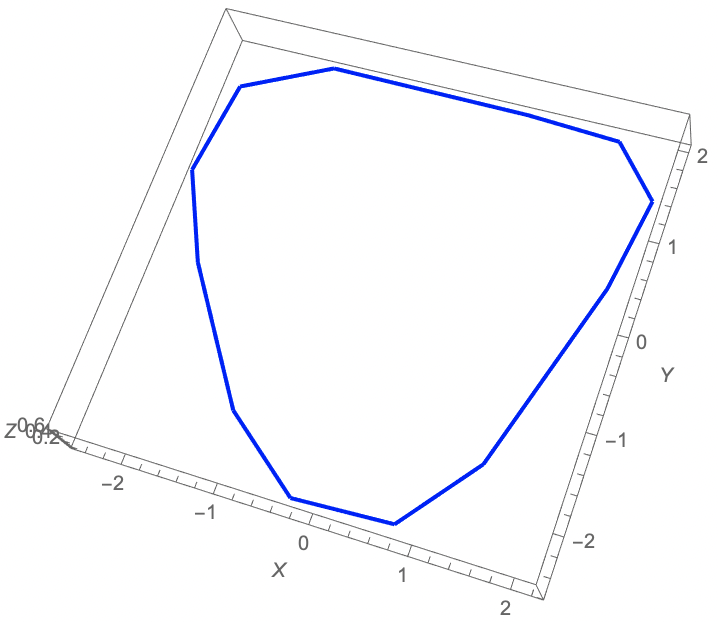}
    \includegraphics[width=0.20\textwidth]{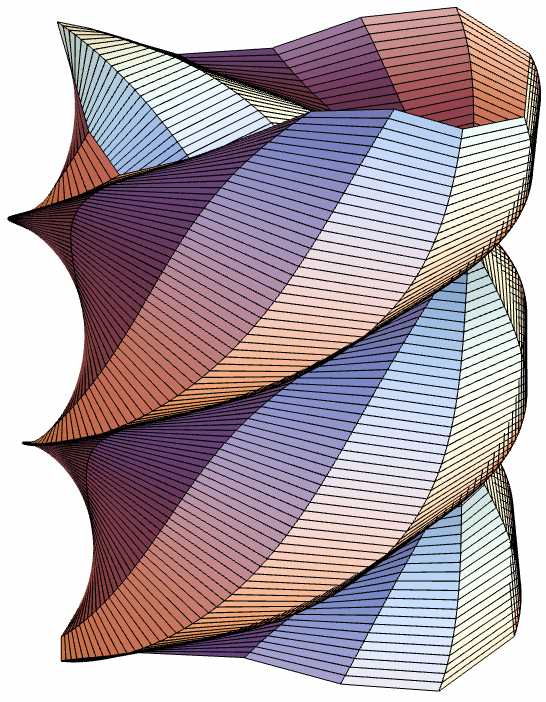}
    \caption{(Left) 
    An example of 15-Kaleidocycle with $(v,r,y) \approx (0.1894, -0.2996, 0.9470)$ having $m=3$ and 
    $\cos\lambda\approx 0.8533$.
    (Middle) 15-Kaleidocycle as a discrete curve.
    (Right) the semi-discrete K-surface obtained as its trajectory.}
    \label{fig:K15}
\end{figure}

\begin{figure}[ht]
    \centering
    \includegraphics[width=0.25\textwidth]{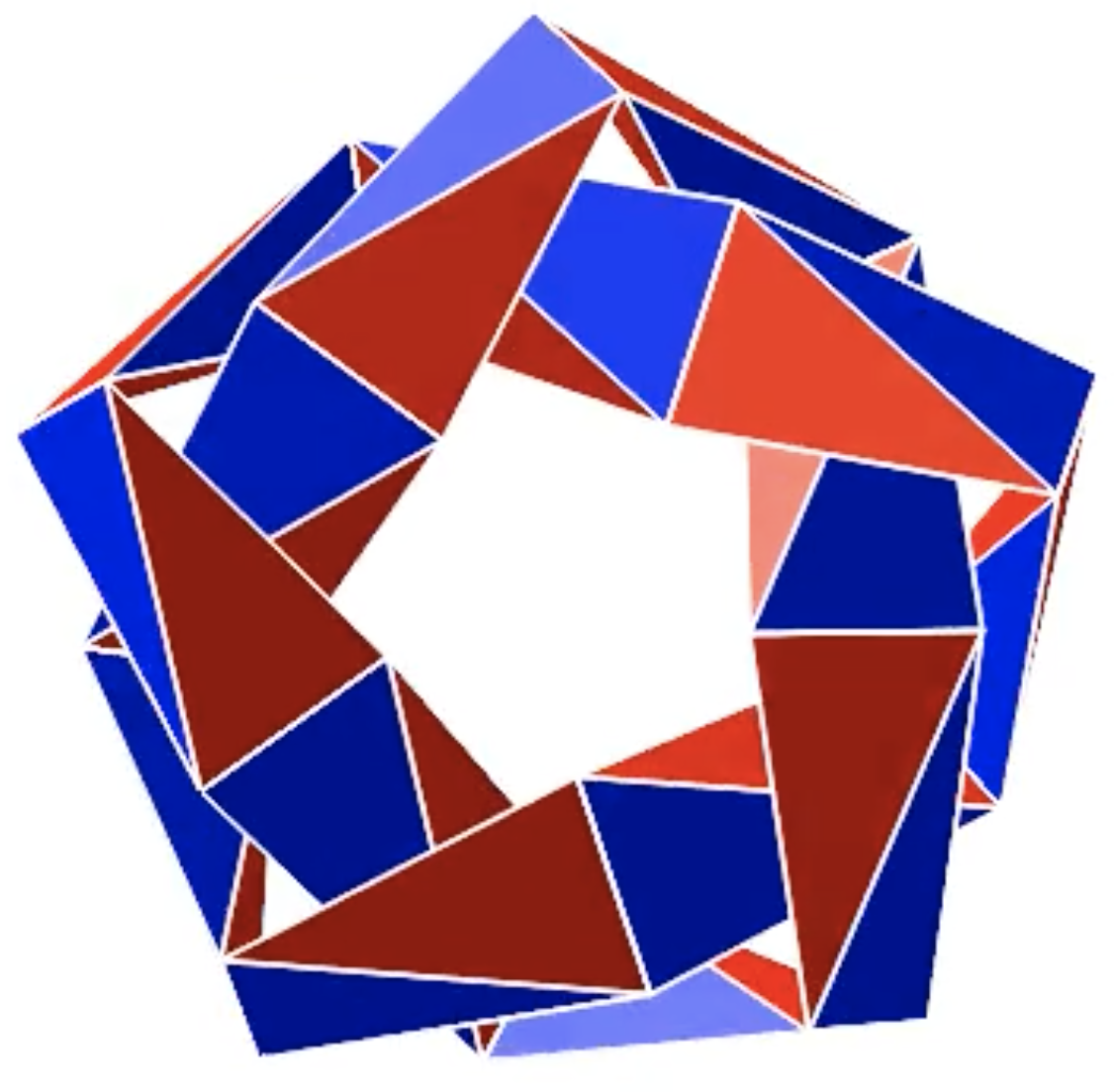}
    \includegraphics[width=0.25\textwidth]{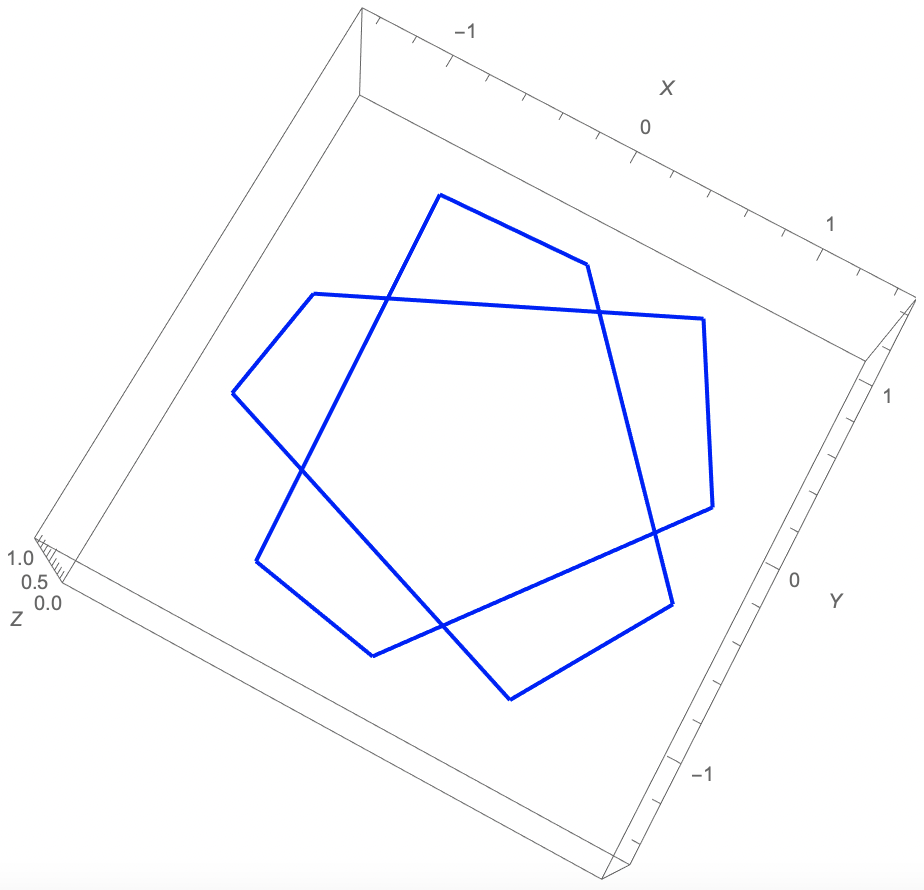}
    \includegraphics[width=0.20\textwidth]{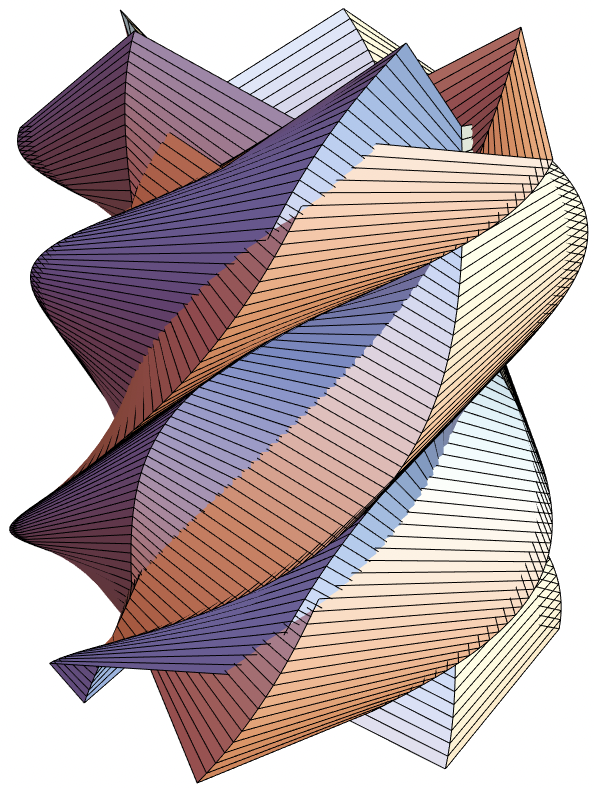}
    \caption{(Left) 
    An example of 15-Kaleidocycle with $(v,r,y) \approx (0.4012, -0.3113, 1.204)$ having $m=5$ and 
    $\cos\lambda\approx 0.6497$.
    (Middle) 15-Kaleidocycle as a discrete curve.
    (Right) the semi-discrete K-surface obtained as its trajectory.}
    \label{fig:K15_m5}
\end{figure}

\appendix
\renewcommand{\thesection}{\Alph{section}}
\renewcommand\thesubsection{\thesection.\arabic{subsection}}
\renewcommand\thesubsubsection{\thesubsection.\arabic{subsubsection}}
\setcounter{section}{0}
\section{Functional identities of elliptic theta functions}\label{section:AppendixA}

In this appendix, we collect several properties of elliptic theta functions on which our calculations rely heavily.

\subsection{Four term identities}
Let
\begin{equation}
\begin{split}
&X_1=\frac{1}{2}(X+Y+U+V),\quad Y_1=\frac{1}{2}(X+Y-U-V),\\
&U_1=\frac{1}{2}(X-Y+U-V),\quad V_1=\frac{1}{2}(X-Y-U+V),
\end{split}
\end{equation}
then for any $X, Y, U, V\in\mathbb{C}$, elliptic theta functions satisfy the following identities \cite{Mumford:Tata_I,Khachev_Zabrodin:theta}.
\begin{align}
&\vartheta_2(X_1)\vartheta_2(Y_1)\vartheta_3(U_1)\vartheta_3(V_1)-\vartheta_4(X_1)\vartheta_4(Y_1)\vartheta_1(U_1)\vartheta_1(V_1)\label{pr4-1}\\
&=\vartheta_2(X)\vartheta_2(Y)\vartheta_3(U)\vartheta_3(V)-\vartheta_4(X)\vartheta_4(Y)\vartheta_1(U)\vartheta_1(V),\nonumber\\
&\vartheta_4(X_1)\vartheta_3(Y_1)\vartheta_1(U_1)\vartheta_2(V_1)-\vartheta_2(X_1)\vartheta_1(Y_1)\vartheta_3(U_1)\vartheta_4(V_1)\label{pr4-2}\\
&=\vartheta_4(X)\vartheta_3(Y)\vartheta_1(U)\vartheta_2(V)-\vartheta_2(X)\vartheta_1(Y)\vartheta_3(U)\vartheta_4(V),\nonumber\\
&\vartheta_2(X_1)\vartheta_2(Y_1)\vartheta_4(U_1)\vartheta_4(V_1)+\vartheta_1(X_1)\vartheta_1(Y_1)\vartheta_3(U_1)\vartheta_3(V_1)\label{pr4-3}\\
&=\vartheta_4(X)\vartheta_4(Y)\vartheta_2(U)\vartheta_2(V)-\vartheta_3(X)\vartheta_3(Y)\vartheta_1(U)\vartheta_1(V),\nonumber\\
&\vartheta_3(X_1)\vartheta_3(Y_1)\vartheta_2(U_1)\vartheta_2(V_1)+\vartheta_1(X_1)\vartheta_1(Y_1)\vartheta_4(U_1)\vartheta_4(V_1)\label{pr4-3-1}\\
&=\vartheta_3(X)\vartheta_3(Y)\vartheta_2(U)\vartheta_2(V)+\vartheta_1(X)\vartheta_1(Y)\vartheta_4(U)\vartheta_4(V),\nonumber\\
&\vartheta_4(X_1)\vartheta_4(Y_1)\vartheta_4(U_1)\vartheta_4(V_1)-\vartheta_1(X_1)\vartheta_1(Y_1)\vartheta_1(U_1)\vartheta_1(V_1)\label{pr4-3-2}\\
&=\vartheta_4(X)\vartheta_4(Y)\vartheta_4(U)\vartheta_4(V)-\vartheta_1(X)\vartheta_1(Y)\vartheta_1(U)\vartheta_1(V),\nonumber\\
&\vartheta_4(X_1)\vartheta_4(Y_1)\vartheta_1(U_1)\vartheta_1(V_1)-\vartheta_1(X_1)\vartheta_1(Y_1)\vartheta_4(U_1)\vartheta_4(V_1)\label{pr4-3-3}\\
&=\vartheta_4(X)\vartheta_4(Y)\vartheta_1(U)\vartheta_1(V)-\vartheta_1(X)\vartheta_1(Y)\vartheta_4(U)\vartheta_4(V),\nonumber\\
&\vartheta_4(X_1)\vartheta_4(Y_1)\vartheta_1(U_1)\vartheta_1(V_1)+\vartheta_2(X_1)\vartheta_2(Y_1)\vartheta_3(U_1)\vartheta_3(V_1)\label{pr4-3-4}\\
&=\vartheta_3(X)\vartheta_3(Y)\vartheta_2(U)\vartheta_2(V)-\vartheta_1(X)\vartheta_1(Y)\vartheta_4(U)\vartheta_4(V),\nonumber\\
&\vartheta_3(X_1)\vartheta_3(Y_1)\vartheta_1(U_1)\vartheta_1(V_1)-\vartheta_1(X_1)\vartheta_1(Y_1)\vartheta_3(U_1)\vartheta_3(V_1)\label{pr4-3-6}\\
&=\vartheta_3(X)\vartheta_3(Y)\vartheta_1(U)\vartheta_1(V)-\vartheta_1(X)\vartheta_1(Y)\vartheta_3(U)\vartheta_3(V),\nonumber\\
&\vartheta_3(X_1)\vartheta_3(Y_1)\vartheta_4(U_1)\vartheta_4(V_1)-\vartheta_1(X_1)\vartheta_1(Y_1)\vartheta_2(U_1)\vartheta_2(V_1)\label{pr4-3-7}\\
&=\vartheta_3(X)\vartheta_3(Y)\vartheta_4(U)\vartheta_4(V)-\vartheta_1(X)\vartheta_1(Y)\vartheta_2(U)\vartheta_2(V).\nonumber
\end{align}
\subsection{Shifts by periods and half-periods}
For any $X\in\mathbb{C}$, $n\in\mathbb{Z}$ and $y>0$, the following holds \cite{Mumford:Tata_I,Khachev_Zabrodin:theta}.
\begin{equation}\label{peri1}
\vartheta_1\left(X+n\right)=(-1)^{n}\vartheta_1\left(X\right),\quad\vartheta_4\left(X+n\right)=\vartheta_4\left(X\right).
\end{equation}
\begin{equation}\label{cl}
\begin{split}
&\vartheta_1\left(X+niy; iy\right)=(-1)^{n}\exp\{-\pi i\left(2nX+n^{2}iy\right)\}\vartheta_1\left(X;iy\right),\\
&\vartheta_2\left(X+niy;iy\right)=\exp\{-\pi i\left(2nX+n^{2}iy\right)\}\vartheta_2\left(X;iy\right),\\
&\vartheta_3\left(X+niy;iy\right)=\exp\{-\pi i\left(2nX+n^{2}iy\right)\}\vartheta_3\left(X;iy\right),\\
&\vartheta_4\left(X+niy;iy\right)=(-1)^{n}\exp\{-\pi i\left(2nX+n^{2}iy\right)\}\vartheta_4\left(X;iy\right).
\end{split}
\end{equation}
\begin{equation}\label{clclcl}
\begin{split}
&\vartheta_1\left(X+\frac{1}{2}\right)=\vartheta_2\left(X\right),\quad \vartheta_2\left(X+\frac{1}{2}\right)=-\vartheta_1\left(X\right),\\
&\vartheta_3\left(X+\frac{1}{2}\right)=\vartheta_4\left(X\right),\quad\vartheta_4\left(X+\frac{1}{2}\right)=\vartheta_3\left(X\right).
\end{split}
\end{equation}
\begin{equation}\label{clclcl2}
\begin{split}
&\vartheta_1\left(X+\frac{iy}{2};iy\right)=i\exp\left\{-\pi i\left(X+\frac{iy}{4}\right)\right\}\vartheta_4\left(X;iy\right),\\
&\vartheta_2\left(X+\frac{iy}{2};iy\right)=\exp\left\{-\pi i\left(X+\frac{iy}{4}\right)\right\}\vartheta_3\left(X;iy\right),\\
&\vartheta_3\left(X+\frac{iy}{2};iy\right)=\exp\left\{-\pi i\left(X+\frac{iy}{4}\right)\right\}\vartheta_2\left(X;iy\right),\\
&\vartheta_4\left(X+\frac{iy}{2};iy\right)=i\exp\left\{-\pi i\left(X+\frac{iy}{4}\right)\right\}\vartheta_1\left(X;iy\right).
\end{split}
\end{equation}
\subsection{Identity for theta constants}
For any $y>0$, elliptic theta functions satisfy the following identity \cite{Mumford:Tata_I,Khachev_Zabrodin:theta}.
\begin{equation}\label{id:const}
\vartheta_1^{'} \left( 0; iy \right)= \pi \vartheta_2 \left( 0; iy \right) \vartheta_3 \left( 0; iy \right) \vartheta_4 \left( 0; iy \right).
\end{equation}
\subsection{Modular transformations}
For any $X\in\mathbb{C}$ and $y>0$, elliptic theta functions satisfy the following identities \cite{Mumford:Tata_I,Khachev_Zabrodin:theta}.
 \begin{equation}\label{id:modular}
 \begin{split}
 &\vartheta_{1}\left(\frac{X}{y}; i\frac{2}{y}\right)=-i\sqrt{\frac{y}{2}}\exp\left(-\pi\frac{X^{2}}{2y}\right)\vartheta_{1}\left(\frac{iX}{2}; i\frac{y}{2}\right),\\
 & \vartheta_{2}\left(\frac{X}{y}; i\frac{2}{y}\right)=\sqrt{\frac{y}{2}}\exp\left(-\pi\frac{X^{2}}{2y}\right)\vartheta_{4}\left(\frac{iX}{2}; i\frac{y}{2}\right),\\
  &\vartheta_{3}\left(\frac{X}{y}; i\frac{2}{y}\right)=\sqrt{\frac{y}{2}}\exp\left(-\pi\frac{X^{2}}{2y}\right)\vartheta_{3}\left(\frac{iX}{2}; i\frac{y}{2}\right),\\
 &\vartheta_{4}\left(\frac{X}{y}; i\frac{2}{y}\right)=\sqrt{\frac{y}{2}}\exp\left(-\pi\frac{X^{2}}{2y}\right)\vartheta_{2}\left(\frac{iX}{2}; i\frac{y}{2}\right).
 \end{split}
 \end{equation}
 \begin{equation}
\frac{\vartheta_{1}\left(X; iy\right)}{\vartheta_{2}\left(X; iy\right)}=-i\frac{\vartheta_{1}\left(\frac{iX}{y}; \frac{i}{y}\right)}{\vartheta_{4}\left(\frac{iX}{y}; \frac{i}{y}\right)}.\label{kappa-4}
\end{equation}
\subsection{Bilinear identities connecting theta functions with \texorpdfstring{$y$}{y} and \texorpdfstring{$y/2$}{y/2}}
 For any $X, Y\in\mathbb{C}$, and $y>0$, elliptic theta functions satisfy the following identities \cite{Mumford:Tata_I,Khachev_Zabrodin:theta}.
\begin{align}
&2\vartheta_{1}\left(X+Y; iy\right)\vartheta_{4}\left(X-Y; iy\right)\nonumber\\
&=\vartheta_{1}\left(X; \frac{iy}{2}\right)\vartheta_{2}\left(Y; \frac{iy}{2}\right)+\vartheta_{2}\left(X; \frac{iy}{2}\right)\vartheta_{1}\left(Y; \frac{iy}{2}\right),\label{kappa-5}\\
&2\vartheta_{2}\left(X+Y; iy\right)\vartheta_{3}\left(X-Y; iy\right)\nonumber\\
&=\vartheta_{2}\left(X; \frac{iy}{2}\right)\vartheta_{2}\left(Y; \frac{iy}{2}\right)-\vartheta_{1}\left(X; \frac{iy}{2}\right)\vartheta_{1}\left(Y; \frac{iy}{2}\right).\label{kappa-6}
\end{align}
  \begin{equation}\label{id:modular2}
 \begin{split}
& \vartheta_{3}\left(\frac{X}{2}; \frac{iy}{2}\right)\vartheta_{4}\left(\frac{X}{2}; \frac{iy}{2}\right)=\vartheta_{4}\left(X; iy\right)\vartheta_{4}\left(0; iy\right),\\
& \vartheta_{1}\left(\frac{X}{2}; \frac{iy}{2}\right)\vartheta_{2}\left(\frac{X}{2}; \frac{iy}{2}\right)=\vartheta_{1}\left(X; iy\right)\vartheta_{4}\left(0; iy\right),\\
 -&\vartheta_{1}\left(\frac{X}{2}; \frac{iy}{2}\right)^{2}+\vartheta_{2}\left(\frac{X}{2}; \frac{iy}{2}\right)^{2}=2\vartheta_{2}\left(X; iy\right)\vartheta_{3}\left(0; iy\right).\\
  \end{split}
 \end{equation}

\noindent
\textbf{Conflicts of Interest}\\
The authors declare no conflicts of interest.\\
\textbf{Data Availability Statement}\\
The computer program used in the experiments is available in the associated GitHub repository \url{https://github.com/shizuo-kaji/Kaleidocycle/}.

\bibliographystyle{abbrv}
\bibliography{references}                                   
\end{document}